\documentclass[11pt]{article}

\usepackage{latexsym}
\usepackage{amsmath}
\usepackage{amssymb}
\usepackage{amsthm}
\usepackage{hyperref}
\usepackage{cite}
\usepackage{graphicx}
\usepackage{color}
\usepackage{enumitem}
\usepackage{bm}
\usepackage[left=1in,right=1in,top=1in,bottom=1in]{geometry}
\usepackage{xspace}
\setlist[description]{font=\normalfont\itshape\textbullet\space}
\usepackage[all]{xy}
\usepackage{doi, url}
\usepackage{arydshln}

\renewcommand{\paragraph}[1]{\vspace{6pt} \noindent \textbf{#1}\xspace}

\theoremstyle{plain}
\newtheorem{theorem}{Theorem}[section]

\newtheorem{corollary}[theorem]{Corollary}
\newtheorem{lemma}[theorem]{Lemma}
\newtheorem{observation}[theorem]{Observation}
\newtheorem{proposition}[theorem]{Proposition}

\theoremstyle{definition}

\newtheorem{remark}[theorem]{Remark}

\newtheorem{definition}[theorem]{Definition}
\newtheorem{example}[theorem]{Example}

\newcommand{\varpc}{\mathfrak{B}}

\newcommand{\GL}{\mathrm{GL}}
\newcommand{\F}{\mathbb{F}}
\newcommand{\Z}{\mathbb{Z}}

\newcommand{\N}{\mathbb{N}}

\newcommand{\rk}{\mathrm{rk}}

\newcommand{\IdMat}{\mathrm{I}}

\newcommand{\sgn}{\mathrm{sgn}}

\newcommand{\poly}{\mathrm{poly}}

\newcommand{\M}{\mathrm{M}}
\newcommand{\T}{\mathrm{T}}
\renewcommand{\S}{\mathrm{S}}
\newcommand{\Mon}{\mathrm{Mon}}

\newcommand{\tuple}[1]{\mathbf{#1}}
\newcommand{\tens}[1]{\mathtt{#1}}
\newcommand{\spa}[1]{\mathcal{#1}}

\newcommand{\cA}{\spa{A}}
\newcommand{\cB}{\spa{B}}

\newcommand{\tA}{\tens{A}}
\newcommand{\tB}{\tens{B}}

\newcommand{\vA}{\tuple{A}}
\newcommand{\vB}{\tuple{B}}

\newcommand{\aut}{\mathrm{Aut}}

\newcommand{\algprobm}[1]{{\rm #1}\xspace}
\newcommand{\AlgIsolong}{\algprobm{Algebra Isomorphism}}
\newcommand{\AlgIso}{\algprobm{AlgIso}}
\newcommand{\PolyEqlong}{\algprobm{Polynomial Equivalence}}

\newcommand{\GI}{\algprobm{GI}}
\renewcommand{\P}{\cc{P}}
\newcommand{\cGI}{\algprobm{color-GI}}
\newcommand{\GIlong}{\algprobm{Graph Isomorphism}}
\newcommand{\GpI}{\algprobm{GpI}}
\newcommand{\GpIlong}{\algprobm{Group Isomorphism}}
\newcommand{\CubicFormlong}{\algprobm{Cubic Form Equivalence}}
\newcommand{\CubicForm}{\algprobm{CFE}}

\newcommand{\NcCubicFormlong}{\algprobm{Trilinear} \algprobm{Form} \algprobm{Equivalence}} 
\newcommand{\AltMatSpIsomlong}{\algprobm{Alternating Matrix Space Isometry}}
\newcommand{\AMSI}{\algprobm{AMSI}}

\newcommand{\STFElong}{\algprobm{Symmetric Trilinear Form Equivalence}}
\newcommand{\STFE}{\algprobm{STFE}}

\newcommand{\MatSpIsomlong}{\algprobm{Matrix Space Isometry}}
\newcommand{\MatSpConjlong}{\algprobm{Matrix Space Conjugacy}}

\newcommand{\TI}{\algprobm{TI}}
\newcommand{\TIlong}{\algprobm{Tensor Isomorphism}}
\newcommand{\pTI}{\algprobm{part-TI}}

\newcommand{\MatLieConjlong}{\algprobm{Matrix Lie Algebra Conjugacy}}

\newcommand{\cc}[1]{\mathsf{#1}}

\newcommand{\ynote}[1]{}
 
 \newcommand{\jnote}[1]{}

\DeclareMathOperator{\Iso}{Iso}

\DeclareMathOperator{\rank}{rank}
\DeclareMathOperator{\diag}{diag}

\newcommand{\too}%
{\xrightarrow{\text{\raisebox{-3pt}{$\sim$}}\,}}

\newcommand{\logbf}{\mathbf{log}}
\newcommand{\expbf}{\mathbf{exp}}

\usepackage[T1]{fontenc}
\def\DJ{{\hbox{D\kern-.8em\raise.15ex\hbox{--}\kern.35em}}}

\title{
	On the complexity of isomorphism problems for tensors, groups, and polynomials IV: linear-length reductions and their applications
	}
\author{
	Joshua A. Grochow
	\footnote{Departments of Computer Science and Mathematics, University of Colorado, 
		Boulder. \tt{jgrochow@colorado.edu}}
	\and
	Youming Qiao
	\footnote{Centre for Quantum Software and Information, University of 
		Technology Sydney. \tt{youming.qiao@uts.edu.au}}
}
\date{}
\begin{document}
\maketitle


\begin{abstract}
Many isomorphism problems for tensors, groups, algebras, and polynomials were recently shown to be equivalent to one another under polynomial-time reductions, prompting the introduction of the complexity class \TI (Grochow \& Qiao, \emph{SIAM J. Comp.} '23 \& \emph{ITCS} '21). 
Using the tensorial viewpoint, Grochow \& Qiao (\emph{ACM Trans. Comput. Theory}, '24 \& \emph{CCC} '21) gave moderately exponential-time search- and counting-to-decision reductions for some class of $p$-groups. A significant issue was that the reductions usually incurred a quadratic increase in the length of the tensors involved. When the tensors represent $p$-groups, this corresponds to an increase in the order of the group of the form $|G|^{\Theta(\log |G|)}$, negating any gains in the Cayley table model.

In this paper, we present a new kind of tensor gadget that allows us to replace those quadratic-length reductions with linear-length ones, yielding the following consequences:
\begin{enumerate}
	\item If \GIlong is in $\cc{P}$, then testing equivalence of cubic forms in $n$ variables over $\F_q$, and testing isomorphism of $n$-dimensional algebras over $\F_q$, can both be solved in time $q^{O(n)}$, improving from the brute-force upper bound $q^{O(n^2)}$ for both of these.
	\item 
	Combined with the $|G|^{O((\log |G|)^{5/6})}$-time isomorphism-test for $p$-groups of class 2 and exponent $p$ (Sun, \emph{STOC} '23), 
	our reductions extend this runtime to $p$-groups of class $c$ and exponent $p$ where $c<p$, and yield algorithms in time 
	$q^{O(n^{1.8}\cdot \log q)}$ 
	for cubic form equivalence and algebra isomorphism.
	\item Polynomial-time search- and counting-to-decision reduction for testing isomorphism of $p$-groups of class $2$ and exponent $p$ when Cayley tables are given. This answers questions of Arvind and T\'oran (\emph{Bull. EATCS}, 2005) for this group class, thought to be one of the hardest cases of Group Isomorphism.
\end{enumerate}

Our reductions are presented in a more modular and composable fashion compared to previous gadgets, making them easier to reason about and, crucially, easier to combine. The results are also used by Ivanyos--Mendoza--Qiao--Sun--Zhang to simplify the algorithm in (Sun, \emph{STOC} '23) and extend to more general $p$-groups.
\end{abstract}

\section{Introduction}

\paragraph{Background: \GIlong and \TIlong.} Given two combinatorial or algebraic structures, the isomorphism problem asks whether they are essentially the same. The most well-known such problem is \GIlong (\GI for short), which asks to decide whether two graph are isomorphic. \GI has received considerable attention since the birth of computational complexity (see \cite[Sec.~1]{AllenderDas}). As many isomorphism problems of combinatorial structures reduce to \GI in polynomial time, the complexity class $\cc{GI}$ was introduced, consisting of problems polynomial-time reducible to \GI \cite{KST93}. Babai's quasipolynomial-time algorithm for \GI \cite{Bab16} is widely regarded as a breakthrough in theoretical computer science. 

The study of isomorphism problems of algebraic structures, such as groups, algebras, and polynomials, has a long tradition. These problems appears naturally in theoretical computer science \cite{AT05,KS06,AS06,GQ17}, cryptography \cite{Pat96}, computer algebra \cite{CH03}, quantum information \cite{BPR+00}, and machine learning \cite{PSS18}. 
Partly motivated by developing a complexity class to capture those problems, the complexity class $\cc{TI}$ was recently introduced in \cite{GQ1}, consisting of problems polynomial-time reducible to \TIlong (\TI for short). We refer the interested readers to Section~\ref{subsec:connection} for a precise definition of \TI, and for how \TI relates to isomorphism problems of algebras, polynomials, and groups. 

Despite being only recently developed, the complexity theory of $\cc{TI}$ has shown to be useful in unifying isomorphism problems for algebraic structures \cite{GQ1,GQT}, with applications to quantum information \cite{GQ1,TI3} and cryptography \cite{JQSY19,TangDJPQS22}.
As \GI reduces to \TI in polynomial time \cite{GQ1}, $\cc{GI}\subseteq\cc{TI}$. 

\paragraph{Organization of the introduction.} In this paper, we further the study of the complexity theory of $\cc{TI}$. As a result, we are able to address some classical complexity-theoretic questions about \GIlong, \GpIlong, and \PolyEqlong. In the remainder of this introduction, we first present these classical questions (Section~\ref{subsec:3q}) and our results on these questions (Section~\ref{subsec:progress_3q}). Then we explain how $\cc{TI}$ relates to these questions (Section~\ref{subsec:connection}), and present our first main result on $\cc{TI}$ (Section~\ref{subsec:results}). In Section~\ref{subsec:background}, we discuss on the techniques. In Section~\ref{subsec:results-matrix}, we present our second main result on $\cc{TI}$.

\paragraph{Some notation.} We use $\F$ to denote a field, and $\F_q$ the finite field of order $q$. We use $\M(n, \F)$ to denote the linear space of $n\times n$ matrices over $\F$, and $\GL(n, \F)$ the general linear group consisting of $n\times n$ invertible matrices over $\F$.

\paragraph{Definitions of \GpIlong and \CubicFormlong.} We now define \GpIlong and \PolyEqlong and briefly introduce their status.

\GpIlong (\GpI for short) asks to decide whether two finite groups are isomorphic. Here groups can be given verbosely (by their Cayley tables) or succinctly (by generators of permutation groups or matrix groups). \GpI has received considerable attention since 1970's, in both theoretical computer science \cite{Mil78,BCGQ11,GQ17,LQ17,DW21,IQ19,Sun23} and computational group theory \cite{FN70,Obr94,CH03,Wil09a,BW12,BMW}. For groups of order $n$, \GpI admits an $n^{\log n+O(1)}$-time algorithm \cite{Mil78,FN70}\footnote{Miller attributes this algorithm to Tarjan.}, and the current best algorithm for general \GpI runs in time $n^{1/4\cdot \log n+o(\log n)}$  \cite{Ros13,GR16}. 

For \PolyEqlong, we focus on cubic forms (that is, homogeneous cubic polynomials). The \CubicFormlong (\CubicForm for short) problem asks whether two cubic forms $f, g\in\F[x_1, \dots, x_n]$, over a field $\F$, can be transformed to one another by an invertible linear change of variables.\footnote{In symbols, $f$ and $g$ are given by lists of coefficients, and \CubicForm asks whether there exists $A=(a_{i,j})\in\GL(n, \F)$, such that $f=g\circ A$, where $(g\circ A)(x_1, \dots, x_n):=g(\sum_{i\in[n]}a_{1, i}x_i, \dots, \sum_{i\in[n]}a_{n, i}x_i)$.} \CubicForm has been studied in cryptography \cite{Pat96,Bou11} and complexity theory \cite{AS06,Saxena_thesis,Kay11,Kay12,GQT}. The worst-case complexity for \CubicForm over $\F_q$ had not been improved from the brute-force $q^{O(n^2)}$ time, and the first improvement comes as a consequence of this paper (Corollary~\ref{cor:algiso}).

\subsection{Three classical questions about isomorphism problems}\label{subsec:3q}

The following questions about \GpI, \CubicForm, and \GI have been either proposed explicitly, or known to experts for some time. 

\begin{description}
	\item[{\bf Question 1.}] (Roadblocks for putting \GI in $\P$.) After Babai's quasipolynomial-time algorithm for \GI \cite{Bab16}, it is natural to wonder about the grand goal of putting \GI in $\P$. Babai raised \GpI as one important road block on the way to putting \GI in $\P$ \cite{Bab16}. This is because (1) $\GI\in \P$ implies a $\poly(n)$-time algorithm to test isomorphism of groups of order $n$, and 
(2) no $n^{o(\log n)}$-time algorithm is known for \GpI with groups of order $n$. 
	
	Following this train of thoughts, it is desirable to identify more isomorphism problems, whose faster algorithms would be a consequence of $\GI\in\P$, therefore standing as bottlenecks for putting $\GI\in\P$. For example, would $\GI\in\P$ imply a $q^{O(n)}$-time algorithm for \CubicForm over $\F_q$ in $n$ variables? If so, then \CubicForm could be regarded as a bottleneck for $\GI\in\P$, as no $q^{O(n)}$-time for \CubicForm is known.
	
	
	\item[{\bf Question 2.}] (Reducing the general \GpI to \GpI for special group classes.) For \GpI, it is generally regarded that $p$-groups of class $2$ and exponent $p$ form a difficult class\footnote{A $p$-group of class $2$ and exponent $p$ is a group $G$ satisfying: (1) $|G|=p^\ell$, (2) every $g\in G$ is of order $p$, and (3) $G/Z(G)$ is Abelian, where $Z(G)$ is the center of $G$.} for \GpI. For convenience, we use $\varpc(p, 2)$ to denote the class of $p$-groups of class $2$ and exponent $p$. A widely-held belief is that progress on \GpI for $\varpc(p, 2)$ would lead to progress on \GpI in general. However, a formal reduction from general \GpI to \GpI for $\varpc(p, 2)$ is not known, so it is desirable to devise reductions for \GpI from certain group classes to $\varpc(p, 2)$. This question is of particular interest now, given the recent breakthrough of Sun \cite{Sun23}, who develops the first $n^{o(\log n)}$-time algorithm to test isomorphism for $\varpc(p, 2)$. 
	
	\item[{\bf Question 3.}] (Search- and counting-to-decision reductions for \GpI.) Isomorphism problems for algebraic structures sometimes demonstrate some peculiar features from the complexity viewpoint. For example, search- and counting-to-decision reductions are classical topics in complexity theory \cite{Val76,BG94}. It is known that \GI admits search- and counting-to-decision reductions \cite{KST93,Mat79}. On the other hand, polynomial-time search- and counting-to-decision reductions for \GpI are not known, raised by Arvind and T\'oran as open questions \cite{AT05}. 
\end{description}

\subsection{Main results I: progress on the three questions}\label{subsec:progress_3q}


The complexity class $\cc{TI}$ and the techniques developed for it in \cite{GQ1,GQ2,GQT} shed light on the three questions in Section~\ref{subsec:3q}. However, they fall short of providing answers for more refined measures, or for the Cayley table model of \GpI. The lack of results for the Cayley table model was in particular annoying, because \GpI in this model has been the focus from the worst-case complexity viewpoint, and stands as a bottleneck in the complexity of \GI.

In the following, for each question, we will present the previous relevant result \cite{GQ1,GQ2}, and then introduce our result. 

\paragraph{On Question 1.} In \cite{GQ1,GQT}, it was shown that \CubicForm, and \GpI for $\varpc(p, 2)$ in the matrix group model, are polynomial-time equivalent for $p>3$. This leads to a reduction from \CubicForm to \GI as follows: first, convert the resulting matrix groups to the Cayley table model, and then use the reduction from \GpI in the Cayley table model to \GI \cite{KST93}, to get a reduction from \CubicForm to \GI. However, the resulting graphs from this construction are of the order $p^{\Omega(n^2)}$, so $\GI\in\P$ does not imply an algorithm faster than the brute-force algorithm for \CubicForm.

In this paper, we are able to show that $\GI\in\P$ has the following consequence for \CubicFormlong.



\begin{theorem}\label{thm:GI}
	Let $\F_p$ be a field of order $p$, $p$ a prime $>3$. If \GIlong is in $\cc{P}$, then \CubicFormlong over $\F_p$ in $n$-variables can be solved in $p^{O(n)}$ time.
\end{theorem}


\paragraph{On Question 2.} To make progress on reducing the general \GpI to \GpI for $\varpc(p, 2)$, a natural step is to consider reducing \GpI for general $p$-groups to \GpI for $\varpc(p, 2)$. 
Recall that $p$-groups are groups of prime power order. $p$-groups are nilpotent, namely the lower (or upper) central series terminates in finite steps. The length of the lower central series of a $p$-group $G$ is called its (nilpotency) class. We use $\varpc(p, c)$ to denote the set of $p$-groups of class $\leq c$ and exponent $p$, and $\varpc(p, c, N)$ to denote the set of $p$-groups of class $c$, exponent $p$, and order $N$. 

In \cite{GQ2}, it was shown that for matrix groups, \GpI for $\varpc(p, c)$, $p>c$, reduces to \GpI for $\varpc(p, 2)$ in polynomial time. This is the first reduction for \GpI from a more general group class to $\varpc(p, 2)$. Unfortunately, the results in \cite{GQ2} do not yield anything for the Cayley table model. 

In this paper, we show that $\GpI$ for $\varpc(p, c)$, $p>c$, reduces to $\GpI$ for $\varpc(p, 2)$, in the Cayley table model. 
\begin{theorem}\label{thm:cto2}
	Given the Cayley tables of two groups $G$ and $H$ from $\varpc(p, c, N)$, $p>c$, there is a polynomial-time algorithm $A$ that outputs $A(G)$ and $A(H)$ in $\varpc(p, 2, \poly(N))$, such that $G$ and $H$ are isomorphic if and only if $A(G)$ and $A(H)$ are isomorphic.  
\end{theorem}

Recently, a breakthrough on \GpI for $\varpc(p, 2, N)$ was achieved by Sun \cite{Sun23}. 
\begin{theorem}[{\cite{Sun23}}]\label{thm:sun}
	Given the Cayley tables of two groups $G$ and $H$ from $\varpc(p, 2, N)$, there is an $N^{\tilde O((\log N)^{5/6})}$-time algorithm testing whether $G$ and $H$ are isomorphic.
\end{theorem}
Combining Theorems~\ref{thm:cto2} and~\ref{thm:sun}, we have the following result for $\varpc(p, c, N)$ with $c<p$. 
\begin{corollary}\label{cor:nil_reduction}
	Given the Cayley tables of two groups $G$ and $H$ from $\varpc(p, c, N)$, $c<p$, there is an $N^{\tilde O((\log N)^{1/2})}$-time algorithm testing whether $G$ and $H$ are isomorphic.
\end{corollary}
Before Corollary~\ref{cor:nil_reduction}, the best algorithm for \GpI for $\varpc(p, c, N)$ runs in time $N^{\frac{1}{4}\log N+o(\log N)}$ \cite{RW15,Ros13}.



\paragraph{On Question 3.} In \cite{GQ2}, moderately exponential-time search- and counting-to-decision reductions were presented for \GpI for $\varpc(p, 2)$ in the matrix group model. However, these results do not carry over to \GpI in the Cayley table model. 

In this paper, we devise polynomial-time search- and counting-to-decision reductions for $\varpc(p, 2)$. For this group class, this answers a question of Arvind and T\'oran in \cite{AT05}.
\begin{theorem}\label{thm:search-counting}
	There are polynomial-time search- and counting-to-decision for isomorphism of groups from $\varpc(p, 2, N)$ when Cayley tables are given. 
\end{theorem}

\subsection{From groups and cubic forms to 3-way arrays}\label{subsec:connection}

In Section~\ref{subsec:3q}, we introduced three classical questions for isomorphism problems, and stated our results for these questions in Section~\ref{subsec:progress_3q}. While these questions are for cubic forms and groups, our results are achieved from the tensor viewpoint, following \cite{GQ1,GQ2}. In the following, we introduce this tensorial viewpoint and show how this viewpoint relates to groups and polynomials. 

\paragraph{3-way arrays and tensors.} A \emph{3-way array} is an array with three indices over a field $\F$, that is, $\tA=(a_{i, j, k})$, where $i\in[n]$, $j\in[m]$, $k\in[\ell]$, and $a_{i, j, k}\in \F$. We use $\T(n\times m\times \ell, \F)$ to denote the linear space of 3-way arrays of size $n\times m\times \ell$ over $\F$. The following parameter is important so we put it in a definition. 
\begin{definition}\label{def:length}
	The \emph{length} of $\tA\in\T(n\times m\times\ell, \F)$ is $n+m+\ell$. 
\end{definition}
The name ``length'' comes from the intuition that an $n\times m\times \ell$ tensor can be viewed as a $n\times m\times \ell$ cuboid of numbers, with the sum of side lengths being $n+m+\ell$.

A natural group action on $3$-way arrays is as follows: $(R, S, T)\in \GL(n, \F)\times\GL(m, \F)\times\GL(\ell, \F)$, where $R=(r_{i,i'})$, $S=(s_{j,j'})$, and $T=(t_{k, k'})$, sends $\tA=(a_{i, j, k})$ to $\tA'=(a'_{i,j,k})$ where
\begin{equation}\label{eq:action}
	a'_{i,j,k}=\sum_{i'\in[n],j'\in[m],k'\in[\ell]}r_{i,i'}s_{j,j'}t_{k,k'}a_{i',j',k'}.
\end{equation}
The \TIlong (\TI for short) problem then asks, given two $3$-way arrays $\tA, \tB\in\T(n\times m\times \ell, \F)$, whether $\tA$ and $\tB$ are in the same orbit under the above action. 

In the following, when we say $3$-tensors, we mean $3$-way arrays together with the action defined in Equation~\ref{eq:action}. To relate $3$-way arrays with groups and cubic forms, we will need to examine other group actions on structured $3$-way arrays, as we will explain now. 



\paragraph{$3$-way arrays from group isomorphism.} Let $\Lambda(n, \F)$ be the linear space of $n\times n$ alternating\footnote{An $n\times n$ matrix $A$ is alternating, if for any $u\in \F^n$, $u^tAu=0$. When the characteristic of $\F$ is not $2$, alternating is equivalent to skew-symmetric (i.e. $A=-A^t$).} matrices over $\F$. 
\begin{definition}\label{def:ms-isometry}
	The \AltMatSpIsomlong (\AMSI for short) problem asks the following: given $\vA=(A_1, \dots, A_m), \vB=(B_1, \dots, B_m)\in \Lambda(n, \F)^m$, decide if there exist $R\in\GL(n, \F)$ and $T=(t_{i,j})\in\GL(m, \F)$, such that for any $i\in[m]$, $RA_iR^t=\sum_{j\in[m]}t_{i,j}B_j$.
\end{definition}
Note that by naturally viewing matrix tuples as 3-way arrays, the input to \AMSI consists of two $3$-way arrays. It differs from \TI in that (1) there is a symmetry between 
two directions of the $3$-way arrays (i.e. $A_i=-A_i^t$), and (2) the same matrix $R\in\GL(n, \F)$ acts on two directions.


By Baer's correspondence \cite{Bae38}, \AltMatSpIsomlong over $\F_p$, $p>2$, is equivalent to \GpI for $p$-groups of class $2$ and exponent $p$ in the categorical sense. Algorithmically, solving \AMSI in time $\poly(n, m, \log p)$ is equivalent to solving \GpI for $\varpc(p, 2)$ in polynomial time in the matrix group model \cite{GQ2}, and solving \AMSI in time $p^{O(n+m)}$ is equivalent to solving \GpI for $\varpc(p, 2)$ in polynomial time in the Cayley table model \cite{GQ17}.

\paragraph{3-way arrays from cubic form equivalence.} Let $f\in\F[x_1, \dots, x_n]$ be a cubic form. When $\F$ is of characteristic not $2$ or $3$, we can associate $f$ with a trilinear form $\hat f:\F^n\times\F^n\times\F^n\to\F$, defined as $\hat f(u, v, w)=\frac{1}{6}(f(u+v+w)-f(u+v)-f(v+w)-f(u+w)+f(u)+f(v)+f(w))$. Note that $\hat f$ is symmetric, that is, for any $\pi\in\S_3$, 
\begin{equation}\label{eq:3sym}
\hat f(u_1, u_2, u_3)=\hat f(u_{\pi(1)}, u_{\pi(2)}, u_{\pi(3)}).
\end{equation} 

A symmetric trilinear form $\hat f$ can then be recorded as a symmetric $3$-way array $\mathtt{F}$, defined as $\mathtt{F}(i, j, k)=\hat f(e_i, e_j, e_k)$, where $e_i$ is the $i$th standard basis vector in $\F^n$. 
\begin{definition}\label{def:stfe}
	The \STFElong (\STFE for short) problem asks the following: given two symmetric $3$-way arrays $\mathtt{F}$ and $\mathtt{G}$, decide if there exists $R\in\GL(n, \F)$, such that the action of $(R, R, R)$ on $\mathtt{F}$ as in Equation~\ref{eq:action} yields $\mathtt{G}$. 
\end{definition}
Note that the differences of \STFE with \TI are (1) there is a symmetry between all the three directions in the $3$-way arrays (Equation~\ref{eq:3sym}) and (2) the same matrix $R\in\GL(n, \F)$ acts on all three directions.

It is well-known that two cubic forms $f$ and $g$ are equivalent if and only if the corresponding symmetric trilinear forms $\mathtt{F}$ and $\mathtt{G}$ are equivalent. Therefore, the above procedure can be viewed as a polynomial-time reduction. 

\paragraph{Three questions from the tensorial viewpoint.} In the above, we introduced \TIlong, and relate \GpI for $\varpc(p, 2)$ and \CubicForm with \AMSI and \STFE respectively, which are problems about group actions on $3$-way arrays. 
We now explain how these formulations help with the three questions in Section~\ref{subsec:3q}. 

For Question 1, by the classical reduction from \GpI in the Cayley table model to \GI \cite{KST93}, $\GI\in\P$ implies that \AMSI for inputs from $\Lambda(n, p)^m$, $p>2$, can be solved in time $p^{O(n+m)}$. If we could further reduce \STFE for inputs from $\F_p[x_1, \dots, x_n]$ to \AMSI for inputs from $\Lambda(n', p)^{m'}$, where $n', m'=O(n)$, we would deduce that $\GI\in\P$ implies \CubicForm for inputs from $\F_p[x_1, \dots, x_n]$ can be solved in time $p^{O(n)}$. 

For Question 2, the goal is to reduce \GpI for $\varpc(p, c)$ to \GpI for $\varpc(p, 2)$ when $p>c$. For this we need to introduce another problem. \AlgIsolong (\AlgIso for short) asks the following: given two bilinear maps $f, g:\F^n\times\F^n\to\F^n$, decide if there exists $A\in\GL(n, \F)$, such that for any $u, v\in \F^n$, $f(A(u), A(v))=A(g(u, v))$. Such bilinear maps can be viewed as defining multiplications for vectors, as in the definitions of associative or Lie algebras (where more properties are required, such as associativity or the Jacobi identity). 
\AlgIso has been studied in computer algebra \cite{Graaf,BW15} and computational complexity \cite{AS06,GrochowLie}.

Getting back to Question 2, the classical Lazard's correspondence \cite{lazard}, we can associate groups in $\varpc(p, c)$ with certain Lie algebras, and such a correspondence preserves and respects isomorphism types. This implies that we can reduce \GpI for $\varpc(p, c)$ to Lie \AlgIsolong. If we could reduce \AlgIso to \AMSI, which is equivalent to \GpI for $\varpc(p, 2)$, we would obtain a reduction from \GpI for $\varpc(p, c)$ to \GpI for $\varpc(p, 2)$. 


For Question 3, we briefly explain why 3-way arrays are more amenable for the purpose of search- and counting-to-decision reductions, compared to general groups. 
To achieve search-to-decision reductions, a common strategy is to start with a partial solution, and devise an instance that respects this partial solution. In \GIlong, a partial solution is a permutation fixing some vertices, and a \GI instance respecting this partial solution is a colored graph, namely graphs whose vertices are colored. To reduce colored \GI to ordinary \GI, certain gadgets are applied to ensure that the colored vertices are special. To carry out this strategy for groups, it is not clear how to construct groups whose automorphisms respect certain partial isomorphisms. Interestingly, for 3-way arrays, there are certain combinatorial constructions, such as ``concatenations'' of 3-way arrays, which allow for the use of certain gadgets.

\subsection{Main results II: on the complexity theory of $\cc{TI}$}\label{subsec:results}

In Section~\ref{subsec:connection}, we saw that \GpI for $\varpc(p, 2)$ and \CubicForm can be formulated as problems about different group actions on possibly structured $3$-way arrays. We also explained how the questions in Section~\ref{subsec:3q} could be answered from this tensorial viewpoint. In this subsection, we introduce the framework of group actions on 3-way arrays, and formally state our results. 

\paragraph{Five actions on 3-way arrays.} The fact that $3$-way arrays admit several natural group actions is not surprising. Indeed, matrices in $\M(n, \F)$, as 2-way arrays, admit three natural actions: $A\to LAR^{-1}$, $A\to LAL^t$, and $A\to LAL^{-1}$. These endow matrices with algebraic or geometric meanings: a linear map from $U$ to $V$, a bilinear form, and a linear map from $U$ to $U$, where $U$ and $V$ are vector spaces isomorphic to $\F^n$, denoted as $U\cong V\cong \F^n$.

For 3-way arrays, there are five natural actions.
\begin{definition}\label{def:5action}
	Let $U\cong\F^n$, $V\cong\F^m$, and $W\cong\F^\ell$.
	\begin{enumerate}
		\item Given $\tA\in\T(n\times m\times \ell, \F)$, $(R, S, T)\in\GL(n, \F)\times\GL(m,\F)\times\GL(\ell,\F)$ sends $\tA=(a_{i,j,k})$ to $\tA'=(a_{i,j,k}')$ as defined in Equation~\ref{eq:action}. This corresponds to the natural action of $\GL(U)\times\GL(V)\times\GL(W)$ on $U\otimes V\otimes W$. This action supports \TIlong. 
		\item Given $\tA\in\T(n\times n\times \ell, \F)$, $(R, T)\in\GL(n, \F)\times\GL(\ell, \F)$ sends $\tA$ to $\tA'$ by $(R, R, T)$ on $\tA$ as defined in Equation~\ref{eq:action}. This corresponds to the natural action of $\GL(U)\times\GL(W)$ on $U\otimes U\otimes W$. This action supports \AltMatSpIsomlong, which in turn relates to \GpIlong (Section~\ref{subsec:connection}).
		\item Given $\tA\in\T(n\times n\times \ell, \F)$, $(R, T)\in\GL(n, \F)\times\GL(\ell, \F)$ sends $\tA$ to $\tA'$ by $(R, R^{-t}, T)$ on $\tA$ as defined in Equation~\ref{eq:action}. Here $R^{-t}$ denotes the transpose inverse of $R$. This corresponds to the natural action of $\GL(U)\times\GL(W)$ on $U\otimes U^*\otimes W$, where $U^*$ denotes the dual space of $U$. 
		\item Given $\tA\in\T(n\times n\times n, \F)$, $R\in\GL(n, \F)$ sends $\tA$ to $\tA'$ by $(R, R, R)$ on $\tA$ as defined in Equation~\ref{eq:action}. This corresponds to the natural action of $\GL(U)$ on $U\otimes U\otimes U$. This action supports \STFElong, which in turn relates to \CubicFormlong (Section~\ref{subsec:connection}).
		\item Given $\tA\in\T(n\times n\times n, \F)$, $R\in\GL(n, \F)$ sends $\tA$ to $\tA'$ by $(R, R, R^{-t})$ on $\tA$ as defined in Equation~\ref{eq:action}. This corresponds to the natural action of $\GL(U)$ on $U\otimes U\otimes U^*$. Note that $\tA$ can be viewed as recording the structure constants of some (possibly non-associative) algebra, so this action supports \AlgIsolong.
	\end{enumerate}
\end{definition}


We also need the following notions for $3$-way arrays. Let $\tA\in\T(n\times m\times \ell, \F)$. 
The frontal slices of $\tA$ are $\{A_1, \dots, A_\ell\}\in\M(n\times m, \F)$, where $A_k(i,j)=a_{i,j,k}$.
Let $\tA\in\T(n\times n\times n, \F)$. Then $\tA$ is \emph{symmetric} if for any $\sigma\in\S_3$, $\tA(i,j,k)=\tA(\sigma(i), \sigma(j), \sigma(k))$, and $\tA$ is anti-symmetric if for any $\sigma\in\S_3$, $\tA(i,j,k)=\sgn(\sigma)\tA(\sigma(i), \sigma(j), \sigma(k))$. 

\paragraph{Equivalence between the five actions under linear-length reductions.} We can now state our first main result in the framework of 3-way arrays.
\begin{theorem}\label{thm:equivalence}
	For $i, j\in[5]$, $i\neq j$, let $\tA$ and $\tB$ two $3$-way arrays compatible to the $i$th action defined in Definition~\ref{def:5action} of length $L$. Then there exists a polynomial-time computable function $f$ that takes $\tA$ and $\tB$ and outputs $3$-way arrays $f(\tA)$ and $f(\tB)$, such that (1) the lengths of $f(\tA)$ and $f(\tB)$ are upper bounded by $O(L)$, and (2) $\tA$ and $\tB$ are in the same orbit under the $i$th action if and only if $f(\tA)$ and $f(\tB)$ are in the same orbit under the $j$th action.
	
	Furthermore, we have the following structural restrictions.
	\begin{enumerate}
		\item If $j=2$, i.e. the action of $\GL(U)\times \GL(W)$ on $U\otimes U\otimes W$, the frontal slices of $f(\tA)$ and $f(\tB)$ are symmetric (or skew-symmetric). 
		\item If $j=4$, i.e. the action of $\GL(U)$ on $U\otimes U\otimes U$, $f(\tA)$ and $f(\tB)$ are symmetric (or anti-symmetric) $3$-way arrays. 
		\item If $j=5$, i.e. the action of $\GL(U)$ on $U\otimes U\otimes U^*$, $f(\tA)$ and $f(\tB)$ record the structure constants of associative or Lie algebras. 
	\end{enumerate}
	
	When $\F=\F_q$, these problems are equivalent under $q^{O(L)}$-time reductions. 
\end{theorem}

Previously in \cite{GQ1}, a version of Theorem~\ref{thm:equivalence} was proved but with the lengths of $f(\tA)$ and $f(\tB)$ upper bounded by $O(L^2)$, instead of $O(L)$ as here. This improvement from quadratic to linear blow-up is what enables the progress on the three questions in Section~\ref{subsec:progress_3q}.

\paragraph{Applications of Theorem~\ref{thm:equivalence}.} Theorem~\ref{thm:equivalence} has the following consequences. First, $\GI\in\P$ implies non-trivial speed-up of a wide range of algebraic isomorphism problems, such as \CubicForm (Theorem~\ref{thm:GI}) and \AlgIso (Theorem~\ref{thm:GI2}). Second, the reduction from testing isomorphism of $p$-groups of class-$c$ and exponent $p$ to $p$-groups of class $2$ and exponent $p$ as in Corollary~\ref{cor:nil_reduction}. Third, combining with \cite[Theorem 1.2]{Sun23}, faster algorithms for \CubicForm and \AlgIso as follows.
\begin{corollary}\label{cor:algiso}
	\begin{enumerate}
		\item \AlgIsolong over $\F_q^n$ can be solved in time $q^{\tilde O(n^{1.8}\cdot\log q)}$. 
		\item When $\F_q$ is of characteristic $>3$, \CubicFormlong over $\F_q$ in $n$-variables can be solved in time $q^{\tilde O(n^{1.8}\cdot \log q)}$.
	\end{enumerate}
\end{corollary}

\begin{remark}[Computing cosets of isomorphisms]
So far we have been focussing on the decision versions of the isomorphism problems. More generally, as in the case of graph isomorphism \cite{Bab16}, when two tensors $\tA, \tB\in\T(n\times m\times \ell, \F)$ are isomorphic, we can require the algorithm output the coset of isomorphisms from $\tA$ to $\tB$, as a coset representative and a generating set of the automorphism group of $\tA$. The reductions supporting Theorem~\ref{thm:equivalence} allow for transforming cosets of isomorphisms from one problem (say \AMSI) to that of the other (say \TI), and the main reason is that they satisfy certain ``functorial'' properties as in \cite[Section 2.3]{GQ1}. See Remarks~\ref{rem:TI4_coset} and~\ref{rem:TI1_coset} for further details.
\end{remark}

\subsection{On the technique: a more efficient gadget for \algprobm{Partitioned Tensor Isomorphism}}\label{subsec:background}

A starting point of the theory of $\cc{TI}$ is to introduce the \algprobm{Partitioned Tensor Isomorphism} problem (\pTI for short) and to show that \pTI reduces to \TI \cite{FGS19}. This problem can be viewed as corresponding to the \emph{colored graph isomorphism} (\cGI for short) in the study of $\cc{GI}$ \cite{KST93}. 

The reduction from \pTI to \TI in \cite{FGS19} is enabled by a gadget which we call the \emph{Futorny--Grochow--Sergeichuk gadget}, or FGS gadget for short. One serious limitation of this gadget design is that it blows up the lengths of the resulting 3-way arrays by a quadratic factor. Our first main technical innovation is a new gadget design that achieves only a linear blow-up of the dimensions.

In the following, we will define \pTI and review the FGS gadget to explain how the quadratic blow-up appears. Then we will present the main ideas and an illustration for our new gadget.

\paragraph{Partitioned tensor isomorphism and the FGS gadget.} Let $\tA=(a_{i,j,k})\in \T(n\times m\times \ell, \F)$. Suppose the third index set $[\ell]$ is \emph{partitioned} as $[\ell]=D_1\uplus D_2\uplus  \dots \uplus D_Q$, where $|D_i|=\ell_i$ and $\sum_{i\in[Q]}\ell_i=\ell$. 
There is a natural embedding of $\GL(\ell_1, \F)\times\dots\times\GL(\ell_Q, \F)$ as a subgroup of $\GL(\ell, \F)$, where $R_i\in \GL(\ell_i, \F)$ acts on the indices in $D_i$.
\begin{definition}[\algprobm{Partitioned Tensor Isomorphism}, \pTI]\label{def:pTI}
Given $\tA, \tB\in \T(n\times m\times \ell, \F)$, and a partition of $[\ell]=D_1\uplus  D_2\uplus  \dots \uplus  D_Q$, the \algprobm{Partitioned Tensor Isomorphism} (\pTI) problem asks whether there exists $(R, S, T)$ that sends $\tA$ to $\tB$ via the action defined in Equation~\ref{eq:action}, where $R\in \GL(n, \F)$, $S\in\GL(m, \F)$, and $T\in\GL(\ell_1, \F)\times\dots\times\GL(\ell_Q, \F)$ viewed as a subgroup of $\GL(\ell, \F)$ respecting the partition.
\end{definition}
In general, we can also impose partitions of the first and second index sets $[n]$ and $[m]$, and require that $R\in \GL(n, \F)$ and $S\in\GL(m, \F)$ respect the partitions; see Section~\ref{sec:partition}.

We review the FGS gadget in the following simplified scenario. By slicing along the third index, a $3$-way array $\tA\in\T(n\times m\times \ell, \F)$ can be represented as a tuple of matrices $\vA=(A_1, \dots, A_\ell)\in\M(n\times m, \F)^\ell$. Suppose $Q=2$, $n\leq m$, and the partition of $[\ell]$ consists of $D_1=\{1,\dots, \ell-1\}, D_2=\{\ell\}$. To reduce this \pTI to \TI, we construct 
\begin{multline}\label{eq:fgs}
	\tilde A_1=\begin{bmatrix}
		A_1 & 0 & 0 & \dots & 0 & 0  \\
		0 & \IdMat_n & 0 & \dots & 0 & 0  \\
		0 & 0 & 0 & \dots & 0 & 0 \\
	\end{bmatrix}, 
	\tilde A_2=\begin{bmatrix}
		A_2& 0 & 0 & \dots & 0 & 0  \\
		0 & 0 & \IdMat_n & \dots & 0 & 0  \\
		0 & 0 & 0 & \dots & 0 & 0  \\
	\end{bmatrix},
	\dots,
	\\
	\tilde A_{\ell-1}=\begin{bmatrix}
		A_{\ell-1} & 0 & 0 & \dots & 0 & 0  \\
		0 & 0 & 0 & \dots & \IdMat_n & 0  \\
		0 & 0 & 0 & \dots & 0 & 0 \\
	\end{bmatrix}, 
	\tilde A_\ell=\begin{bmatrix}
		A_\ell & 0 & 0 & \dots & 0 & 0 \\
		0 & 0 & 0 & \dots & 0 & 0 \\
		0 & 0 & 0 & \dots & 0 &  \IdMat_{2n} \\
		\end{bmatrix}.
\end{multline} 
Let $\tilde\vA=(\tilde A_1, \dots, \tilde A_\ell)$ and $\tilde\tA$ be the 3-way array whose frontal slices are $\tilde A_i$. By \cite[Lemma 2.2]{FGS19}, $\tA$ and $\tB$ in $\M(n\times m, \F)^\ell$ are partitioned isomorphic if and only if $\tilde\tA$ and $\tilde\tB$ are isomorphic. 

The key idea of the FGS gadget is to utilise the basic fact that ranks are invariant under left- and right-multiplying invertible matrices. Note that those $\IdMat_n$ and $\IdMat_{2n}$ ensure that the partition is respected. For example, if a linear combination involves both $\tilde A_1$ and $\tilde A_\ell$, then it would yield a matrix of rank larger than $\rank(\tilde A_1)$ or $\rank(\tilde A_\ell)$, which is not allowed if we wish to test isomorphism with another matrix tuple with the same rank profile. While these observations give a first clue as to why the FGS gadget works, a rigorous proof requires some work and utilises the Krull--Schmidt theorem for quiver representations \cite{FGS19}. 

A key issue with the FGS gadget is that the sizes of $\tilde A_i$ are quadratic in the sizes of $A_i$. Indeed, in Equation~\ref{eq:fgs}, $\tilde A_i$ is of size $4n\times (m+(\ell-1)n+2n)$. This is the cause of the quadratic blow-ups of the results in \cite{GQ1}. 

\paragraph{A gadget for partitioned tensor isomorphism with linear-size blow-up.} Our main technical result is the following reduction from \pTI to \TI with only linear-size blow-ups in the lengths.
\begin{theorem}\label{thm:informal_pTI}
	Fix a partition $[\ell]=D_1\cup \dots \cup D_Q$, and suppose $n\leq m$. There exists a function $g:\T(n\times m\times\ell, \F)\to\T(n'\times m'\times \ell', \F)$, such that $\tA, \tB\in\T(n\times m\times \ell, \F)$ are isomorphic as partitioned tensors if and only if $g(\tA)$ and $g(\tB)$ are isomorphic as plain tensors. In particular, $n', m', \ell'$, and the time needed to compute $g$, are upper bounded by $2^{O(Q)}\cdot O(n+\ell+m)$.
\end{theorem}
In contrast, the previous result \cite[Theorem 2.1]{FGS19} has $2^{O(Q)}\cdot (n\cdot \ell+m)$ as the bound. Note that the resulting dimension depends exponentially on the number of parts $Q$, and this is fine because for Theorem~\ref{thm:equivalence}, $Q$ is used as a constant.

\paragraph{Main ideas and an illustration of the new gadget.} 
We now briefly explain our new gadget design. Besides utilising ranks to distinguish certain slices as in the FGS gadget, two new ideas are introduced: the first one is a linear algebra result, and the second one is a certain cancellation property. The gadget is achieved by combining these three ingredients carefully.

Let us consider a simplified setting. Let a $3$-way array $\tA\in\T(n\times m\times \ell, \F)$ be represented as a matrix tuple $\vA=(A_1, \dots, A_\ell)\in\M(n\times m, \F)^\ell$. Similarly, let $\tB\in\T(n\times m\times \ell, \F)$ be represented as a matrix tuple $\vB=(B_1, \dots, B_\ell)\in\M(n\times m, \F)^\ell$.

For the convenience of exposition, we now partition the first index $[n]$ as $D_1=\{1, \dots, a\}$ and $D_2=\{a+1, \dots, n\}$ for some $a\in[n-1]$. Set $b:=n-a$. Our goal is to restrict the left matrix $R\in\GL(n, \F)$ is of the block diagonal form $\begin{bmatrix}
	R_1 & 0 \\
	0 & R_2
\end{bmatrix}$ where $R_1\in\GL(a, \F)$ and $R_2\in\GL(b, \F)$.

The starting point is to introduce the matrix $C:=\begin{bmatrix}
	0_{a\times m} & \IdMat_{a} \\
	0_{b\times m} & 0_{b\times a}
\end{bmatrix}\in\M(n\times (m+a), \F)$. Consider $R=\begin{bmatrix}
R_1 & R_2\\
R_3 & R_4
\end{bmatrix}\in \GL(n, \F)$ where $R_1\in\GL(a, \F)$, and $S=\begin{bmatrix}
S_1 & S_2\\
S_3 & S_4
\end{bmatrix}\in\GL(m+a, \F)$ where $S_1\in\GL(m, \F)$. It is a basic linear algebra fact that, if $RCS=C$, then $R_3=0_{b\times a}$ and $S_3=0_{a\times m}$ (see Lemma~\ref{lem:block}). This means that $R$ has to be block upper-triangular, which is a step closer to block diagonal. To get block diagonal can be done by enlarging the matrices and adding $\IdMat_b$ in an appropriate position, but let us omit this for now to ease the exposition of the main ideas.

As a result, construct $\vA'=(A_1', \dots, A_\ell', A_{\ell+1}')\in \M(n\times (m+a), \F)^{\ell+1}$, where for $i\in[\ell]$, $A_i'=\begin{bmatrix}
	A_i & 0_{n\times a}
\end{bmatrix}$, and $A_{\ell+1}'=C$. Similarly, construct $\vB'\in\M(n\times (m+a), \F)^{\ell+1}$ from $\vB$. If we could ensure that $A_{\ell+1}'=C$ must be mapped to $B_{\ell+1}'=C$, then for any isomorphism $(R, S', T')$ sending $\vA'$ to $\vB'$ needs to satisfy that $R\in\GL(n, \F)$ is block upper-triangular. 

To achieve that, let us first make $A_{\ell+1}'$ ``distinguished'' using \emph{rank}, as in the FGS gadget setting.
That is, for $i\in[\ell]$, set $A_i''=\begin{bmatrix}
A_i' & 0 \\
0 & 0_{n\times n}
\end{bmatrix}$, and set $A_{\ell+1}''=\begin{bmatrix}
A_{\ell+1}'& 0 \\
0 & \IdMat_n
\end{bmatrix}$. Then consider $\vA''=(A_1'', \dots, A_\ell'', A_{\ell+1}'')\in\M(2n\times (m+a+n), \F)^{\ell+1}$. Similarly, construct $\vB''\in \M(2n\times (m+a+n), \F)^{\ell+1}$. Because for any $i\in [\ell]$, $\rk(A_i)\leq n<n+a=\rk(A_{\ell+1}'')=\rk(B_{\ell+1}'')$, we see that any isomorphism $(R'', S'', T'')$ sending $\vA''$ to $\vB''$ needs to send some rank-$(n+a)$ matrix in the linear span of $A_i''$ to $B_{\ell+1}''$. Note that any rank-$(n+a)$ matrix $\bar A$ in the linear span of $A_i''$ has to involve $A_{\ell+1}''$, so $\bar A$ is of the form $\begin{bmatrix}
\bar A_1 & \lambda \IdMat_a & 0 \\
\bar A_2 & 0 & 0 \\
0 & 0 & \lambda \IdMat_n
\end{bmatrix}$ for some $\lambda\in\F$. On the other hand, $B_{\ell+1}''=\begin{bmatrix}0 & \IdMat_a & 0 \\
0_{b\times m} & 0 & 0 \\
0 & 0 & \lambda \IdMat_n
\end{bmatrix}$. So the submatrix $\begin{bmatrix}
\bar A_1 \\
\bar A_2
\end{bmatrix}$ stands in the way of our plan. 

To \emph{cancel} $\begin{bmatrix}
	\bar A_1 \\
	\bar A_2
\end{bmatrix}$, we enlarge the matrix again to add $\IdMat_m$ below $\begin{bmatrix}
\bar A_1 \\
\bar A_2
\end{bmatrix}$ to $A_{\ell+1}''$. More specifically, for $i\in[\ell]$, let $\hat{A_i}=\begin{bmatrix}
A_i''\\
0_{m\times (m+a+n)}
\end{bmatrix}$, and let $\hat{A}_{\ell+1}=\begin{bmatrix}
A_{\ell+1}'' \\
\IdMat_m \quad 0_{m\times (a+n)}
\end{bmatrix}=\begin{bmatrix}
0_{a\times m} & \IdMat_a & 0_{a\times n}\\
0_{b\times m} & 0_{b\times a} & 0_{b\times n}\\
0_{n\times m} & 0_{n\times a} & \IdMat_n \\
\IdMat_m & 0_{m\times a} & 0_{m\times n}
\end{bmatrix}$. Let $\hat{\vA}=(\hat{A_1}, \dots, \hat{A_\ell}, \hat{A}_{\ell+1})\in\M((2n+m)\times (m+a+n), \F)^{\ell+1}$. Similarly construct $\hat{\vB}\in \M((2n+m)\times (m+a+n), \F)^{\ell+1}$. 

This is the desired construction that achieves enforcing the original $R\in\GL(n, \F)$ to be block upper-triangular. Note that the resulting $3$-way array $\hat{\tA}$ does have length $O(n+m+\ell)$. A brief explanation or recap goes as follows. First, because of \emph{rank} property, an isomorphism from $\hat{\vA}$ to $\hat{\vB}$ needs to send a rank-$(a+m+n)$ matrix $\bar{A}$ in the span of $\hat{\vA}$ to $\hat{B}_{\ell+1}=\hat{A}_{\ell+1}$. Second, this $\bar{A}$ may not be zero in the upper-left $n\times m$ submatrix, but this can be \emph{cancelled} out using $\IdMat_m$ introduced in the last step. After this transformation, and as the third argument, sending $\hat{A}_{\ell+1}$ to $\hat{B}_{\ell+1}=\hat{A}_{\ell+1}$ has the effect \emph{enforcing} a block upper-triangular structure. 
Of course, the arguments need to be carefully expanded into a formal proof. 


The above is about two parts in the partition. To deal with the general situation, besides spelling out the proof details missing from the above, one also needs to keep track all the parts carefully and calculate the parameters, as done in Section~\ref{sec:partition}. To facilitate an easier understanding of that, we also provide a relatively shorter but complete proof for another setting in Section~\ref{sec:warm-up}, partly as a warm up for the proof in Section~\ref{sec:partition}. It should also be noted that the specific gadget formats may be different across different settings (compare Equation~\ref{eq:gadget_amsi} and Equation~\ref{eq:partTI}), and the cancelling argument could be more involved than the above (see e.g. the proof of clearing out $Y$ in Theorem~\ref{thm:block-diagonal-AMSI}).

\subsection{Main results III: \TIlong over matrix groups}\label{subsec:results-matrix}

\paragraph{\TIlong over matrix groups.}  
The following group-theoretic formulation captures colored graph isomorphism (\cGI) and partitioned tensor isomorphism (\pTI) and will be useful for search- and counting-to-decision reductions for \GpI. Let $G\leq\GL(n, \F)$ be a matrix group. Suppose we wish to test whether $\tA, \tB\in\T(n\times m\times \ell, \F)$ are in the same orbit under the action of $G\times\GL(m, \F)\times\GL(\ell, \F)$, and our goal is to reduce such an isomorphism to the plain \TI problem. 
 
Indeed, \cGI and \pTI fall into this framework: for \cGI, we are interested in permutations of vertices that come from some \emph{Young subgroup} of the symmetric group\footnote{Given a partition of $[n]=S_1\cup\dots\cup S_d$, the Young subgroup corresponding to this partition is the set of permutations $\pi$ that respect this partition, i.e. $\forall i\in[d]$, $\pi(S_i)=S_i$; see \cite{Ker06}.}. For \pTI, the invertible matrices of interest are from some \emph{Levi subgroup} of the general linear group \footnote{Given a direct sum decomposition $\F^n=U_1\oplus \dots\oplus U_d$, the Levi subgroup corresponding to this decomposition consists of invertible matrices $T$ that preserve this decomposition, i.e. $\forall i\in[d]$, $T(U_i)=U_i$; see \cite{Bor12}.}.

For the purpose of search- and counting-to-decision reductions, the key is to tackle this problem for $G$ being \emph{monomial subgroups} (consisting of invertible matrices where each row and each column has exactly one nonzero entry) and \emph{diagonal subgroups} (consisting of diagonal matrices).
In \cite{GQ2}, gadgets were designed for restricting to monomial and diagonal groups with quadratic blow-ups, and they require the field order to be large enough in the diagonal case. The two gadgets designs were achieved in a spontaneous fashion.

\paragraph{A framework for restricting to isomorphisms by matrix groups and instantiations.} 
By utilising our \pTI gadget, we develop a framework for restricting to any subgroup of $G\leq\GL(n, \F)$, provided that $G$ appears as the first component of some $\tA_G\in\T(n\times m'\times \ell', \F)$ (see Section~\ref{subsec:framework}). 
We demonstrate the uses of this framework by instantiating $G$ as monomial or diagonal groups. For diagonal groups, we can get rid of the field order conditions in \cite{GQ2}. 
\begin{theorem}\label{thm:restriction}
	Let $G=\Mon(n, \F)\leq\GL(n, \F)$ be the monomial subgroup. There exists a polynomial-time computable function $f:\T(n\times m\times\ell, \F)\to\T(n'\times m'\times \ell', \F)$, such that $\tA, \tB\in\T(n\times m\times \ell, \F)$ are isomorphic under $G\times\GL(m, \F)\times\GL(\ell, q)$ if and only if $f(\tA)$ and $f(\tB)$ are isomorphic as plain tensors. In particular, $n', m', \ell'$ are upper bounded by $O(n+m+\ell)$.
	
	The above also holds for $G=\diag(n, \F)$ when $f$ is a Las-Vegas randomised polynomial-time computable function.
\end{theorem}
Interestingly, the proof for the diagonal subgroup relies on connections between graphs and matrix spaces \cite{LQWWZ} and random regular graphs \cite{Wor99}.

\paragraph{Organisation of this paper.} In Section~\ref{sec:warm-up}, we present a reduction from \TIlong to \AltMatSpIsomlong. This is a relatively succinct example to illustrate the main construction and ideas in a concrete setting. In Section~\ref{sec:prel} we present some preliminaries. In Section~\ref{sec:gadget}, we distil he basic construction of the new gadget and prove some basic properties of it, so that it will be easy to use in more complex settings. We also present the framework for restricting to other subgroups and instantiate that with monomial and diagonal subgroups in relation to Theorem~\ref{thm:restriction}. In Section~\ref{sec:partition} we prove Theorem~\ref{thm:informal_pTI}. In Section~\ref{sec:equivalence} we prove Theorem~\ref{thm:equivalence}. Finally in Section~\ref{sec:application} we give proofs for Theorems~\ref{thm:GI}, \ref{thm:cto2}, and~\ref{thm:search-counting}, and Corollary~\ref{cor:algiso}.

\section{From \TI to \AltMatSpIsomlong}\label{sec:warm-up}

\begin{theorem} \label{thm:AltMatSpIsom}
	\TI reduces to \AltMatSpIsomlong with linear blow-up. In particular, for $n \times m \times \ell$ tensors, the output is an $(\ell+1)$-dimensional space of matrices of size at most $3n+2m+2$.
\end{theorem}

\begin{proof}
	Suppose $\tA$ is $n \times m \times \ell$, and let the frontal slices of $\tA$ be $A_1, \dotsc, A_\ell$, each of which is an $n \times m$ matrix. We will use parameters $r$ and $s$ which we will set later. 
	Let $\cA$ be the matrix space spanned by the following slices:
	\begin{itemize}
		\item For $i=1,\dotsc,\ell$,
		\[
		\tilde A_{i} =
		\begin{bmatrix}
			0_{n} & A_i & 0_{n \times n} & 0_{n \times 2r}  \\
			-A_i^t & 0_{m}  \\
			& & 0_n \\
			& & & 0_{2r} \\
		\end{bmatrix}.
		\]
		
		\item A standard alternating slice of rank $2r + 2n$, connected to the first $n$ rows by an $\IdMat_n$ in the appropriate place:
		\[
		\tilde A_{\ell+1} = 
		\begin{bmatrix}
			0_n& 0_{n \times m}  & \IdMat_n  & 0_{n \times 2r} \\
			0_{m \times n} & 0_m & \\
			-\IdMat_n &  & 0_n \\
			& & & \begin{bmatrix} 0 & \IdMat_{r} \\ -\IdMat_{r} & 0 \end{bmatrix} \\
		\end{bmatrix}.
		\]
	\end{itemize}
	
	Now, note that the rank of any linear combinations of the first $\ell$ slices is at most $n + m$. The condition on $r$  that we need is
	\[
	2r > n + m.
	\]
	This will enforce that we cannot add $\tilde A_{\ell+1}$ to any of the first $\ell$ slices, as this would make their ranks strictly larger than $n+m$. (Note that the $\IdMat_n$ in $\tilde A_{\ell+1}$ might not contribute $n$ to the rank of $\tilde A_i + \tilde A_{\ell+1}$, since it occurs in the same rows as $A_i$.) This condition is easily satisfied, for example by setting $r=\lfloor (n+m)/2 \rfloor + 1$.
	
	We claim that the map $\tA \mapsto \cA = \langle \tilde A_1, \dotsc, \tilde A_{\ell+1} \rangle$ gives a reduction from \TI to \AltMatSpIsomlong, that is, that $\tA \cong \tB$ as 3-tensors if and only if $\cA$ and $\cB$ are isometric matrix spaces.
	
	Note that the matrix tuple which is a basis for $\cA$ has dimensions $(2n + m + 2r) \times (2n + m + 2r) \times (\ell+1)$, and that $2n + m + 2r < 3n + 2m + 2$, so the dimensions of $\cA$ are linear in those of $\tA$.
	
	($\Rightarrow$) Suppose $\tA \cong \tB$, via $(P,Q,R)$, that is, $(P \tA Q^t)^R = \tB$, or in coordinates
	\[
	\tB(i,j,k) = \sum_{i'j'k'} P_{ii'} Q_{jj'} R_{kk'} \tA(i',j',k').
	\]
	Let $\tilde P = \diag(P, Q, P^{-t}, \IdMat_{2r})$, $\tilde R = \begin{bmatrix} R & 0 \\ 0 & 1 \end{bmatrix}$; we claim that $(\tilde P, \tilde R)$ is a pseudo-isometry of the matrix tuples $(\tilde A_1, \dotsc, \tilde A_{\ell+1})$ and $(\tilde B_1, \dotsc, \tilde B_{\ell+1})$. 
	
	First let us consider the last slice. Since $\tilde R$ is 1 in its lower-right corner, the last slice is unchanged by $\tilde R$. Let us see how it is affected by the isometry action of $\tilde P$. For the $\ell+1$ slice we have:
	\[
	\begin{bmatrix}
		P \\
		& Q \\
		& & P^{-t} \\
		& & & \IdMat_{2r} \\
	\end{bmatrix}
	\begin{bmatrix}
		0_n& 0_{n \times m}  & \IdMat_n  & 0_{n \times 2r} \\
		0_{m \times n} & 0_m & \\
		-\IdMat_n &  & 0_n \\
		& & & \begin{bmatrix} 0 & \IdMat_{r} \\ -\IdMat_{r} & 0 \end{bmatrix} \\
	\end{bmatrix}
	\begin{bmatrix}
		P^t \\
		& Q^t \\
		& & P^{-1} \\
		& & & \IdMat_{2r} \\
	\end{bmatrix}
	\]
	From this we see that in the 4th diagonal block, we get $\IdMat_{2r} \begin{bmatrix} 0 & \IdMat_r \\ -\IdMat_r & 0 \end{bmatrix} \IdMat_{2r} = \begin{bmatrix} 0 & \IdMat_r \\ -\IdMat_r & 0 \end{bmatrix}$, as desired. Because $\tilde P$ is block diagonal, all the other zeros in $\tilde A_{\ell+1}$ remain zero, and all that is left to check are the (1,3) and (3,1) blocks; we check the (1,3) and the other follows by (skew-)symmetry: it is $P \IdMat_n P^{-1} = \IdMat_n$, as desired. 
	
	All that remains is to check the first $\ell$ slices. Because these have the form $\begin{bmatrix} 0 & A_i \\ -A_i^t & 0 \\ & & 0_{n+2r} \end{bmatrix}$, and $\tilde P$ is block-diagonal commensurate with the blocks of the $\tilde A_i$, we only focus on the upper $2 \times 2$ blocks. For these, we have
	\[
	\begin{bmatrix}
		P \\
		& & Q
	\end{bmatrix}
	\begin{bmatrix} 
		0 & A_i \\
		-A_i^t & 0 
	\end{bmatrix}
	\begin{bmatrix}
		P^t \\
		& & Q^t
	\end{bmatrix}
	=
	\begin{bmatrix}
		0 & P A_i Q^t \\
		- Q A_i^t P^t & 0
	\end{bmatrix}
	\]
	Since $\tilde R$ acts on the first $\ell$ slices the same as $R$, after applying $\tilde R$ we find in the (1,2) block $\sum_{i'} R_{ii'} P A_{i'} Q^t = B_i$, as desired. Thus $(\tilde P, \tilde R)$ is a pseudo-isometry of the corresponding matrix tuples, and thus $\tilde A$ and $\tilde B$ are isometric matrix spaces, as claimed.
	
	($\Leftarrow$) Suppose that $\cA$ and $\cB$ are isometric matrix spaces, via $(P, R)$, that is, $\sum_{i'} R_{ii'} P \tilde A_{i'} P^t = B_i$ for $i = 1,\dotsc,\ell+1$. Let us write $R$ in block form commensurate with (some of) the blocks of the construction:
	\[
	R = \begin{bmatrix}
		R_{11} & R_{12} \\
		R_{21} & R_{22}
	\end{bmatrix}
	\]
	where $R_{11}$ is $\ell \times \ell$ and $R_{22}$ is $1 \times 1$.
	
	From Lemma~\ref{lem:indiv} with (following the notation of the lemma) $N=1$ and $T_1 := \begin{bmatrix} 0 & \IdMat_r \\ -\IdMat_r & 0 \end{bmatrix}$, we get that $R_{12} = 0$. (Note that our $\IdMat_n$ blocks here occur in the $*$ blocks of Lemma~\ref{lem:indiv}.)
	
	
	Now, let $(\tilde A_1',\dotsc, \tilde A_{\ell+1}') = (\tilde A_1, \dotsc, \tilde A_{\ell+1})^R$. Since the actions of $P$ and $R$ commute with one another, we have $P \tilde A_i' P^t = \tilde B_i$ for all $i=1,\dotsc,\ell+1$. 
	
	We have 
	\[
	\tilde A_{\ell+1}' = \begin{bmatrix}
		0 & X & \sigma  \IdMat_n & 0  \\
		-X^t & 0 & 0 & 0 \\
		-\sigma \IdMat_n & 0 & 0 & 0 \\
		0 & 0 & 0 & \sigma J_r
	\end{bmatrix},
	\]
	where $X = \sum_{i=1}^\ell (R_{21})_{1i} A_i$, $J_r = \begin{bmatrix} 0 & \IdMat_r \\ -\IdMat_r & 0 \end{bmatrix}$, $\sigma$ is the nonzero entry of $R_{22}$, and the block sizes are $n, m, n, 2r$ (in order). Now apply the following isometry to $(\tilde A_1', \dotsc, \tilde A_{\ell+1}')$: 
	\[
	P_0 := \begin{bmatrix}
		\IdMat_n & \\
		& \IdMat_m & -X^t/\sigma  \\
		& 0 & \IdMat_n  \\
		& & & \IdMat_{2r}
	\end{bmatrix}.
	\]
	
	Because of the zeros in the third block-column and block-row of $\tilde A_i'$ for $i=1,\dotsc,\ell$, those slices are unchanged by this isometry. $\tilde A_{\ell+1}'$ is modified as follows:
	\begin{eqnarray*}
		\begin{bmatrix}
			\IdMat_n & \\
			& \IdMat_m & -X^t/\sigma  \\
			& 0 & \IdMat_n  \\
			& & & \IdMat_{2r}
		\end{bmatrix}
		\begin{bmatrix}
			0 & X & \sigma I \\
			-X^t & 0 & 0 \\
			-\sigma I & 0 & 0 \\
			& & & \sigma J_{2r}
		\end{bmatrix}
		\begin{bmatrix}
			\IdMat_n & 0 & 0 \\
			0 &  \IdMat_m & 0 \\
			0 & -X/\sigma &  \IdMat_n \\
			& & & \IdMat_{2r}
		\end{bmatrix} \\
		=
		\begin{bmatrix}
			0 & 0 & \sigma \IdMat_n \\
			0 & 0 & 0 \\
			-\sigma \IdMat_n & 0 & 0 \\
			& & & \sigma J_{2r}
		\end{bmatrix}
		= \sigma \tilde A_{\ell+1}.
	\end{eqnarray*}

		
		Thus we have $P_0(\tilde A'_1, \dotsc, \tilde A'_{\ell+1})P_0^t = (\tilde A'_1, \dotsc, \tilde A'_{\ell}, \sigma \tilde A_{\ell+1})$. Let $\cA'' = \langle \tilde A'_1, \dotsc, \tilde A'_{\ell}, \sigma \tilde A_{\ell+1} \rangle$. Since we already have $P \cA' P^t = \cB$ and $\cA'' = P_0 \cA' P_0^t$, we get 
		\[
		(P P_0^{-1}) \cA'' (P_0^{-t} P^t) = \cB.
		\]
		Let $P' = P P_0^{-1}$, and let us see what we can learn about $P'$.
		
		%
		%
		%
		
		Since we have already applied the action of $R$, we have $P' \sigma \tilde A_{\ell+1} (P')^t = \tilde B_{\ell+1}$. Let us see what constraints this puts on $P'$. The following lemma is a direct application of Lemma~\ref{lem:block}, but we list it here in a form more directly applicable to the current setting, because it will also be useful in this form in the proof of Theorem~\ref{thm:altform} below.
		
		\begin{lemma} \label{lem:block2}
			Let $n,m,r,t \geq 0$. Suppose a matrix
			\[
			P' = \begin{bmatrix} 
				P_{11} & \dotsb & P_{16} \\
				\vdots & \ddots  & \vdots \\
				P_{61} & \dotsb & P_{66}
			\end{bmatrix},
			\]
			with diagonal block sizes $n,m,n,r,r,t$,\ynote{Why do we need the $t$-size block?} is such that
			\[
			\sigma P'
			\begin{bmatrix}
				0_n & 0 & \IdMat_n  \\
				0 & 0_m & 0  \\
				-\IdMat_n & 0 & 0_n \\
				&  &  &  0_r & \IdMat_r \\
				&  &  & -\IdMat_r & 0_r \\
				& & & & & 0_t
			\end{bmatrix}
			(P')^t
			=
			\begin{bmatrix}
				0 & 0 & \IdMat_n  \\
				0 & 0 & 0  \\
				-\IdMat_n & 0 & 0 \\
				&  &  &  0 & \IdMat_r \\
				&  &  & -\IdMat_r & 0 \\
				& & & & & 0_t
			\end{bmatrix},
			\]
			for some nonzero scalar $\sigma$. Then 
			\[
			\begin{bmatrix}
				P_{21} & P_{23} & P_{24} & P_{25} \\
				P_{61} & P_{63} & P_{64} & P_{65}
			\end{bmatrix} = 0.
			\]
		\end{lemma}
		
		\begin{proof}
			Let $\pi$ be the permutation matrix that moves the second block-row to the fifth block-row. Then 
			\[
			\pi P' \pi^t = 
			\left[\begin{array}{cccc;{2pt/2pt}cc}
				P_{11} & P_{13} & P_{14} & P_{15} & P_{12} & P_{16}\\
				P_{31} & P_{33} & P_{34} & P_{35} & P_{32} & P_{36}\\
				P_{41} & P_{43} & P_{44} & P_{45} & P_{42} & P_{46} \\
				P_{51} & P_{53} & P_{54} & P_{55} & P_{52} & P_{56} \\ \hdashline[2pt/2pt] 
				P_{21} & P_{23} & P_{24} & P_{25} & P_{22} & P_{26} \\
				P_{61} & P_{63} & P_{64} & P_{65} & P_{62} & P_{66}
			\end{array}\right].
			\]
			The equation in the assumption of the lemma is equivalent to 
			\[
			(\pi P' \pi^t) 
			\begin{bmatrix}
				0 & \sigma \IdMat_n \\
				-\sigma \IdMat_n & 0 \\
				& & 0_r & \sigma \IdMat_r \\
				& & -\sigma \IdMat_r & 0_r \\
				& & & & 0_{m+t}
			\end{bmatrix}
			(\pi P' \pi^t)^t
			=
			\begin{bmatrix}
				0 & \IdMat_n \\
				-\IdMat_n & 0 \\
				& & 0_r & \IdMat_r \\
				& & -\IdMat_r & 0_r \\
				& & & & 0_{m+t}
			\end{bmatrix}
			\]
			As the upper-left $(2n + 2r) \times (2n + 2r)$ matrix here is full rank, by Lemma~\ref{lem:block}, we have that the lower-left $(m+t) \times (2n + 2r)$ block of $\pi P' \pi^t$ is zero. But this lower left block is precisely the block in the conclusion of the lemma, completing the proof of Lemma \ref{lem:block2}.
		\end{proof}

Now let us consider what happens within the first $\ell$ frontal slices. From the action of $R$ on the original tuple $(\tilde A_1, \dotsc, \tilde A_{\ell+1})$, and the fact that $P_0$ did not affect the first $\ell$ matrices, we have $\tilde A_i' = \sum_{i'} R_{ii'} \tilde A_i$ for $i=1,\dotsc,\ell$. Such a slice has the form
\[
\tilde A_i' = \begin{bmatrix}
	0 & A_i' \\
	-(A_i')^t & 0 \\
	& & 0_{n+m+2r+2s}
\end{bmatrix}
\]
where all empty blocks are zero, and $A_i' = \sum_{i'} R_{ii'} A_i$. Let us see how $P'$ affects such a slice. Since $(P',R)$ was a pseudo-isometry from $(\tilde A_1', \dotsc, \tilde A_{\ell}', \sigma \tilde A_{\ell+1})$ to $(\tilde B_1, \dotsc, \tilde B_{\ell+1})$, it must be the case that $P' \tilde A_i' (P')^t = \tilde B_i$ for $i=1,\dotsc,\ell$. We will now focus only on the upper $2 \times 2$ blocks, grouping the remaining blocks all together:
\begin{align*}
	\begin{bmatrix} 
		P_{11} & P_{12} & * \\
		P_{21} & P_{22} & * \\
		* & * & *
	\end{bmatrix}
	\begin{bmatrix}
		0 & A_i' \\
		-(A_i')^t & 0 \\
		& & 0_{n+m+2r+2s}
	\end{bmatrix}
	\begin{bmatrix} 
		P_{11}^t & P_{21}^t& * \\
		P_{21}^t & P_{22}^t & *  \\
		* & * & *
	\end{bmatrix}
	& =
	\begin{bmatrix}
		-P_{12} (A_i')^t & P_{11} A_i' & 0 \\
		-P_{22} (A_i')^t & P_{21} A_i' & 0 \\
		* & * & 0
	\end{bmatrix}
	\begin{bmatrix} 
		P_{11}^t & P_{21}^t& * \\
		P_{12}^t & P_{22}^t & * \\
		*& * & *
	\end{bmatrix} \\
\end{align*}
(Because of the zeros in $\tilde A_i'$, the $*$ blocks play no role.) The (1,2) entry here gives us the equation
\[
-P_{12} (A_i')^t P_{21}^t + P_{11} A_i' P_{22}^t = B_i \\
\]
However, from above we saw that $P_{21} = 0$, so we are left with
\[
P_{11} A_i' P_{22}^t = B_i.
\]
Recalling that $A_i' = \sum_{i'} R_{ii'} A_i$, the preceding equation, when applied for all $i=1,\dotsc,\ell$, gives us
\begin{equation}\label{eq:coset}
(P_{11} \tA P_{22}^t)^{R_{11}} = \tB,
\end{equation}
and thus $\tA \cong \tB$ as tensors, as desired.
\end{proof}

\begin{remark}\label{rem:TI4_coset}
We explain how to strengthen Theorem~\ref{thm:AltMatSpIsom} to give a correspondence between cosets. That is, in the case of isomorphic inputs to \TI, the coset of equivalences of tensors for \TI can be obtained from the coset of resulting alternating matrix spaces for \AMSI. 

To see this, recall that for the $\Leftarrow$ direction, $\cA$ and $\cB$ are isomorphic via $(P, R)$. First, by the rank argument, $R=\begin{bmatrix}
	R_{11} & R_{1,2}\\
	0 & \sigma
\end{bmatrix}$. Second, note that $P=P'P_0$, where $P'=	\begin{bmatrix} 
P_{11} & P_{12} & * \\
0 & P_{22} & * \\
* & * & *
\end{bmatrix}$, and $P_0= \begin{bmatrix}
\IdMat_n & \\
& \IdMat_m & -X^t/\sigma  \\
& 0 & \IdMat_n  \\
& & & \IdMat_{2r}
\end{bmatrix}$. Therefore, $P$ is also of the form $\begin{bmatrix} 
P_{11} & P_{12} & * \\
0 & P_{22} & * \\
* & * & *
\end{bmatrix}$. Then $(P_{11}, P_{22}, R_{11})$ is an equivalence between $\tA$ and $\tB$. 

From the above, we see that the coset of isometries from $\cA$ to $\cB$ consists of block upper-triangular matrices, whose diagonal blocks can be used to construct the coset of equivalences from $\tA$ to $\tB$. This in particular shows that this reduction satisfies the algebraic containment property defined in \cite[Definition 2.4]{GQ1}.

The above reasoning applies to Theorems~\ref{thm:pTI} and~\ref{thm:altform}. 
\end{remark}

We now consider the following problem which is the partitioned version of alternating matrix space isometry.
\begin{definition}[Block-diagonal alternating matrix space isometry]\label{def:block-diagonal-AMSI}
	Given $\cA, \cB\leq\Lambda(N, \F)$ and $N=n+m$, decide if there exists $T=\diag(T_1, T_2)\in\GL(N, \F)$ where $T_1\in\GL(n, \F)$ and $T_2\in\GL(m, \F)$, such that $\cA=T\cB T^t$.
\end{definition}

\begin{theorem}\label{thm:block-diagonal-AMSI}
Block-diagonal \AltMatSpIsomlong reduces to \AltMatSpIsomlong with linear blow-up. In particular, for $\vA\in \Lambda(N, \F)^\ell$ and $N=n+m$, the output is $\vA'\in\Lambda(N', \F)^{\ell+1}$ with $N'$ is at most $3n+2m+2$.
\end{theorem}
\begin{proof}[Proof sketch.] The proof follows the same strategy as that for Theorem~\ref{thm:AltMatSpIsom}, with some minor changes. 

Let $\vA=(A_1, \dots, A_\ell)\in \Lambda(N, \F)^\ell$. Suppose for $i\in[\ell]$, $A_i=\begin{bmatrix}
	A_{i,1} & A_{i,2}\\
	-A_{i,2}^t & A_{i,3}
\end{bmatrix}$, where $A_{i,1}\in \Lambda(n, \F)$. Then construct the following: 
\begin{itemize}
	\item For $i=1,\dotsc,\ell$,
	\[
	\tilde A_{i} =
	\begin{bmatrix}
			A_{i,1} & A_{i,2} & 0_{n \times n} & 0_{n \times 2r}  \\
		-A_{i,2}^t & A_{i,3}  \\
		& & 0_n \\
		& & & 0_{2r} \\
	\end{bmatrix}.
	\]
	
	\item A standard alternating slice of rank $2r + 2n$, connected to the first $n$ rows by an $\IdMat_n$ in the appropriate place:
	\begin{equation}\label{eq:gadget_amsi}
	\tilde A_{\ell+1} = 
	\begin{bmatrix}
		0_n& 0_{n \times m}  & \IdMat_n  & 0_{n \times 2r} \\
		0_{m \times n} & 0_m & \\
		-\IdMat_n &  & 0_n \\
		& & & \begin{bmatrix} 0 & \IdMat_{r} \\ -\IdMat_{r} & 0 \end{bmatrix} \\
	\end{bmatrix}.
	\end{equation}
\end{itemize}

The arguments in the proof of Theorem~\ref{thm:AltMatSpIsom} then go through, until the point for the $\Leftarrow$ direction, when the $\tilde A'_{\ell+1}$ is calculated. Here, $\tilde A'_{\ell+1}$ is of the form 
\[
\begin{bmatrix}
	Y & X & \sigma  \IdMat_n & 0  \\
	-X^T & Z & 0 & 0 \\
	-\sigma \IdMat_n & 0 & 0 & 0 \\
	0 & 0 & 0 & \sigma J_r
\end{bmatrix}.
\]
Now we observe that $\rk(\tilde A'_{\ell+1})=2r+2n+\rk(Z)$, which needs to be equal to $\rk(\tilde B_{\ell+1})=2r+2n$. It follows that $Z$ has to be the zero matrix. Therefore, $\tilde A'_{\ell+1}$ is of the form \[ \begin{bmatrix}
	Y & X & \sigma  \IdMat_n & 0  \\
	-X^T & 0 & 0 & 0 \\
	-\sigma \IdMat_n & 0 & 0 & 0 \\
	0 & 0 & 0 & \sigma J_r
\end{bmatrix}.
\]
We then need to cancel the $Y$, $X$, and $X^t$ in the above. Clearing $X$ and $X^t$ can be done in the same way as in the proof of Theorem~\ref{thm:AltMatSpIsom}, so after this, we get  
\[ \begin{bmatrix}
	Y & 0 & \sigma  \IdMat_n & 0  \\
	0 & 0 & 0 & 0 \\
	-\sigma \IdMat_n & 0 & 0 & 0 \\
	0 & 0 & 0 & \sigma J_r
\end{bmatrix}.
\]
To clear $Y$, we distinguish between the characteristic of $\F$ being $2$ or not.

When the characteristic of $\F$ is not $2$, let 
$$
P_1 := \begin{bmatrix}
	\IdMat_n & 0 & Y/2\sigma\\
	& \IdMat_m &   \\
	& 0 & \IdMat_n  \\
	& & & \IdMat_{2r}
\end{bmatrix}.
$$
Then by $Y$ being skew-symmetric ($Y^t=-Y$), we have 
\begin{eqnarray*}
	\begin{bmatrix}
		\IdMat_n & 0 & Y/2\sigma\\
		& \IdMat_m &   \\
		& 0 & \IdMat_n  \\
		& & & \IdMat_{2r}
	\end{bmatrix}
	\begin{bmatrix}
		Y & 0 & \sigma I \\
		0 & 0 & 0 \\
		-\sigma I & 0 & 0 \\
		& & & \sigma J_{2r}
	\end{bmatrix}
	\begin{bmatrix}
		\IdMat_n & 0 & 0 \\
		0 &  \IdMat_m & 0 \\
		Y^t/2\sigma &  &  \IdMat_n \\
		& & & \IdMat_{2r}
	\end{bmatrix} \\
	=
	\begin{bmatrix}
		0 & 0 & \sigma \IdMat_n \\
		0 & 0 & 0 \\
		-\sigma \IdMat_n & 0 & 0 \\
		& & & \sigma J_{2r}
	\end{bmatrix}
	= \sigma \tilde A_{\ell+1}.
\end{eqnarray*}

When the characteristic of $\F$ is $2$, let $Y_u$ be the matrix whose upper triangular part is the same as $Y$, and the lower triangular part is $0$. \jnote{I think we need to remind the reader here that in characteristic 2 $A$ is assumed to be alternating, not merely skew-symmetric, so it is (skew-)symmetric with zero diagonal, and hence the same is true of $Y$. Because if $Y$ is (skew-)symmetric, then $Y+Y_u+Y_u^t$ could have nonzero diagonal. Does this seem right to you?}
By characteristic $2$, we have $Y_u^t$ is the matrix whose lower triangular part is the same as $Y$, and the upper triangular part is $0$. In particular, $Y+Y_u+Y_u^t=0$. 
$$
P_1 := \begin{bmatrix}
	\IdMat_n & 0 & Y_u/\sigma\\
	& \IdMat_m &   \\
	& 0 & \IdMat_n  \\
	& & & \IdMat_{2r}
\end{bmatrix}.
$$
Then by $\sigma I=-\sigma I$, we have 
\begin{eqnarray*}
	\begin{bmatrix}
		\IdMat_n & 0 & Y^u/\sigma\\
		& \IdMat_m &   \\
		& 0 & \IdMat_n  \\
		& & & \IdMat_{2r}
	\end{bmatrix}
	\begin{bmatrix}
		Y & 0 & \sigma I \\
		0 & 0 & 0 \\
		-\sigma I & 0 & 0 \\
		& & & \sigma J_{2r}
	\end{bmatrix}
	\begin{bmatrix}
		\IdMat_n & 0 & 0 \\
		0 &  \IdMat_m & 0 \\
		Y_u^t/\sigma &  &  \IdMat_n \\
		& & & \IdMat_{2r}
	\end{bmatrix} \\
	=
	\begin{bmatrix}
		0 & 0 & \sigma \IdMat_n \\
		0 & 0 & 0 \\
		-\sigma \IdMat_n & 0 & 0 \\
		& & & \sigma J_{2r}
	\end{bmatrix}
	= \sigma \tilde A_{\ell+1}.
\end{eqnarray*}

After showing that $Y$ can be cleared out, the rest of the arguments go through. This completes the proof sketch.
\end{proof}

\section{Preliminaries}\label{sec:prel}

\subsection{Notation and definitions}

We collect some basic notation and definitions for future references. 

For $n\in\N$, $[n]:=\{1, 2, \dots, n\}$.

\paragraph{Linear algebra.} We use $\F^n$ to denote the linear space of length-$n$ column vectors over $\F$, $\M(n\times m, \F)$ the linear space of $n\times m$ matrices over $\F$, and $\T(n\times m\times \ell, \F)$ the linear space of $n\times m\times \ell$ 3-way arrays over $\F$. The length of $\tA\in \T(n\times m\times \ell, \F)$ is $n+m+\ell$. 

\paragraph{Groups.} Let $\GL(n, \F)$ be the general linear group over $\F^n$. A subgroup of $\GL(n, \F)$ is called a matrix group. We use $\S_n$ to denote the symmetric group consisting of permutations of $[n]$. We can view $\S_n$ naturally as a subgroup of $\GL(n, \F)$ by embedding a permutation as a permutation matrix. The monomial subgroup of $\GL(n, \F)$, denoted as $\Mon(n, \F)$, consists of invertible matrices where exists only one non-zero entry in each row and each column. The diagonal subgroup of $\GL(n, \F)$, denoted as $\diag(n, \F)$, consists of invertible diagonal matrices. 

\paragraph{Tensors and slices.} Given $\tA=(a_{i,j,k})\in \T(\ell\times m\times n, \F)$, we can slice $\tA$ along one direction and obtain several matrices, which are called slices. For example, slicing along the third coordinate, we obtain the \emph{frontal} slices, namely $\ell$ matrices $A_1, \dots, A_\ell\in \M(m\times n, \F)$, where $A_k(i,j)=a_{i,j,k}$. Similarly, we also obtain the \emph{horizontal} slices by slicing along the first coordinate, and the \emph{lateral} slices by slicing along the second coordinate. 

\paragraph{Matrices and matrix spaces.} Let $A, B\in\M(n, \F)$. Then $A$ and $B$ are equivalent, if there exist $L, R\in\GL(n, \F)$, such that $A=LBR$. They are conjugate if there exists $L\in\GL(n, \F)$ such that $A=L^{-1}BL$, and congruent if there exists $L\in\GL(n, \F)$ such that $A=L^tBL$. The transpose of $A$ is denoted by $A^t$. A matrix $A\in \M(n, \F)$ is symmetric if $A=A^t$, skew-symmetric if $A=-A^t$, and alternating if for any $u\in\F^n$, $u^tAu=0$. The linear space of symmetric matrices is denoted by $\S(n, \F)$, and the linear space of alternating matrices is denoted by $\Lambda(n, \F)$. 

A matrix space is a linear subspace of $\M(n, \F)$. These three equivalences of matrices extend naturally to matrix spaces. That is, let $\cA, \cB\leq\M(n, \F)$ be two matrix spaces. Then $\cA$ and $\cB$ are equivalent if there exist $L, R\in\GL(n, \F)$ such that $\cA=L\cB R:=\{LBR \mid B\in\cB\}$. One can define matrix space conjugacy and matrix space congruence similarly. In the literature matrix space congruence is also known as matrix space isometry \cite{LQ17,Sun23}.

\paragraph{Symmetric, alternating, anti-symmetric trilinear forms.} Let $U\cong\F^n$, and $f:U\times U\times U\to\F$ be a trilinear form. We say that $f$ is \emph{symmetric}, if for any $u_1, u_2, u_3\in U$, and any $\sigma\in\S_3$, $f(u_1, u_2, u_3)=f(u_{\sigma(1)}, u_{\sigma(2)}, u_{\sigma(3)})$. It is anti-symmetric, if for any $u_1, u_2, u_3\in U$, and any $\sigma\in\S_3$, $f(u_1, u_2, u_3)=\sgn(\sigma)\cdot f(u_{\sigma(1)}, u_{\sigma(2)}, u_{\sigma(3)})$. It is alternating, if for any $u, v\in \F^n$, $f(u, u, v)=f(u, v, u)=f(v, u, u)=0$. Note that alternating implies anti-symmetric, but not vice versa over characteristic-$2$ fields. 

\paragraph{Algebras.} An algebra $A$ is a vector space $U$ with a bilinear map $\circ: U\times U\to U$. It is associative if for any $u, v, w \in U$, $(u\circ v)\circ w=u\circ (v\circ w)$. It is commutative if for any $u, v\in U$, $u\circ v=v\circ u$. It is a Lie algebra, if it satisfies the Jacobi identity. In algorithms, an algebra $A$ is usually represented by a linear basis $b_1, \dots, b_n$ of the underlying vector space $U$, together with the structure constants $\tA=(a_{i, j, k})\in\T(n\times n\times n, \F)$ are those field elements satisfying $b_i\cdot b_j=\sum_k a_{i,j,k}b_k$. 

\subsection{Previous results and their lengths}
Here we gather previous reductions and highlight their lengths. The key ones we are aware of that are \emph{not} already linear size are:

\begin{theorem}[Super-linear-size reductions]\label{thm:previous_quadratic}
\begin{enumerate}
\item \TI $\leq$ \AltMatSpIsomlong with quadratic size ($\ell \times n \times m$ tensors with $\ell \leq n$ go to alternating matrix spaces of dimension $m + \ell(2n+1)+n(4n+2)$, for matrices of size $\ell + 7n + 3$) \cite[Prop.~5.1]{GQ1}.

\item \algprobm{Partitioned 3-Tensor Isomorphism} reduces to \TI with quadratic size ($n \times m \times t$ tensors with $\overline{t}$ blocks in the third direction go to tensors of size $(\Theta(2^{\overline{t}} r)+n) \times (\Theta(2^{\overline{t}} t r) + m) \times t$, where $r=\min\{n,m\}+1$) \cite[Theorem~2.1]{FGS19}.

\item \AltMatSpIsomlong $\leq$ \algprobm{Alternating Trilinear Form Equivalence} in quadratic size, and similar for symmetric matrix spaces and symmetric trilinear forms \cite[Prop.~4.1]{GQT}.

\item Isomorphism of $n$-dimensional commutative, unital, associative algebras reduces to \CubicFormlong for cubic forms on $\Theta(n^6)$ variables \cite[Thm.~4.1]{AS06}.\footnote{For the size estimate, see the equation between equations (12) and (13) in \cite{AS06}. They produce a local algebra on $\binom{n+1}{2} + n$ generators, whose dimension is cubic in the number of generators.}
\end{enumerate}
\end{theorem}
We note that \algprobm{Partitioned 3TI} is usually used where the number of parts is $O(1)$, so the $2^{\overline{t}}$ factors are not so concerning (nor do we know how to avoid it).
\ynote{Please check.} The quadratic factor $tr$ in the middle side length is what we aim to improve to linear. (And when one applies this reduction to remove partitions in each of the three directions, this will then show up in all three side lengths.)

Essentially all the other reductions around \TI-complete problems that we are aware of are already linear size:

\begin{theorem}[Linear-size reductions]\label{thm:previous_linear}
	The following reductions all have only linear size increase.
\begin{enumerate}
\item \TI $\leq$ \MatSpConjlong \cite[Prop.~6.1]{GQ1}

\item \TI $\leq$ \MatLieConjlong \cite[Cor.~6.2]{GQ1}

\item \TI $\leq$ \algprobm{Associative Matrix Algebra Conjugacy} \cite[Cor.~6.2]{GQ1}

\item \MatSpIsomlong $\leq$ \AlgIsolong \cite[Prop.~6.3]{GQ1}

\item \MatSpIsomlong $\leq$ \NcCubicFormlong \cite[Prop.~6.3]{GQ1}

\item \MatSpIsomlong $\leq$ \algprobm{Associative Algebra Isomorphism} \cite[Cor.~6.5]{GQ1}

\item \MatSpIsomlong $\leq$ \algprobm{Lie Algebra Isomorphism} \cite[Cor.~6.5]{GQ1}

\item \algprobm{Linked Partitioned Isomorphism of Degree-3 Tensor Networks} $\leq$ \algprobm{Partitioned 3-Tensor Isomorphism} \cite[Thm.~4.1]{FGS19}. 

\item \algprobm{Monomial Code Equivalence} $\leq$ \TI \cite[Prop.~7]{GQ2}.

\item \algprobm{Permutational Code Equivalence} $\leq$ \algprobm{Diagonalizable Matrix Lie Algebra Conjugacy} \cite[Thm.~II.1]{GrochowLie}.

\item \GIlong $\leq$ \algprobm{Code Equivalence} (monomial or permutational) \cite{PR} (see also \cite[Lem.~II.4]{GrochowLie}). Here the output size is $O(n+m)$, where $n$ is the number of vertices and $m$ is the number of edges of $G$.

\item \GIlong $\leq$ \algprobm{Semisimple Matrix Lie Algebra Conjugacy} \cite[Lem.~IV.5]{GrochowLie}. Here the output 3-way array has size $3n \times 9m \times 9m$.
\end{enumerate}
\end{theorem}

\begin{remark}\label{rem:TI1_coset}
In \cite[Remark 2.5]{GQ1}, the above results were stated for their decision versions. In the case of isomorphic inputs, the reductions in Theorem~\ref{thm:previous_linear} (1-7) can be strengthened to computing the cosets. To see this, take one reduction from Theorem~\ref{thm:previous_linear} (1-7), such as \TI $\leq$ \MatSpConjlong \cite[Prop.~6.1]{GQ1}. 

The reduction goes by first reducing to certain non-degenerate tensors. Then, in the case of non-degenerate tensors, the reduction produces matrix spaces whose conjugation matrices must be in the upper-block triangular form, whose diagonal blocks are equivalence matrices for the original tensors. This gives a correspondence between cosets of conjugation matrices for \MatSpConjlong, and cosets of equivalence matrices for \TI, in the non-degenerate case. 

In the case of degenerate tensors, there is an easy correspondence between cosets of equivalence matrices of the degenerate tensors, and cosets of equivalence matrices of their non-degenerate parts, which can be seen from \cite[Observation 2.2]{GQ1}. This then covers all the cases. 
\end{remark}



\section{New gadgets}\label{sec:gadget}

In order to make our gadgets more modular and easy to reason about, we introduce them in a series of lemmas which have the form, informally, ``If you're looking at a set $G$ of potential isomorphisms between two tensors $A,B$, and we extend $A,B$ with certain gadgets, then that restricts the potential isomorphisms to lie in $G \cap H$'', where $H$ is a set of isomorphisms with desired special properties. This enables us to apply multiple gadgets in succession and reason about them by repeatedly adding structure to the subcoset of allowed isomorphisms. We make this formal in the remainder of this section.

\subsection{Gadget restricting to a partition}

\begin{lemma}[Individualizing by rank] \label{lem:indiv}
Let $r \geq 0$. Let $T_1, \dotsc, T_N$ be matrices such that $\rk(T_I) = r 2^{I-1}$ for each $I=1,\dotsc,N$. 

Let $\tA$ be a 3-way array with the following frontal slices:
\begin{itemize} 
\item $A_1,\dotsc,A_\ell$ are of the form:
\[
\left[\begin{array}{cc;{2pt/2pt}ccccc}
*_{n \times m} & * & * & * & \cdots & * \\
* & 0 &  0 & 0  & \cdots & 0\\ \hdashline[2pt/2pt]
* & 0 & 0 \\
* & 0 & & 0 \\
\vdots & \vdots & & & \ddots \\
* & 0 & & & & 0
\end{array}\right],
\]
where each entry represents a block, $*$'s represent arbitrary blocks (one of whose sizes is indicated), $n+m < r$, and $0$ or blank represents a block of zeroes; 

\item For $I=1,\dotsc,N$, the $\ell+I$-th slice $A_{\ell+I}$ is of the form
\[
\left[\begin{array}{cc;{2pt/2pt}ccccc}
* & * &  * & \cdots & * & \cdots & *\\
* & 0 &  0 & \cdots & 0 & \cdots & 0\\ \hdashline[2pt/2pt]
* & 0  & 0 \\
\vdots & \vdots & & \ddots \\
* & 0 & & & T_I \\
\vdots & \vdots & & & & \ddots \\
* & 0 & & & & & 0
\end{array}\right],
\]
(with the block sizes are the same as those in the first form, where $T_i$ appears on the $(2+I)$-th diagonal block).
\end{itemize}
Let $\tA'$ be another 3-way array with the same description as above (but possibly different values filled in for the $*$ blocks).

If $(P,Q,R) \cdot \tA = \tA'$, then $R$ must have the form $R = \begin{bmatrix} R_{11} & 0 \\ R_{21} & R_{22} \end{bmatrix}$ where $R_{11}$ is $\ell \times \ell$, and $R_{22}$ is $N \times N$ and diagonal.
\end{lemma}

\begin{proof}
First, we show that the upper-right block of $R$ must be zero. To see this, note that that block being nonzero means adding some multiples of $\tilde A_{\ell + I}$ with $I \geq 1$ to (some of) the first $\ell$ frontal slices. But because the $T_I$'s in the lower-right block appear in positions that are zero in the first $\ell$ slices, such a linear combination would have rank at least $\min\{\rk(T_1), \dotsc, \rk(T_N)\} = \rk(T_1) = r$. But this is strictly greater than $n+m$, which in turn is an upper bound on the rank of any of the first $\ell$ slices of $\tA'$ because they are supported on a union of $n$ rows and $m$ columns. Thus the upper-right block of $R$ must be zero.

Next, we show that $R_{22}$ is diagonal. The effect of $R_{22}$ is to take linear combinations of $\{A_{\ell + 1}, \dotsc, A_{\ell+N}\}$. Now consider 
\[
\hat A_{\ell+I} := \sum_{i=1}^\ell (R_{21})_{I,i}  A_i + \sum_{I'=1}^N (R_{22})_{II'}  A_{\ell+I'}.
\]
Since $(P,Q,R)$ was an isomorphism, and the actions of $P$, $Q$, and $R$ commute with one another, we have $P  \hat{A}_{\ell+I} Q^t = A'_{\ell+I}$, and therefore $\rk \hat{A}_{\ell+I} = \rk A'_{\ell+I}$. Because the first block-row has height $n$ and the first block-column has width $m$, we have that $\rk A'_{\ell+I}$ must be in the range 
\begin{equation} \label{eq:range}
\left[\rk(T_I), \rk(T_I) + n + m\right]. 
\end{equation}
Let $S = \{I' \in [N] : (R_{22})_{II'} \neq 0\}$ be the support of the $I$-th row of $R_{22}$. 
Then we have that $\rk(\hat{A}_{\ell+I})$ lies in the range 
\begin{equation} \label{eq:range2}
\left[\sum_{I \in S} \rk(T_I), n+m+\sum_{I \in S} \rk(T_I)\right]. 
\end{equation}
For $\rk(\hat{A}_{\ell+I})$ to equal $\rk A'_{\ell+I}$, the range (\ref{eq:range2}) must thus overlap with the range (\ref{eq:range}).

Since $\rk(T_{I'}) > 2\rk(T_I) > n+m+\rk(T_I)$ for $I' > I$, we must have $S \subseteq [I]$ (that is, $R_{22}$ is lower triangular). If $S$ includes $I$ itself, it cannot include any other indices without going over the allowed range (\ref{eq:range}), and then we would be done (for $S = \{I\}$ would mean that $R_{22}$ is diagonal). The only remaining possibility is if $S$ is a subset of $[I-1]$. But even if it were all of $[I-1]$, we would then have that the rank is at most 
\begin{align*}
n + m + \sum_{t=1}^{I-1} \rk(T_t) & = n + m + \sum_{t=1}^{I-1} r 2^{t-1} \\
 & < r + \sum_{t=1}^{I-1} r 2^{t-1} \\
 & = r2^I = \rk(T_I).
\end{align*}
Because of the strict inequality on the second line, we conclude that no subset of the first $I-1$ slices can have the ranks add up to the necessary rank, so we have $S = \{I\}$, and hence $R_{22}$ is diagonal. This completes the proof of the lemma.
\end{proof}


\subsection{Cancelling ``errors'' in the gadget}\label{subsec:cancel}
The purpose of this subsection is to formalise the ``cancellation'' strategy employed in the latter half of the proof of Theorem~\ref{thm:AltMatSpIsom}. At present, the materials in this subsection are not used in other parts of the paper, but it is our hope that this will be useful in future works. 

First, we introduce the following notions.

\begin{definition}[Gadget, attaching a gadget]
We call a 3-way array $\Gamma$ of size $(n+n') \times (m+m') \times \ell'$ a \emph{$(n,m) \oplus (n',m',\ell')$ gadget} if the upper-left $n \times m \times \ell'$ block is zero.

Given an $n \times m \times \ell$ tensor $\tA$, and $n',m' \geq 0$, let $\tA^0$ denote the extension of $\tA$ to an $(n + n') \times (m + m') \times \ell$ tensor by padding it with zero. Let $\tA \bowtie \Gamma:=\tA^0 \boxplus_{1,2} \Gamma$ denote the $(n + n') \times (m + m') \times (\ell + \ell')$ tensor gotten from $\tA^0$ by appending the frontal slices of $\Gamma$ to those of $\tA^0$.
%
%
%
\end{definition}

(There is a more general, but more complicated, definition that would allow attaching in all three directions at once rather than just two directions, which would encompass more of the gadgets from previous papers \cite{FGS19, GQ1, GQ2, GQT}. But as we won't need it for the reductions in this paper, we avoid the additional complications.)

\begin{definition}[Gadget cancellation property]
Let $\mathcal{E} \leq M(n \times m)$ be a linear subspace of matrices, and let  $G \leq \GL_{n+n'} \times \GL_{m+m'} \times 1$ be a subgroup. Let $\mathcal{T}_{\mathcal{E}}$ denote the subspace of $n \times m \times \ell'$ tensors all of whose frontal slices lie in $\mathcal{E}$. For $T \in \mathcal{T}_{\mathcal{E}}$, and an $(n,m) \oplus (n',m',\ell')$ gadget $\Gamma$, we use $T \boxplus \Gamma$ to denote the $(n + n') \times (m + m') \times \ell'$ 3-way array that is $T$ in the upper-left $n \times m \times \ell'$ block and agrees with $\Gamma$ outside that block. We say that $\Gamma$ \emph{admits $G$-cancellation of $\mathcal{E}$ errors} if, for any $T \in \mathcal{T}_{\mathcal{E}}$, there is a $g \in G$ such that 
\begin{enumerate}
\item $g \cdot (T \boxplus \Gamma) = \Gamma$, that is, $g$ has the effect of ``zero-ing out'' the $T$ in the upper-left $n \times m \times \ell'$ block, and

\item $g$ has the form $\left(\begin{bmatrix} \IdMat_n & * \\ 0 & * \end{bmatrix}, \begin{bmatrix} \IdMat_m & * \\ 0 & * \end{bmatrix}, 1\right)$ (where different $*$'s can be different).
\end{enumerate}
\end{definition}

\begin{example}
The extra slice $\tilde A_{\ell+1}$ in the proof of Theorem~\ref{thm:AltMatSpIsom} is an $(n+m, n+m) \oplus (r,r,1)$ gadget; call this gadget $\Gamma$. Let $\mathcal{E}$ be the subspace of $(n+m) \times (n+m)$ matrices of the form $\begin{bmatrix} 0 & A \\ -A^t & 0 \end{bmatrix}$, and let $G \leq \GL_{n+m+r} \times \GL_{n+m+r} \times 1$ be the subgroup $\{ (P, P^t, 1) : P \in \GL_{n+m} \}$. Then the part of the proof of Theorem~\ref{thm:AltMatSpIsom} that introduces the matrix $P_0$ shows that $\Gamma$ admits $G$-cancellation of $\mathcal{E}$ errors.
\end{example}

\begin{lemma}[Cancellation lemma]\label{lem:cancel}
Let $\mathcal{E} \leq \M(n \times m, \F)$ be a subspace of matrices, and let $G \leq \GL_{n+n'} \times \GL_{m+m'} \times 1$ be a subgroup. Suppose $\Gamma$ is an $(n,m) \oplus (n',m',\ell')$ gadget that admits $G$-cancellation of $\mathcal{E}$ errors. 

If $\tA,\tB$ are two $n \times m \times \ell$ tensors all of whose frontal slices lie in $\mathcal{E}$ such that  $\tA \bowtie \Gamma$ is isomorphic to $\tB \bowtie \Gamma$ via $(P,Q,R)$ where $R = \begin{bmatrix} R_{11} & 0 \\ R_{21} & R_{22} \end{bmatrix}$ where $R_{11}$ is $\ell \times \ell$, then $\tA^0$ is isomorphic to $\tB^0$ by some $(P',Q',R_{11}) \in (P,Q,R_{11}) G$ such that $(P',Q',R_{22}) \cdot \Gamma = \Gamma$.
\end{lemma}

Here, $(P,Q,R)G$ denotes the left coset $\{(P,Q,R) g : g \in G\}$.

\begin{proof}
Suppose $\tA \bowtie \Gamma$ is isomorphic to $\tB \bowtie \Gamma$ via $(P,Q,R)$ satisfying the assumptions of the lemma. Since the actions of $P$, $Q$, and $R$ commute with one another, we have that $(P,Q,1) \cdot (\tA \bowtie \Gamma)^R = \tB \bowtie \Gamma$. Since $R$ is lower-triangular by assumption, we have that the first $\ell$ frontal slices of $(\tA \bowtie \Gamma)^R$ are precisely $\tA^0$. Thus $(P,Q,1) \cdot \tA^0 = \tB^0$.

Next, the last $\ell'$ slices of $(\tA \bowtie \Gamma)^R$ may be of the form $E \boxplus \Gamma$, where $E$ is a linear combination (according to the lower-left block of $R$) of the frontal slice of $\tA$. By assumption of the cancellation property, there is $g \in G$ such that $g \cdot E \boxplus \Gamma = \Gamma$. 

Let $(P',Q',1) = (P,Q,1) g$ (recall that $g$ must have its third coordinate 1 by definition of the cancellation property). Since $R$ is lower-triangular, the first $\ell$ frontal slices of $(\tA \bowtie \Gamma)^R$ are precisely $\tA^0$. Since $g$ is upper-triangular by definition of the cancellation property, and $\tA^0$ is zero in its lower blocks and the identity in its upper-left block, we have $g \cdot \tA^0 = \tA^0$. Thus we have that $(P',Q',1) \cdot \tA^0 = (P,Q,1) g \cdot \tA^0 = (P,Q,1)$, and thus $(P',Q',R_{11}) \cdot \tA^0 = \tB^0$.
\end{proof}

\subsection{Restricting to the automorphism group of an arbitrary tensor}\label{subsec:framework}

Let $A$ be an $n \times m \times \ell$ tensor and $T$ an $n \times m' \times \ell'$ tensor. Let $A \boxplus_{1} T$ denote the $n \times (m + m') \times (\ell + \ell')$ partitioned tensor,
whose first $\ell$ frontal slices are $\begin{bmatrix} A_i & 0 \end{bmatrix}$, where $A_i$ are the frontal slices of $A$, and whose last $\ell'$ frontal slices are $\begin{bmatrix} 0 & T_i \end{bmatrix}$, where $T_i$ are the frontal slices of $T$. Note that $A \boxplus_{1} T$ naturally comes with a partition on the second direction (into the first $m$ vertical slices and the last $m'$ slices), and a partition on the third direction (into the first $\ell$ frontal slices, and the last $\ell'$ frontal slices).

\textit{Notational conventions.} For two $n \times m \times \ell$ tensors $\tA,\tB$, and a subset $G \subseteq \GL_n \times \GL_m \times \GL_\ell$, we use $\Iso_G(\tA,\tB)$ to denote $\Iso(\tA,\tB) \cap G$, and refer to such isomorphisms from $\tA$ to $\tB$ as $G$-isomorphisms. If $\Iso_G(\tA,\tB)$ is nonempty, we say $\tA$ and $\tB$ are $G$-isomorphic. This allows us to phrase general gadget lemmas for many different actions simultaneously; for example, when $n=m$ and $G = \{(g,g,h)\}$, then $\tA$ and $\tB$ are $G$-isomorphic if and only if they represent isometric matrix spaces. When speaking of automorphisms of partitioned 3-tensors, we may write things like $\GL_n \times (\GL_m \oplus \GL_{m'}) \times (\GL_\ell \oplus \GL_\ell')$. As a group, $\GL_m \oplus \GL_{m'}$ is isomorphic to the direct product $\GL_m \times \GL_{m'}$; we use the notation this way to keep track of which factors correspond to partitioned blocks versus which factors correspond to the three legs of the tensor.

We first present a restricted version of the automorphism gadget lemma, as the proof is nearly identical to the full version but there are many fewer indices to keep track of.

\begin{lemma}[Automorphism gadget lemma, simple version]
Let $G \leq \GL_n \times \GL_m \times \GL_\ell$ and let $\hat{G} = \{(g_1, g_2 \oplus h_2, g_3 \oplus h_3) : (g_1,g_2,g_3) \in G, h_2 \in \GL_{m'}, h_3 \in \GL_{\ell'}\}$ be its extension. Suppose $\Gamma$ is an $n \times m' \times \ell'$ tensor, and let $H = \{g_1 \in \GL_n : (\exists h_2, h_3)[(g_1,h_2,h_3) \cdot \Gamma = \Gamma]\}$. Then for any two $n \times m \times \ell$ tensors $\tA,\tB$, we have that
\[
\tA \text{ and } \tB \text{ are $(G \cap (H \times \GL_{m} \times \GL_\ell))$-isomorphic } \Longleftrightarrow \tA \boxplus_{1} \Gamma \text{ and } \tB \boxplus_{1} \Gamma \text{ are $\hat{G}$-isomorphic}.
\]
Furthermore, we have:
\[
\Iso_{\hat{G}}(\tA \boxplus_{1} \Gamma, \tB \boxplus_{1} \Gamma) = \{(g_1, g_2 \oplus h_2, g_3 \oplus h_3) : (g_1, g_2, g_3) \in \Iso_G(\tA,\tB) \text{ and } (g_1, h_2, h_3) \in \aut(\Gamma)\}.
\]
\end{lemma}

The key idea here that makes $\Gamma$ into a useful gadget is the first part of the statement; namely, that we can effectively restrict the first coordinate of the isomorphism we wish to consider between $\tA,\tB$ to lie in $H$. And, by mimicking the same construction in the other two directions, we can effectively restrict all three coordinates of isomorphisms to lie in the automorphism group of three different gadgets $T_1, T_2, T_3$.

\begin{proof}
We prove the ``furthermore'' part, as then the first statement follows immediately. By definition, $(g_1, g_2 \oplus h_2, g_3 \oplus h_3)$ is in $\hat{G}$ if and only if $(g_1, g_2, g_3)$ is in $G$ and $h_2,h_3$ are invertible. Since the part of $(g_1, g_2 \oplus h_2, g_3 \oplus h_3)$ that acts on the first $n \times m \times \ell$ block is precisely $(g_1,g_2,g_3)$, and the part that acts on the second $n \times m' \times \ell'$ block is precisely $(g_1,h_2,h_3)$, we have that $(g_1, g_2 \oplus h_2, g_3 \oplus h_3)$ is an isomorphism $\tA \boxplus_{1} T$ to $\tB \boxplus_{1} T$ if and only if $(g_1,g_2,g_3)$ is an isomorphism from $\tA$ to $\tB$ and $(g_1,h_2,h_3)$ is an isomorphism from $T$ to itself. This completes the proof.
\end{proof}


\begin{lemma}[Automorphism gadget lemma, full version] \label{lem:aut}
Let:
\begin{itemize}
\item $n = \sum_{I \in [N]} n_I, m = \sum_{J \in [M]} m_J, \ell = \sum_{K \in [L]} \ell_K$; 

\item $I_0 \in [N], J_0 \in [M], K_0 \in [L]$;

\item $\Gamma \in T(n_{I_0} \times m_{J_0} \times \ell_{K_0})$; 

\item $H = \{h_1 \in \GL_{n_{I_0}} : (\exists h_2 \in \GL_{m_{J_0}}, h_3 \in \GL_{\ell_{K_0}})[(h_1, h_2, h_3) \cdot \Gamma = \Gamma]\}$.

\item $G \leq \left(\bigoplus_{I \in [N]} \GL_{n_I}\right) \times \left(\bigoplus_{J \in [M] \backslash \{J_0\}} \GL_{m_J}\right) \times \left(\bigoplus_{K \in [L] \backslash \{K_0\}} \GL_{\ell_K} \right)$;

\item  $\hat{G} \leq \left(\bigoplus_{I \in [N]} \GL_{n_I}\right) \times \left(\bigoplus_{J \in [M]} \GL_{m_J}\right) \times \left(\bigoplus_{K \in [L]} \GL_{\ell_K} \right)$ be its extension; that is, $\hat{G}$ consists of those elements whose projection to $\left(\bigoplus_{I \in [N] \backslash \{I_0\}} \GL_{n_I}\right) \times \left(\bigoplus_{J \in [M] \backslash \{J_0\}} \GL_{m_J}\right) \times \left(\bigoplus_{K \in [L] \backslash \{K_0\}} \GL_{\ell_K} \right)$ lie in $G$. 
\end{itemize}

Suppose $\hat{\tA},\hat{\tB}$ are two partitioned 3-tensors such that (1) the $(I_0,J_0,K_0)$ blocks of both are equal to $\Gamma$, (2) in both tensors, all other blocks at indices $(I_0,J_0,*)$, $(I_0,*,K_0)$, and $(*,J_0,K_0)$ are zero. Let $\tA$ be the same as $\hat{\tA}$ but with all blocks at indices $(I_0,J_0,*), (I_0,*,K_0), (*,J_0,K_0)$ removed (including the $(I_0,J_0,K_0)$ block); similarly for $\tB$. Then
\begin{multline*}
	\tA \text{ and } \tB \text{ are $\left(G \cap \left(\left(H \oplus \bigoplus_{I \in [N]  \backslash \{I_0\}} \GL_{n_I} \right)\times \GL_{m} \times \GL_\ell\right)\right)$-isomorphic }
	\\ 
	\Longleftrightarrow \hat{\tA} \text{ and } \hat{\tB} \text{ are $\hat{G}$-isomorphic}.
\end{multline*}
Furthermore, we have:
\[
\Iso_{\hat{G}}(\hat{\tA} , \hat{\tB}) = \{(g_1, g_2 \oplus h_2, g_3 \oplus h_3) : (g_1, g_2, g_3) \in \Iso_G(\tA,\tB) \text{ and } (\pi_{I_0}(g_1), h_2, h_3) \in \aut(\Gamma)\},
\]
where $\pi_{I_0}(g_1)$ denotes the projection of $g_1$ onto the $I_0$ summand $\GL_{n_{I_0}}$.
\end{lemma}

\begin{proof}
The proof is essentially identical to the simple version above, because of the presence of zero blocks in the same block-row, block-column, and block-depth as the block $(I_0, J_0, K_0)$ that contains the gadget $\Gamma$.
\end{proof}

\subsection{Gadget restricting to monomial transformations}
We can leverage the Automorphism Gadget Lemma (\ref{lem:aut}) to restrict to the monomial subgroup. A monomial matrix is an invertible matrix supported on a permutation, or equivalently, is the product of a permutation matrix and an invertible diagonal matrix. For a permutation $\pi$ we denote its corresponding permutation matrix by $P_\pi$. The monomial matrices form a subgroup of $\GL_n$, which we denote $\Mon_n$. The gadget tensor we use which has $\Mon_n$ as the first coordinate of its automorphisms (playing the role of $H$ in the Automorphism Gadget Lemma) is the so-called unit tensor. This is an $n \times n \times n$ tensor $\Gamma$ whose 3-way array has $\Gamma_{iii} = 1$ for $i=1,\dotsc,n$, and all other entries zero. The key result was proved by B\"{u}rgisser and Ikenmeyer:

\begin{proposition}[{B\"{u}rgisser and Ikenmeyer \cite[Prop.~4.1]{BI}}]
The automorphism group of the $n$-th unit tensor is $\{(P_\pi d_1, P_\pi d_2, P_\pi d_3) : \pi \in S_n, d_1, d_2, d_3 \text{ diagonal matrices such that } d_1 d_2 d_3 = \IdMat_n\}$. In particular, the projection onto the first coordinate of this automorphism group is $\Mon_n$.
\end{proposition}

\begin{proof}
They showed the statement about the automorphism group. To see the ``in particular'', note that for any monomial matrix $P_{\pi} d_1$ which we wish to see in the first coordinate, we may fill the second two coordinates with $P_{\pi} d_1^{-1}$ and $P_{\pi}$, respectively, to get an element of the automorphism group.
\end{proof}


\subsection{Gadget restricting to diagonal transformations}\label{subsec:diagonal}

We would like to construct a 3-way array $\tA\in \T(n\times m\times\ell, \F)$, where $m, \ell=O(n)$, such that $\aut(T)=\{(R, S, T)\in \GL(n, \F)\times\GL(m, \F)\times\GL(\ell, \F)\mid R, S, T\text{ are diagonal}\}$. 

Recall that a graph is asymmetric if its automorphism group is trivial. We show that from an asymmetric regular bipartite graph $G=(L\cup R, E)$ where $L=R=[n]$ and $|E|=\ell$, there is a naturally associated tensor $\tA_G\in \T(n\times n\times \ell, \F)$, such that $\aut(\tA_G)\subseteq \diag(n, \F)\times\diag(n, \F)\times\diag(\ell, \F)$.

For $(i, j)\in [n]\times [n]$, let $E_{i,j}\in \M(n, \F)$ be the elementary matrix where the $(i, j)$th entry is $1$, and the other entries are $0$. Let $G=(L\cup R, F)$ be a bipartite graph, where $L=R=[n]$, and $F=\{e_1, \dots, e_\ell\}\subseteq L\times R=[n]\times [n]$. Then construct $\tA_G\in\T(n\times n\times \ell, \F)$ where the $i$th frontal slice is $E_{e_i}$.

The proof of the following lemma is inspired by that of \cite[Proposition 6.2]{LQWWZ}.
\begin{lemma}\label{lem:diag}
Let $G=(L\cup R, F)$ and $\tA_G\in \T(n\times n\times \ell, \F)$ be as above. Suppose $G$ is regular. Then $G$ is asymmetric if and only if $\aut(\tA_G)\subseteq\diag(n, \F)\times\diag(n, \F)\times\diag(\ell, \F)$. 
Furthermore, in this case, every invertible diagonal matrix occurs in the first coordinate of some automorphism of $\tA_G$.
\end{lemma}
\begin{proof}
The if direction is trivial. For the only if direction, let $(R, S, T)\in \aut(\tA_G)$, where $R=(r_{i,j})$, $S=(s_{i,j})$, and $T=(t_{i,j})$. Note that 
\begin{equation}\label{eq:RES}
	RE_{a, b}S^t=\begin{bmatrix}
	r_{1, a}\\
	r_{2,a}\\
	\vdots\\
	r_{n,a}
\end{bmatrix}\begin{bmatrix}
s_{1, b} & s_{2, b} & \dots & s_{n, b}
\end{bmatrix}.
\end{equation} By the structure of $\tA_G$, if $(a, b)\in F$ and $r_{c, a}s_{d, b}\neq 0$, then $(c, d)\in F$. 

Consider $(\sigma, \tau)\in\S_n\times\S_n$, such that $\forall i\in[n]$, $r_{\sigma(i), i}\neq 0$ and $\forall j\in[n]$, $s_{\tau(j), j}\neq 0$. Then we get $\forall i, j\in[n]$, $r_{\sigma(i), i}s_{\tau(j), j}\neq 0$. This implies that for $(i, j)\in F$, $(\sigma(i), \tau(j))\in F$, so $(\sigma, \tau)$ induces an automorphism of $G$. By $G$ being asymmetric, $\sigma$ and $\tau$ can only be the identity permutation.

Therefore, $R=D_1+R_1$ where $D_1\in\GL(n, \F)$ is diagonal and $R_1=R-D_1$. We need to show that $R_1$ is all-zero. If not, suppose $R_1(c, a)$, $c\neq a$, is non-zero. By Equation~\ref{eq:RES}, we see that $(a, b)\in F\Rightarrow (c, b)\in F$. As $G$ is regular, the neighbourhoods of $a$ and $c$ are the same. This means that switching $a$ and $c$ is a non-trivial automorphism, contradicting that $G$ is asymmetric. This shows that $R_1$ is all-zero, so $R$ is diagonal. Similarly, it can be shown that $S$ must be diagonal, and $R$ and $S$ being diagonal implies that $T$ is diagonal too. This concludes the proof 
of the first part.

To prove the furthermore, let $R$ be any invertible diagonal matrix. Let $S = I_n$. For each $i=1,\dotsc,\ell$, if $e_i = (a,b)$, then set $t_{ii} = 1 / r_{aa}$. Set all the off-diagonal entries of $T$ to zero. It is then readily verified that $(R,S,T) \in \aut(\tA_G)$.
\end{proof}

While we won't need it, we remark that in fact any pair of invertible diagonal matrices $(R,S)$ occurs as the first two coordinates of an automorphism of $\aut(\tA_G)$, by a nearly identical construction with $t_{ii} = 1/(r_{aa} s_{bb})$. 

\begin{proposition}
There is a randomized Las-Vegas polynomial-time algorithm $P$ that, given $[n]$ in unary, outputs $\tA\in\T(n\times n\times \ell, \F)$ where $\ell=O(n)$, such that $\aut(\tA)\subseteq\diag(n, F)\times\diag(n, \F)\times\diag(n, \F)$.
\end{proposition}
\begin{proof}
We sample a $3$-regular bipartite graph $G=(L\cup R, F)$ where $|L|=|R|=n$ and set $\tA:=\tA_G$. It is known that random $k$-regular bipartite graphs can be uniformly sampled \cite{Wor99}, and a random $k$-regular bipartite graph is asymmetric with high probability \cite{MW84}. By Lemma~\ref{lem:diag}, $\aut(\tA_G)\subseteq \diag(n, \F)\times\diag(n, \F)\times\diag(n, \F)$. To verify whether $\aut(\tA_G)$ is trivial, we can use the polynomial-time algorithm of Luks that computes the automorphism group of a bounded-degree graph \cite{Luks82}.
\end{proof}



\section{Partitioned tensor isomorphism}\label{sec:partition}

To start with, let us define partitioned tensor isomorphism in full generality. 
\begin{quotation}
\noindent \algprobm{Partitioned Tensor Isomorphism} (over a field $\F$)

\noindent \textit{Input:} Two 3-way arrays $\tA, \tB$ over $\F$ of size $n \times m \times \ell$, and numbers $n_1, \dotsc, n_N$, $m_1, \dotsc, m_M$, $\ell_1, \dotsc, \ell_L$ such that $\sum n_i = n$, $\sum m_i = m$, $\sum \ell_i = \ell$.

\noindent \textit{Decide:} Do there exist invertible matrices $P_1 \in \GL_{n_1}(\F), \dotsc, P_N \in \GL_{n_N}(\F), Q_1 \in \GL_{m_1}(\F), \dotsc, Q_M \in \GL_{m_M}(\F), R_1 \in \GL_{\ell_1}(\F), \dotsc, R_L \in \GL_{\ell_L}(\F)$ such that $(P,Q,R)=(P_1 \oplus \dotsb \oplus P_N, Q_1 \oplus \dotsb \oplus Q_M, R_1 \oplus \dotsb \oplus R_L)$ is a tensor isomorphism from $\tA$ to $\tB$?
\end{quotation}

\begin{theorem}\label{thm:pTI}
\algprobm{Partitioned Tensor Isomorphism} of $n \times m \times \ell$ 3-tensors, with $N \times M \times L$ many blocks, reduces to \TI for 3-tensors whose three side lengths are $n''' \times m''' \times \ell'''$ where
\[
\begin{array}{rllll}
n''' & = n(2^{N}  + 2^{N+M + L} ) & + m(N + 2^{M+L}) & + \ell 2^L & + 2^N-2 + M + 2^{N+M+L} + N 2^L \\
m''' & = n 2^{N+M} & + m 2^M & + \ell M & -1 + NM + 2^{N+M} + L \\
\ell''' & = n(L2^N + 2^{N+M+L}) & + m(NL + 2^{M+L}) & + \ell(2^L-1) & -1 + (2^N-1)L + ML + 2^{N+M+L} + N 2^L
\end{array}
\]

In particular, for $N,M,L=O(1)$, this is linear in the size of the original tensor.
\end{theorem}

We will frequently apply this where $N,M,L \in \{1,2\}$. 

\begin{proof}
Suppose $\tA, \tB$ are partitioned 3-tensors of size $n \times m \times \ell$, with $N \times M \times L$ many blocks, of sizes $n_I \times m_J \times \ell_K$ for $I \in [N], J \in [M], K \in [L]$. It will suffice to show how to reduce $N$ to $1$; the same procedure can then be repeated in the other two directions.

Let the frontal slices of $\tA$ be $A_1, \dotsc, A_\ell$, each of which is an $n \times m$ matrix. Let $\tilde \tA$ be the 3-tensor whose frontal slices are as follows:
\begin{itemize}
\item For $i=1,\dotsc,\ell$, 
\[
\tilde A_i = \left[\begin{array}{c;{2pt/2pt}c;{2pt/2pt}c}
A_i & 0 & 0\\ \hdashline[2pt/2pt]
0 & 0 & 0 \\ \hdashline[2pt/2pt]
0 & 0 & 0 
\end{array}\right],
\]
where the sizes and dashed lines line up with those in the remaining slices below.

\item We add $N$ additional frontal slices as follows. Define $r_I = (1+n)2^{I-1}$ for $I=1,\dotsc,N$ (slightly weaker bounds would suffice, but we choose this form for simplicity). Then for $I=1,\dotsc,N$, we add the frontal slices $\tilde A_{\ell + I} =$
\begin{equation}\label{eq:partTI}
	{\small
		\left[\begin{array}{c;{2pt/2pt}cccccc;{2pt/2pt}ccccc}
			0_{n_1 \times m} & 0_{n_1 \times n_1} & 0_{n_1 \times n_2} & \cdots & 0_{n_1 \times n_I} & \cdots & 0_{n_1 \times n_N} & 0_{n_1 \times r_1} & \cdots & 0_{n_1 \times r_I} & \cdots & 0_{n_1 \times r_N}  \\ 
			0_{n_2 \times m} & & 0_{n_2 \times n_2}  & & & & & \vdots & \ddots & & & \vdots\\ 
			\vdots & & & \ddots & & & & \vdots & & \ddots & & \vdots \\
			0_{n_I \times m} & & & & \IdMat_{n_I} & & & \vdots  & & & \ddots & \vdots\\
			\vdots & & & & & \ddots & & \vdots & & & & \vdots\\
			0_{n_N \times m} & & & & & & 0_{n_N \times n_N} & 0_{n_N \times r_1} & \cdots & 0_{n_N \times r_I} & \cdots & 0_{n_N \times r_N} \\ \hdashline[2pt/2pt]
			0_{m \times m} & 0_{m \times n_1} & & & & & & \\
			\vdots & & \ddots & & & & & \\
			\IdMat_{m} & & & & 0_{m \times n_I} & & &  \\
			\vdots & & & & & \ddots & & \\
			0_{m \times m} & & & & & & 0_{m \times n_N} & \\ \hdashline[2pt/2pt]
			0_{r_1 \times m} &  & & & & & & 0_{r_1 \times r_1} \\
			\vdots & & & & & & & & \ddots \\
			0_{r_I \times m} & & & & & & & & & \IdMat_{r_I} \\
			\vdots & & & & & & & & & & \ddots \\
			0_{r_N \times m} & & & & & & & & & & &  0_{r_N \times r_N}
		\end{array}\right]},
\end{equation}
 where blanks indicate blocks of zeros (we have included some zero blocks explicitly to help clarify their sizes). The $\IdMat_m$ in the left block-column is in the $I+1$-th block-row (rows $n + m(I-1) + 1$ to $n + mI$). 
\end{itemize}
Let $r = \sum_{I=1}^N r_I = (1+n)(2^{N}-1)$. Then $\tilde \tA$ has size $(n + mN + r) \times (m + n + r) \times (\ell + N)$. 
The number of blocks for the partition of $\tilde \tA$ is $1 \times M \times L$, and the blocks have sizes $(n + mN + r) \times \tilde m_J \times \tilde \ell_K$, where 
\[
\tilde m_J = \begin{cases} 
m_J & 1 \leq J \leq M-1 \\
m_M + n + r & J=M
\end{cases}
\qquad
\text{ and } 
\qquad
\tilde \ell_K = \begin{cases}
\ell_K & 1 \leq K \leq L-1 \\
\ell_L + N & K=L.
\end{cases} 
\]
That is, we include the new columns in the final $m$-block, and the new depth in the final $\ell$-block.

We claim that the map $\tA \mapsto \tilde \tA$ is a reduction from \algprobm{Partitioned Tensor Isomorphism} to itself, but in which the output has only one horizontal part. That is, $\tA \cong \tB$ as partitioned 3-tensors iff $\tilde \tA \cong \tilde \tB$ as partitioned 3-tensors (with only one horizontal part).

($\Rightarrow$) Suppose $\tA \cong \tB$ via $(P_1 \oplus \dotsb \oplus P_N, Q_1 \oplus \dotsb \oplus Q_M, R_1 \oplus \dotsb \oplus R_L)$. Then we claim $\tilde \tA \cong \tilde \tB$ via $(P_1 \oplus \dotsb \oplus P_{N-1} \oplus \tilde P_N, Q_1 \oplus \dotsb \oplus Q_{M-1} \oplus \tilde Q_{M}, R_1 \oplus \dotsb \oplus R_{L-1} \oplus \tilde R_{L})$, where we define 
\begin{align*}
\tilde P_{N} & = P_N \oplus (Q^{-t})^{\oplus M} \oplus \IdMat_{r} \\
\tilde Q_{M} & = Q_M \oplus (P_1^{-1} \oplus \dotsb \oplus P_N^{-1}) \oplus \IdMat_{r} \\
\tilde R_{M} & = R_M \oplus \IdMat_N.
\end{align*}
The verification is straightforward, but we sketch it here. In the front-upper-right $n \times m \times \ell$ block, we get the same isomorphism as we had from $\tA$ to $\tB$. For the additional $N$ frontal slices, the $\IdMat_N$ block of $\tilde R_M$ does not mix them, so the $Q^{-t}$'s in $\tilde P_{N}$ result in $Q^{-t} \IdMat_{m} Q^t$ in the only nonzero block in the first $m$ columns. The identity $\IdMat_r$ in the last block of $\tilde Q_{M}$ leaves the $\IdMat_{r_I}$ in the last $r$ columns unchanged. The $P_I^{-1}$ in $\tilde Q_M$ results in the only nonzero block in columns $m + 1$ through $m + n$ being $P_I \IdMat_{n_I} P_I^{-1} = \IdMat_{n_I}$. 

($\Leftarrow$) Suppose that $\tilde \tA$ and $\tilde \tB$ are block-isomorphic via $(P, Q,R)$ where $Q = Q_1 \oplus \dotsb \oplus Q_{M}$ and $R = R_1 \oplus \dotsb \oplus R_{L}$. That is, $\sum_{i'} R_{ii'} P \tilde A_{i'} Q^t = \tilde B_i$ for $i=1,\dotsc,\ell+N$. Let us write $R_L$ in block form commensurate with some of the blocks of the construction:
\[
R_L = \begin{bmatrix}
R_{11} & R_{12} \\
R_{21} & R_{22},
\end{bmatrix}
\]
where $R_{11}$ is $\ell_L \times \ell_L$ and $R_{22}$ is $N \times N$. By Lemma~\ref{lem:indiv}, $R_{12} = 0$ and $R_{22}$ is diagonal.

We now consider 
\[
\tilde A'_{\ell+I} := \sum_{i=1}^\ell (R_{21})_{I,i} \tilde A_i + \sum_{I'=1}^N (R_{22})_{II'} \tilde A_{\ell+I'}.
\]
Since $(P,Q,R)$ was an isomorphism, and the actions of $P$, $Q$, and $R$ commute with one another, we have $P \tilde A'_{\ell+I} Q^t = \tilde B_{\ell+I}$ for all $I$.

Now, let us break $P$ and $Q$ up into blocks commensurate with the blocks in the description of $\tilde A'_{\ell+I}$. We will break $P$ up into $3N \times 3N$ many blocks $P_{IJ}$, where 
\[
\text{size}(P_{II}) = \begin{cases}
n_I \times n_I & 1 \leq I \leq N \\
m \times m & N+1 \leq I \leq 2N \\
r_{I-2N} \times r_{I-2N} & 2N+1 \leq I \leq 3N.
\end{cases}
\]
Similarly, we will break $Q$ up into $(2N+1) \times (2N+1)$ many blocks $Q_{IJ}$, where
\[
\text{size}(Q_{II}) = \begin{cases}
m \times m & I = 1 \\
n_{I-1} \times n_{I-1} & 2 \leq I \leq N+1 \\
r_{I-(N+1)} \times r_{I-(N+1)} & N+2 \leq I \leq 2N+1. 
\end{cases}
\]

Suppose the upper-left $n \times m$ block of $\tilde A'_{\ell+I}$ is $X^{(I)} = \begin{bmatrix} X^{(I)}_{1} \\ \vdots \\ X^{(I)}_N \end{bmatrix} = \sum_{i=1}^\ell (R_{21})_{Ii} A_i$. Let $\alpha_I = (R_{22})_{II}$. Then define $S$ to be a block-matrix of the same size and block sizes as $P$, where $S$ looks like a giant identity matrix, except that $S_{J, N+I} = -X^{(I)}_J$ for all $J=1,\dotsc,N$ and $I=1,\dotsc,N$. Then we have $S \tilde A'_{\ell+I} = \alpha_I \tilde A_{\ell+I} = \alpha_I \tilde B_{\ell+I}$ for all $I=1,\dotsc,N$. 
Let $\tilde \vA'' = S \vA'$. Since $(P,Q,1)$ is an isomorphism from $\tilde \tA'$ to $\tilde \tB$, we have that $(PS^{-1}, Q, 1)$ is an isomorphism from $\tilde \tA''$ to $\tilde \tB$. 

As the last $N$ frontal slices of $\tilde \tA''$ and $\tilde \tB$ already agree up to the scalar multiples $\alpha_I$, let us see what constraints this puts on $P' = PS^{-1}$ and $Q$. Write $P'$ in block form commensurate with the blocks of $P$, viz. $P'_{IJ}$. To make things visually clearer, we permute the block rows and block-columns to put the nonzero ones in the upper-left, and we use $*$'s to indicate a value we don't care about at the moment:
\begin{equation} \label{eq:PAQ}
\begin{bmatrix}
P'_{N+I,N+I} & P'_{N+I,I} & P'_{N+I, 2N+I} & * \\
P'_{I,N+I} & P'_{I,I} & P'_{I,2N+I} & * \\
P'_{2N+I,N+I} & P'_{2N+I, I} & P'_{2N+I, 2N+I} & * \\
P'_{*,N+I} & P'_{*, I} & P'_{*, 2N+I} & *
\end{bmatrix}
\begin{bmatrix}
\alpha_I \IdMat_m & 0 & 0 & 0 \\
0 & \alpha_I \IdMat_{n_I} & 0 & 0\\
0 & 0 & \alpha_I \IdMat_{r_I} & 0 \\
0 & 0 & 0 & 0
\end{bmatrix}
(\pi Q)^t
=
\begin{bmatrix}
\IdMat_m & 0 & 0 & 0 \\
0 & \IdMat_{n_I} & 0 & 0 \\
0 & 0 & \IdMat_{r_I} & 0 \\
0 & 0 & 0 & 0
\end{bmatrix}
\end{equation}
where $\pi$ is a permutation matrix that brings the three block-rows of $Q$ corresponding to the nonzero blocks in $\tilde A_{\ell+I}$ into the first three block-rows.

To ascertain what we want from the above, we use the following lemma:
%

\begin{lemma} \label{lem:block}
Suppose $A,A'$ are full-rank $a \times a$ matrices. The set of pairs $(P,Q) \in \GL_{a+b}(\F) \times \GL_{a+c}(\F)$  such that
\[
P \begin{bmatrix} A & 0 \\ 0 & 0_{b \times c} \end{bmatrix} Q^t =  \begin{bmatrix} A '& 0 \\ 0 & 0_{b \times c} \end{bmatrix}
\]
are precisely those of the form
\[
P = \begin{bmatrix} P_{11} & P_{12} \\ 0& P_{22} \end{bmatrix}
\qquad
Q = \begin{bmatrix} Q_{11} & Q_{12} \\ 0 & Q_{22} \end{bmatrix}
\]
where $P_{11} A Q_{11}^t = A'$, $P_{22} \in \GL_{b}(\F)$, $Q_{22} \in \GL_{c}(\F)$, and $P_{12}, Q_{12}$ are arbitrary.
\end{lemma}

\begin{proof}
\[
\begin{bmatrix} P_{11} & P_{12} \\ P_{21} & P_{22} \end{bmatrix}
 \begin{bmatrix} A & 0 \\ 0 & 0_{b \times c} \end{bmatrix}
\begin{bmatrix} Q_{11}^t & Q_{21}^t \\ Q_{12}^t & Q_{22}^t \end{bmatrix}
=
\begin{bmatrix}
P_{11} A Q_{11}^t & P_{11} A Q_{21}^t \\
P_{21} A Q_{11}^t & P_{21} A Q_{21}^t
\end{bmatrix}
=
 \begin{bmatrix} A' & 0 \\ 0 & 0_{b \times c} \end{bmatrix}
\]
From the (1,1) entry, we find that $P_{11} A Q_{11}^t=A'$, which is full rank; hence, both $Q_{11}$ and $P_{11}$ have full rank. Next, from the fact that $AQ_{11}^t$ is full rank, by examining the (2,1) position, we find that $P_{21} = 0$. Similarly, the fact that $P_{11}A$ is full rank, by examining the (1,2) position, we find that $Q_{21}=0$. Finally, we get the invertibility of $P_{22}$ and $Q_{22}$ from the fact that $P$ and $Q$ are block-triangular invertible matrices, so each of their diagonal blocks must be invertible. Lastly, note that $P_{12}, Q_{12}, P_{22}, Q_{22}$ do not occur in the above equations, so they can otherwise be arbitrary.
\end{proof}

From Lemma~\ref{lem:block} and Equation (\ref{eq:PAQ}), we get that, for all $I=1,\dotsc,N$, $P'_{J,N+I}$, $P'_{J, I}$, and $P'_{J,2N+I}$ are zero unless $J \in \{I, N+I, 2N+I\}$. In other words, the nonzero blocks in $P'$ occur in the following pattern:
\[
P' = 
\left[\begin{array}{cccc;{2pt/2pt}cccc;{2pt/2pt}cccc}
* & & & & * & & & & * \\
& * & & & & * & & & & * \\
& & \ddots & & & & \ddots & & & & \ddots\\
& & & * & & & & * & & & & * \\ \hdashline[2pt/2pt]
* & & & & * & & & & * \\
& * & & & & * & & & & * \\
& & \ddots & & & & \ddots & & & & \ddots\\
& & & * & & & & * & & & & * \\ \hdashline[2pt/2pt]
* & & & & * & & & & * \\
& * & & & & * & & & & * \\
& & \ddots & & & & \ddots & & & & \ddots\\
& & & * & & & & * & & & & * 
\end{array}\right]
\]

Now let us consider what happens to the first $\ell$ frontal slices. First, note that left multiplication by $S^{-1}$ has no effect on the first $\ell$ matrices, since the only place it differs from the identity matrix is in blocks that will add $0$ blocks from $\tilde A''_i$ to the upper-left block. So the action of $P$ on $\tilde A'_i$ ($i=1,\dotsc,\ell$) is the same as that of $P'$. And the only blocks of $P'$ that affect those slices are $P'_{11}, \dotsc, P'_{NN}$, which are in partitioned form commensurate with the partitions of our original tensor $\tA$. 

Similarly, because the only nonzero block in $\tilde A'_i$ with $i=1,\dotsc,\ell$ is the upper-left $n \times m$, the only block of $Q$ which has any effect is the $Q_{11}$ block (which is $m \times m$). By assumption, this block is already in partitioned form, since the columns of $\tA'$ are partitioned in the same manner as those of $\tA$. Thus, these blocks of $P$ (equivalently, $P'$) and $Q$, together with $R_{11}$, give a partitioned isomorphism of $\tA$ with $\tB$, as claimed.

\paragraph{Bounding the size.} As stated above, the size of $\tA$ is $(n + mN + r) \times (m + n + r) \times (\ell + N)$ where $r = (1+n)(2^N-1)$. Here we see what happens when we apply this construction iteratively in all three directions.

We now track how the dimensions and the number of blocks change at each step. Let $r'$ and $r''$ play the role of $r$ in the subsequent steps. The first step results in the dimensions changing as follows:
\[
(n \times m \times \ell, N \times M \times L) \mapsto ((n+mN + r) \times (m+n+r) \times (\ell + N), 1 \times M \times L)
\]
For simplicity of calculation, if we call the resulting size $n' \times m' \times \ell'$, then in the second step we apply the construction with the indices shifted by 1, resulting in the size changing as follows:
\begin{align*}
n' \times m' \times \ell' & \mapsto (n' + M) \times (m' + \ell'M + r') \times (m' + \ell' + r') \\
 & = (n + mN + r + M) \times (m + n + r + (\ell + N)M + r') \times (m+n+r+\ell+N+r').
\end{align*}

Call the resulting size $n'' \times m'' \times \ell''$. Then in the final step we apply the construction with the indices shifted by 2, resulting in:
\begin{align*}
n'' \times m'' \times \ell'' \mapsto & (\ell'' + n'' + r'') \times (m'' + L) \times (\ell'' + n'' L + r'') \\
 = & (m+2n+2r+\ell+N+r' + mN + M + r'') \\
 & \times (m + n + r + (\ell + N)M + r' + L) \\
 & \times (m+n+r+\ell+N+r' + (n + mN + r + M)L + r'') 
 \end{align*}

Next we calculate $r'$ and $r''$ in terms of the original parameters: 
\begin{align*}
r' & = (1 + m') (2^M-1) \\
 & = (1 + m+n+r)(2^M-1) \\
 & = (1 + n + m + (1+n)(2^N-1))(2^M-1) \\
 & = (m + (1+n)2^N)(2^M-1). 
\end{align*}
And for $r''$ we have:
\begin{align*}
r'' & = (1+ \ell'')(2^L-1) \\
 & = (1 + m+n+r+\ell+N+r')(2^L-1) \\
 & = (1 + m+n+(1+n)(2^N-1)+\ell+N+(m + (1+n)2^N)(2^M-1))(2^L-1) \\
 & = ((1+n)2^{N+M} + \ell + N + m 2^M)(2^L-1)
\end{align*}
Finally, we plug these back in to the bounds for the total dimensions of the output tensor. The first dimension is:
\begin{align*}
& m+2n+2r+\ell+N+r' + mN + M + r'' \\
= & m+2n+2(1+n)(2^N-1)+\ell+N+(m + (1+n)2^N)(2^M-1) \\
& + mN + M + ((1+n)2^{N+M} + \ell + N + m 2^M)(2^L-1) \\
= & n(2 + 2(2^N-1) + 2^N(2^M-1) + 2^{N+M}(2^L-1)) \\
 & + m(1 + 2^M-1 + N + 2^M(2^L-1)) \\
 & + \ell(1 + 2^L-1) \\
 & + 2(2^N-1) + N + 2^N(2^M-1) + M + 2^{N+M}(2^L-1) + N (2^L-1) \\
= & n(2^{N}  + 2^{N+M + L} ) + m(N + 2^{M+L}) + \ell 2^L + 2^N-2 + M + 2^{N+M+L} + N 2^L
\end{align*}
The second dimension is:
\begin{align*}
& m + n + r + (\ell + N)M + r' + L \\
= & m + n + (1+n)(2^N-1) + \ell M + NM + (m + (1+n)2^N)(2^M-1) + L \\
= & n 2^{N+M} + m 2^M + \ell M -1 + NM + 2^{N+M} + L \\
\end{align*}
The third dimension is:
\begin{align*}
 & m+n+r+\ell+N+r' + (n + mN + r + M)L + r'' \\
 = & m + n + (1+n)(2^N-1) + N + (m + (1+n)2^N)(2^M-1) + (n + mN + (1+n)(2^N-1) + M)L \\
  & + ((1+n)2^{N+M} + \ell + N + m 2^M)(2^L-1) \\
 = & n(L2^N + 2^{N+M+L}) + m(NL + 2^{M+L}) + \ell(2^L-1) -1 + (2^N-1)L + ML + 2^{N+M+L} + N 2^L.\qedhere
  \end{align*}
\end{proof}

Now that we've seen the proof, it is interesting to look at the roles the various parts of the gadget played. The $\IdMat_{r_I}$ in the lower-right blocks enforce that the gadget slices can't be added to the original slices nor to one another. The $\IdMat_{n_i}$ in the upper blocks then enforce the block structure on $P$. The $\IdMat_m$ in the left-most part of the gadget really just serves to help with the proof, letting us use $S$ to put the slices into a (partial) normal form that is easier to reason about, but probably can be removed with a little extra work in the argument.

\paragraph{Linked-partition tensor isomorphism.} The \pTI-to-\TI reduction enables reductions to \TI from all the other four actions in Definition~\ref{def:5action}. This is achieved via the following linked-partition tensor isomorphism notion. 

Let $\tA, \tB\in\T(n\times m\times \ell, \F)$. Let $[n]=D_1\uplus\dots\uplus D_N$, $[m]=E_1\uplus\dots\uplus E_M$, and $[\ell]=F_1\uplus\dots\uplus F_L$ be three partitions, where $|D_i|=n_i$, $|E_i|=m_i$, and $|F_i|=\ell_i$. Let $I_U=[N]$, $I_V=[M]$, and $I_W=[L]$. Suppose two binary
relations $\sim$ and $\Join$ on $I_U\cup I_V\cup I_W$ satisfy the
following: (1) $\sim$ is an equivalent relation; (2) if $a\Join b$ then $a\not\sim
b$; and (3) if $a\Join b$, then $b\Join c\iff a\sim c$.

Consider $(P, Q, R)=(P_1\oplus \dots\oplus P_N, Q_1\oplus\dots\oplus Q_M, R_1\oplus\dots\oplus R_L)$, where $P_i\in\GL(n_i, \F)$, $Q_i\in\GL(m_i, \F)$, and $R_i\in\GL(\ell_i, \F)$. For convenience, we shall use $X_a$ to denote $P_a$, $Q_a$, or $R_a$ depending on
whether $a\in I_U$, $a\in I_V$, or $a\in I_W$. Briefly speaking, $a \sim
b$ denotes that the corresponding two blocks are acted covariantly, and $a\Join b$
denotes that the corresponding two blocks are acted contravariantly. So if $a\sim
b$ or $a\Join b$, then $\dim(X_a)=\dim(X_b)$.

Given binary relations $\sim$ and $\Join$, we can define that a partition-isomorphism between $\tA$ and $\tB$ is \textit{linked-partition-isomorphism} if for any $a,b\in I_U\cup I_V\cup I_W$, the following conditions for decompositions of $U,V$ and $W$ holds:
\begin{align*}
	X_a=X_b~\text{if}~a\sim b,\qquad X_a=X_b^{*}~\text{if}~a\Join b.
\end{align*}

We need the following result from \cite{FGS19}. 

\begin{theorem}[{\cite[Theorem 4.1]{FGS19}}]\label{thm:linked}
	Let $\tA, \tB\in\T(n\times m\times \ell, \F)$. Let $\tA, \tB\in\T(n\times m\times \ell, \F)$. Let $[n]=D_1\uplus\dots\uplus D_N$, $[m]=E_1\uplus\dots\uplus E_M$, and $[\ell]=F_1\uplus\dots\uplus F_L$ be three partitions, where $|D_i|=n_i$, $|E_i|=m_i$, and $|F_i|=\ell_i$. Let $I_U=[N]$, $I_V=[M]$, and $I_W=[L]$. Suppose two binary
	relations $\sim$ and $\Join$ on $I_U\cup I_V\cup I_W$ satisfy the above. 
	
	Then there exists a polynomial-time computable function $f$ that takes $\tA$ and $\tB$, and outputs $f(\tA)$ and $f(\tB)$ in $\T(n'\times m'\times \ell', \F)$, and partitions of $[n']$, $[m']$ and $[\ell']$ with $N'$, $M'$, and $L'$ parts respectively, such that $\tA$ and $\tB$ are linked-partition isomorphic if and only if $\tA'$ and $\tB'$ are partition isomorphic. Furthermore, $n', m', \ell'=O(n+m+\ell+(N+M+L)^2)$, and $N', M', L'=O((N+M+L)^2)$.
\end{theorem}

Note that the four actions, namely $U\otimes U\otimes W$, $U\otimes U^*\otimes W$, $U\otimes U\otimes U$, and $U\otimes U\otimes U^*$ from Definition~\ref{def:5action}, can be formulated as linked partition tensor isomorphism. Indeed, the partitions here are trivial: $I_U=[1]$, $I_V=[1']$, and $I_W=[1'']$. (We use $1'$ and $1''$ to distinguish the parts.) The links are $1\sim 1'$ for $U\otimes U\otimes W$, $1\Join 1'$ for $U\otimes U^*\otimes W$, $1\sim 1'\sim 1''$ for $U\otimes U\otimes U$, and $1\Join 1''$ and $1'\Join 1''$ for $U\otimes U\otimes U^*$. In particular, $N=M=L=O(1)$, so Theorem~\ref{thm:linked} gives us the following.
\begin{corollary}\label{cor:linked}
	For $i\in\{2, 3, 4, 5\}$, let $\tA$ and $\tB$ be two $3$-way arrays compatible to\footnote{That is, for the 5th action in Definition~\ref{def:5action}, the sizes of $\tA$ and $\tB$ need to be $n\times n\times n$.} the $i$th action defined in Definition~\ref{def:5action} of length $L$. Then there exists a polynomial-time computable function $f$ that takes $\tA$ and $\tB$ and outputs $3$-way arrays $f(\tA)$ and $f(\tB)$, such that (1) the lengths of $f(\tA)$ and $f(\tB)$ are upper bounded by $O(L)$, and (2) $\tA$ and $\tB$ are in the same orbit under the $i$th action if and only if $f(\tA)$ and $f(\tB)$ are isomorphic as $3$-tensors.
\end{corollary}

\section{Equivalence of five actions}\label{sec:equivalence}

To prove Theorem~\ref{thm:equivalence}, we first note that for the claims for associative and Lie algebra isomorphism, we can use Theorem~\ref{thm:previous_linear} (6,7). Then by Theorems~\ref{thm:previous_quadratic} (1-3) and~\ref{thm:previous_linear}, we need to present linear-length reductions for the following. 
\begin{enumerate}
	\item \TI to \AltMatSpIsomlong. This was solved by Theorem~\ref{thm:AltMatSpIsom}. Note that one can replace alternating matrix spaces with symmetric matrix spaces, and the same result would still hold.
	\item \pTI to \TI. This was solved by Theorem~\ref{thm:pTI}. Note that by Corollary~\ref{cor:linked}, this implies that all four actions reduce to \TI with linear-length blow-ups. 
	\item \AltMatSpIsomlong to \algprobm{Alternating Trilinear Form Equivalence}. This is what we aim to achieve in the following.
\end{enumerate}

For the case of alternating trilinear forms, we recall the exterior product from multilinear algebra. For two vectors $u,v \in \F^n$, their exterior product, denoted $u \wedge v$, is the alternating matrix $u \otimes v - v \otimes u$. In particular, we have $u \wedge v = -v \wedge u$. The exterior product is also multilinear, in that $(u + u') \wedge v = u \wedge v + u' \wedge v$. In particular, $g \in \GL(n,\F)$ acts on the exterior product $u \wedge v$ by $g \cdot (u \wedge v) = (gu) \wedge (gv)$. For three vectors $u,v,w$, we may iterate this construction as follows. The exterior product $u \wedge v \wedge w$ is the alternating tensor $u \otimes v \otimes w + w \otimes u \otimes v + v \otimes w \otimes u - u \otimes w \otimes v - w \otimes v \otimes u - v \otimes u \otimes w$. This alternating product has the properties that: 
\begin{itemize}
\item It is linear in each variable separately, e.g. $(u + u') \wedge v \wedge w = u \wedge v \wedge w + u' \wedge v \wedge w$ (and similarly for $v$ and $w$). In particular, $g \in \GL_n(\F)$ acts on such products by $g \cdot (u \wedge v \wedge w) = (gu) \wedge (gv) \wedge (gw)$.

\item It is alternating, that is, if we permute the three vectors $u,v,w$, the tensor gets multiplied by $\pm 1$ according to the sign of the permutation, viz. $u \wedge v \wedge w = w \wedge u \wedge v = v \wedge w \wedge u = -v \wedge u \wedge = - w \wedge v \wedge u = - u \wedge w \wedge v$.

\item A basis for the space of alternating trilinear forms is given by $\{e_i \wedge e_j \wedge e_k : i < j < k\}$.
\end{itemize}

\begin{theorem} \label{thm:altform}
\AltMatSpIsomlong 
reduces to \algprobm{Alternating Trilinear Form Equivalence} with linear blow-up. In particular, for $m$-dimensional spaces of $n \times n$ matrices, the output is an alternating trilinear form on at most $2(n+m) + 4$ variables.

The same holds for \algprobm{Symmetric Matrix Space Isometry} reducing to \algprobm{Symmetric Trilinear Form Equivalence}, with the same bounds. \algprobm{Symmetric Trilinear Form Equivalence} may be replaced by \algprobm{Cubic Form Equivalence}.
\end{theorem}

\subsection{Proof of Theorem~\ref{thm:altform}}

\begin{proof}
We give the proof for the alternating case, to make it clear how to get the signs right. The proof for the symmetric case is the same \emph{mutatis mutandis}, turning all the negative signs into positive ones (using the symmetric product $u \odot v := u \otimes v + v \otimes u$ and $u \odot v \odot w = u \otimes v \otimes w + w \otimes u \otimes v + v \otimes w \otimes u + u \otimes w \otimes v + w \otimes v \otimes u + v \otimes u \otimes w$ in place of the exterior product). The final reduction from \algprobm{Symmetric Trilinear Form Equivalence} to \algprobm{Cubic Form Equivalence} is the usual one in which a symmetric trilinear form $\sum_{i,j,k} T_{ijk} x_i x_j x_k$ gets mapped to the cubic form $\sum_{\text{distinct } i,j,k} T_{ijk} x_i x_j x_k + \sum_{\text{distinct } i,j} T_{i,i,j} x_i^2 x_j + \sum_{i} T_{i,i,i} x_i^3$.\footnote{We note, because one of us---ahem, the first author---keeps getting confused, that one only needs to avoid characteristics 2 and 3 for the reduction in the other direction, from cubic forms to symmetric trilinear forms, because it requires dividing by 6. The reduction in the direction we use here is fine over arbitrary fields.}

Given an alternating matrix space $\cA$ of $n \times n$ matrices and dimension $m$, with corresponding 3-way array $\tA_{ijk}$ ($i,j \in [n], k \in [m]$), we first build an alternating trilinear form as follows. Let $x_1, \dotsc, x_n, z_1, \dotsc, z_m$ be new variables. The alternating matrix space $\cA$ can be faithfully represented using the exterior product as $\sum_k \sum_{i < j} \tA_{ijk} (x_i \wedge x_j) \otimes z_k$. Then we first consider the alternating trilinear form $\sum_k \sum_{i < j} \tA_{ijk} x_i \wedge x_j \wedge z_k$. (Note that, because we used distinct variables $z_k$ in the third coordinate, we don't need to restrict, e.g., $k > j$ in the summation.) 

Next we add a gadget similar to the ones used in the results above. Let $r > (n+m+1)/2$. Introduce additional new variables $u_1, \dotsc, u_n, U_1,\dotsc,U_{2r}, g$ (``$g$'' for ``gadget''). Let $\tilde A$ denote the alternating trilinear form 
\begin{align}
\tilde A := & \sum_k \sum_{i < j} \tA_{ijk} x_i \wedge x_j \wedge z_k + \sum_{i \in [n]} x_i \wedge u_i \wedge g + \sum_{i \in [r]} U_i \wedge U_{r+i} \wedge g.  \label{eq:altform:gadget}
\end{align}

Given another alternating matrix space $\cB$, we similar build the corresponding alternating trilinear form $\tilde B$. We claim that the map $(\cA, \cB) \mapsto (\tilde A,\tilde B)$ is a reduction from \AltMatSpIsomlong to \algprobm{Alternating Trilinear Form Equivalence}. If the original matrix spaces were $m$-dimensional spaces of $n \times n$ matrices, and we let $r = \lfloor \frac{n+m+1}{2} \rfloor+1$, then the alternating forms are on  $n + m + 2r + 1 = 2(n + m) + 4$ variables, which is a linear size increase.

For brevity, let $N = n + m + 2r + 1$ be the number of variables in our alternating forms. Let $\tilde A_i$ (for $i \in [N]$) be the matrix which is the $i$-th frontal slice of the 3-way array corresponding to $\tilde A$. Since we'll need them several times throughout the proof, let us examine what these frontal slices look like now. We list our variables (slices, rows, columns) in the order:
\[
x_1, \dotsc, x_n, z_1, \dotsc, z_m, u_1, \dotsc, u_n, U_1, \dotsc, U_{2r}, g.
\]
We'll use $\vec{E}_i$ to denote the $i$-th standard column vector (1 in position $i$, 0 in the other positions). And we'll write our matrices in block form with blocks of sizes $n, m, n, 2r, 1$, corresponding to the ordering of the variables above.

\begin{itemize}
\item The first $n$ frontal slices, indexed by the $x_i$, are of the form
\begin{equation} \label{eq:slice_x}
\tilde A_i = \begin{bmatrix} 
0_n & * & & & & 0_{n \times 1} \\
-*^t & 0_m & & & & 0_{m \times 1} \\
 & & 0_n & & & \vec{E}_i \\
 & & & 0_{2r} \\
0_{1 \times n} & 0_{1 \times m} & -\vec{E}_i^t & & & 0 \\
\end{bmatrix}
\end{equation}
The $*$'s here will be the horizontal slices of the original tensor $\tA$, but we won't be very concerned with their exact values.

\item The next $m$ frontal slices, index by the $z_k$, hold the frontal slices of original tensor $\tA$ in the upper-left corner and outside that are zero:
\begin{equation} \label{eq:slice_z}
\tilde A_{n+i} = \begin{bmatrix} 
A_{i}  \\
 & 0_{m+n+2r+1} 
\end{bmatrix}
\end{equation}

\item The next $n$ frontal slices, corresponding to the $u_i$, are of the form
\begin{equation} \label{eq:slice_u}
\tilde A_{n+m+i} = \begin{bmatrix} 
0_n & & & & -\vec{E}_i  \\
& 0_m \\
& & 0_n \\
& & & 0_{2r} \\
\vec{E}_i^t & & & & 0 
\end{bmatrix}
\end{equation}
(Notice the minus sign, because these come from the term $x_i \wedge u_i \wedge g$, but to have $u_i$ in the third index it is the tensor $x_i \otimes g \otimes u_i$, which has a negative sign in the expansion of $x_i \wedge u_i \wedge g$.)

\item The next $r$ frontal slices are
\begin{equation} \label{eq:slice_U1}
\tilde A_{n+m+n+i} 
= 
\begin{bmatrix} 
0_n & & & & 0_{n \times 1} \\
& 0_m \\
& & 0_n \\
& & & 0_{2r} &  \vec{E}_{r+i} \\
0_{1 \times n} & & & -\vec{E}^t_{r+i} &  0 \\
\end{bmatrix}
\end{equation}

\item The next $r$ frontal slices are
\begin{equation} \label{eq:slice_U2}
\tilde A_{n+m+n+r+i}
=
\begin{bmatrix} 
0_n & & & & 0_{n \times 1} \\
& 0_m \\
& & 0_n \\
& & & 0_{2r} &  -\vec{E}_{i} \\
0_{1 \times n} & & & \vec{E}^t_{i} &  0 \\
\end{bmatrix}
\end{equation}

\item The last slice, corresponding to $g$, is essentially the same as the gadget we used in Theorem~\ref{thm:AltMatSpIsom} (with extra zeros appended):
\begin{equation} \label{eq:slice_g1}
\tilde A_{N} = 
\begin{bmatrix}
0_n& 0_{n \times m}  & \IdMat_n  & 0_{n \times 2r} & 0_{n \times 1} \\
0_{m \times n} & 0_m & \\
 -\IdMat_n &  & 0_n \\
 & & & \begin{bmatrix} 0 & \IdMat_{r} \\ -\IdMat_{r} & 0 \end{bmatrix} \\
 & & & & 0 \\
\end{bmatrix},
\end{equation}
The $\IdMat_n$ in the first block-row (and $-\IdMat_n$ in the first block-column) corresponds to the terms $x_i \wedge u_i \wedge g$ in the construction, and the alternating matrix of rank $2r$ on the block-diagonal corresponds to the terms $U_i \wedge U_{r+i} \wedge g$.
\end{itemize}

Now we proceed to prove the claim that the above construction gives a reduction.

($\Rightarrow$) Suppose $\cA$ and $\cB$ are isometric alternating matrix spaces. Let $A_1, \dotsc, A_m$ be a basis of the matrix space $\cA$, and similarly $B_1, \dotsc, B_m$ for $\cB$. Suppose that $(P,R)$ sends one to the other, that is, $B_i = \sum_{i' \in [m]} R_{ii'} PA_{i'} P^T$ for all $i \in [m]$. We claim that the block-diagonal matrix $P \oplus R \oplus P^{-t} \oplus \IdMat_{2r+1}$ is an equivalence from the alternating trilinear form $\tilde A$ to $\tilde B$. The verification is routine, and is similar to that in the first part of Theorem~\ref{thm:AltMatSpIsom}.

($\Leftarrow$) Suppose that $P \in \GL_n$ is an equivalence from the alternating trilinear form $\tilde A$ to $\tilde B$. First, write $P = \begin{bmatrix} P^{11} & P^{12} \\ P^{21} & P^{22} \end{bmatrix}$ where $P^{11}$ is a square matrix of side length $N-1$ and $P^{22}$ is $1 \times 1$. In order to apply Lemma~\ref{lem:indiv}, note that in the description of the slices above, all the slices other than $\tilde A_N$ (corresponding to $g$) are supported on the union of the $x$-rows, $z$-rows, and $g$-rows, and the $x$-columns, $z$-columns, and $g$-columns. Thus Lemma~\ref{lem:indiv} applies with the $n$ and $m$ of the lemma both being our $n+m+1$, and we conclude that $P^{12} = 0$. Note that here, Lemma~\ref{lem:indiv} is telling us about $P$ based on its action in the third direction; although it also acts simultaneously on the rows and columns, Lemma~\ref{lem:indiv} is agnostic to that action, and lets us conclude that $P^{12}=0$ anyway.

Now, since $P^{12}=0$, none of the terms involving $g$ contribute to the terms of the form $x_i \wedge x_j \wedge z_k$ after the application of $P$. Thus we have
\begin{eqnarray*}
\tilde B = P \cdot \tilde A & = & \sum_{k \in [m]} \sum_{1 \leq i < j \leq n} \tA_{ijk} (P x_i) \wedge (P x_j) \wedge (P z_k)  + (\text{terms involving $g$}) \\
 & = & \sum_{k \in [m]} \sum_{1 \leq i < j \leq n} \tA_{ijk} (P^{11} x_i) \wedge (P^{11} x_j) \wedge (P^{11} z_k) + (\text{terms involving $g$}).
\end{eqnarray*}
Equating the terms that don't involve $g$, we then get
\begin{equation} \label{eq:P11}
\sum_k \sum_{i < j} \tB_{ijk} x_i \wedge x_j \wedge z_k = \sum_k \sum_{i < j} \tA_{ijk} (P^{11} x_i) \wedge (P^{11} x_j) \wedge (P^{11} z_k).
\end{equation}

\begin{observation} \label{obs:altforms}
Notation as above. If for all $i$ and $k$, the coefficient of $z_k$ in $P^{11} x_i$ is zero, then $\cA$ and $\cB$ are isometric (alternating) matrix spaces.
\end{observation}

\begin{proof}
As all the terms on the LHS of (\ref{eq:P11}) contain two $x$'s and one $z$---let us call such terms ``good''---any non-good terms on the RHS must cancel. By assumption, since $P^{11} x_i$ and $P^{11} x_j$ do not include any $z_k$, the only way to get good terms on the RHS is from the $x$'s appearing in $P^{11} x_i$, the $x$'s appearing in $P^{11} x_j$, and the $z$'s appearing in $P^{11} z_k$. Let $Q$ denote the linear map sending $\text{Span}\{x_1, \dotsc, x_n\}$ to itself that is induced by the action of $P^{11}$ and modding out by all the $u,U$ variables. Let $R$ be the analogous linear map sending $\text{Span}\{z_1, \dotsc, z_m\}$ to itself. Then from \label{eq:P11} we get $(Q \tA Q^t)^R = \tB$, and thus $\cA$ and $\cB$ are isometric matrix spaces.
\end{proof}

Showing that $P^{11} x_i$ does not include any $z_k$, as in the hypothesis of the preceding observation, occupies the remainder of the proof.


Mimicking the proof of Theorem~\ref{thm:AltMatSpIsom}, let us temporarily separate out the action of $P$ by taking linear combinations of the frontal slices from the isometry action of $P$ on each slice. Define \[
\tilde A_i' := \sum_{i' \in [N]} P_{ii'} \tilde A_i.
\]
Note that the corresponding 3-way array $\tilde \tA'$ may no longer be alternating in all three directions, but separating things this way helps facilitate our analysis (and it will still have its frontal slices being alternating matrices, that is, it corresponds to an alternating matrix space rather than to an alternating trilinear form).

Further break up $P^{21}$  and $P^{22}$ as:
\begin{align*}
(P^{21})^t \text{ } & \text{ size} \\
\begin{bmatrix}
\vec{a} & \vec{\alpha} \\
\vec{b} & \vec{\beta} \\
\vec{c} & \vec{\gamma} \\
\vec{d} & \vec{\delta} \\
\vec{e} & \vec{\epsilon} \\
\end{bmatrix}
& 
\begin{bmatrix} n \\ m \\ n \\ r \\ r  \end{bmatrix}  
&  P^{22} = \begin{bmatrix} \sigma \end{bmatrix} 
\end{align*}
where we've indicated the number of components of each vector in the array beside $P^{21}$. 

Consider first the frontal slice of $\tilde \tA'$ corresponding to $g$ (at index $N$). Using the above notation and the description of the frontal slices above, we get that this has the form
\[
\tilde A'_{N} = \begin{bmatrix}
Z & X & \sigma \IdMat_n  && & -\vec{c} \\
-X^t & 0_m & & & &   0_{m \times 1} \\
-\sigma \IdMat_n & & 0_n  & &   & \vec{a} \\
 & & & 0_r & \sigma \IdMat_r &  -\vec{e} \\
 & & & -\sigma \IdMat_r & 0_r  &  \vec{d} \\
\vec{c}^t & 0_{m \times 1}  & -\vec{a}^t & \vec{e}^t & -\vec{d}^t  & 0 \\
\end{bmatrix}
\]


Next, let $(\tilde A'_{N})^{11}$ denote the upper-left $(N-1) \times (N-1)$ sub-matrix of $\tilde A'_{N}$. Let $J_r = \begin{bmatrix} 0_r & \IdMat_r \\ -\IdMat_r & 0_r \end{bmatrix}$. From the fact that $P \tilde A'_{N} P^t = \tilde B_{N}$, and since $P^{12} = 0$, we find that:
\begin{equation} \label{eq:P11matrix}
P^{11}
\begin{bmatrix}
Z & X & \sigma \IdMat_n & &  \\
-X^t & 0_m & & &  \\
-\sigma \IdMat_n & & 0_n & &  \\
 & & & \sigma J_r & \\ 
\end{bmatrix} 
(P^{11})^t = 
\begin{bmatrix}
0_n & 0 & \IdMat_n & &  \\
0 & 0_m & & &   \\
 -\IdMat_n & & 0_n & &   \\
 & & & J_r &    \\
\end{bmatrix} 
\end{equation}

\ynote{Modified here.}
Our goal is then to clear $Z$ and $X$ in $P^{11}$. To clear $Z$, we can use the techniques from the proof of Theorem~\ref{thm:block-diagonal-AMSI}. For example, when the characteristic of $\F$ is not $2$, we can use 
\[
P_0 := \begin{bmatrix}
\IdMat_n & & Z/(2\sigma) \\
& \IdMat_m & -X^t / \sigma \\
& & \IdMat_n \\
& & & \IdMat_{2r + m + 2s}
\end{bmatrix},
\]
and let $P' = P^{11} P_0^{-1}$. When the characteristic of $\F$ is $2$, we replace $Z/(2\sigma)$ in the above by $Z_u/\sigma$ where $Z_u$ is defined as setting the lower-triangular part of $Z$ to be $0$. 

 Then (\ref{eq:P11matrix}) becomes
\[
P' P_0 \begin{bmatrix}
Z & X & \sigma \IdMat_n & &  \\
-X^t & 0_m & & &  \\
-\sigma \IdMat_n & & 0_n & &  \\
 & & & \sigma J_r & \\ 
\end{bmatrix} 
P_0^t (P')^t
=
\begin{bmatrix}
0_n & 0 & \IdMat_n & &  \\
0 & 0_m & & &   \\
 -\IdMat_n & & 0_n & &   \\
 & & & J_r &    \\
\end{bmatrix} 
\]
Simplifying the left-hand side, and using the fact that $Z$ is a skew-symmetric matrix,\footnote{This is the only place we need the assumption of characteristic not 2. This can be avoided at the cost of adding another $m$ variables $v_1, \dotsc, v_m$, and adding the sum $\sum_{k \in [m]} z_k \wedge v_k \wedge g$ to the cosntruction of $\tilde A$. For then what we find is that if $\vec{b}$ is nonzero, it would add nonzero terms of the form $v_k \wedge g$ to the $g$ slice of $\tilde \tA^P$, which would increase its rank above the correct rank of $2n+2r$. Thus we also would get $Z=0$. As doing that would have added additional variables, indices, and made every matrix take up more space, we chose the more economical route.} we get
\[
P'  \begin{bmatrix}
0 & 0 & \sigma \IdMat_n &   \\
0 & 0_m & &   \\
-\sigma \IdMat_n & & 0_n &   \\
 & & & \sigma J_r  \\ 
\end{bmatrix} 
 (P')^t
=
\begin{bmatrix}
0_n & 0 & \IdMat_n &   \\
0 & 0_m & &    \\
 -\IdMat_n & & 0_n  &   \\
 & & & J_r     \\
\end{bmatrix},
\]
that is, $P'$ is in the projective stabilizer of the gadget slice $\tilde A_{N}$. 

Further break up $P'$ in blocks commensurate with the above, as follows:
\[
P' = \begin{bmatrix}
P^{xx} & P^{xz} & P^{xu} & P^{xU}  \\
P^{zx} & P^{zz} & P^{zu} & P^{zU} \\
\vdots & \vdots & \vdots & \vdots \\
P^{Ux} & P^{Uz} & P^{Uu} & P^{UU} \\
\end{bmatrix},
\]
where $P^{ab}$ denotes the part of the matrix that sends variables of the form $b_*$ to variables of the form $a_{\bullet}$. Then by Lemma~\ref{lem:block2}, we get
\[
\begin{bmatrix}
P^{zx} & P^{zu} & P^{zU}  \\
\end{bmatrix} = 0.
\]
In particular, $P^{zx} = 0$. 

Now, as right multiplication by $P_0$ only affects the third block-column (those of the form $P^{*u}$), we find that $P^{11}$ also has the property that for all $i$, $P^{11} x_i$ does not include any $z_j$. Finally, by Observation~\ref{obs:altforms}, we conclude that $\cA$ and $\cB$ are isometric alternating matrix spaces.
\end{proof}


\section{Applications}\label{sec:application}

\subsection{Some immediate applications of Theorem~\ref{thm:equivalence}}

The following theorem covers \CubicForm, as stated in Theorem~\ref{thm:GI}, and adds \AlgIso. 
\begin{theorem}\label{thm:GI2}
	If \GIlong is in $\cc{P}$, then we have
		\begin{enumerate}
				\item for a prime $p>3$, \CubicFormlong over $\F_p$ in $n$-variables can be solved in $p^{O(n)}$ time.
				\item \algprobm{Algebra Isomorphism} (associative or Lie) over $\F_p^n$ can be solved in $p^{O(n)}$ time.
			\end{enumerate}
			
	Furthermore, if \GIlong can be solved in time $2^{O((\log n)^c)}$, then the preceding problems can be solved in time $p^{O(n^c)}$.
\end{theorem}

Note that the ``furthermore'' beats the trivial bound for these problems whenever $c < 2$, and beats the current state of the art for them $c < 1.5$ (combining our reductions with the theorem below); $c=1$ corresponds to \GIlong $\in \cc{P}$, and the current state of the art for \GIlong \cite{Bab16} has $c=3$ \cite{helfgott}.


We will prove Theorem~\ref{thm:GI2} and Corollary~\ref{cor:algiso} as an immediate application of Theorem~\ref{thm:equivalence}. Before that we need the following result.

\begin{theorem}[{\cite[Theorem 1.2]{Sun23}}]\label{thm:sun2}
	\AltMatSpIsomlong for $\Lambda(n, q)^m$ can be solved in time $q^{O((n+m)^{1.8}\cdot \log q)}$.
\end{theorem}

\begin{proof}[Proofs of Theorem~\ref{thm:GI} and Corollary~\ref{cor:algiso}.] By Theorem~\ref{thm:equivalence}, \STFElong and \AlgIsolong over $\F_p^n$ reduce to \AltMatSpIsomlong for $\Lambda(n', p)^{m'}$ where $n'+m'=O(n)$. Corollary~\ref{cor:algiso} follows immediately. 
	
	We then recall that (1) \AMSI for $\Lambda(n', p)$ can be reduced to \GpI for $\varpc(p, 2, p^{O(n'+m')}=p^{O(n)})$ \cite{GQ17}, and (2) \GpI reduces to \GI in polynomial time \cite{KST93}. Therefore the resulting graphs are of order $p^{O(n)}$, and Theorem~\ref{thm:GI} then follows.
\end{proof}

\subsection{Nilpotency class reduction}\label{subsec:nil_reduction}

In this subsection we prove Corollary~\ref{cor:nil_reduction}, stated more formally as Theorem~\ref{thm:nil_reduction}.

Our goal is to reduce testing isomorphism of $p$-groups of class $c$ and exponent $p$, $c<p$, to that of testing isomorphism of $p$-groups of class $2$ and exponent $p$, when groups are given by Cayley tables. Such a reduction for matrix groups over finite fields was given in \cite{GQ2}. This reduction is achieved by first reducing to Lie algebra isomorphism based on the Lazard's correspondence, and then reducing Lie algebra isomorphism to $p$-group isomorphism of class $2$ and exponent $p$. To extend that reduction to the Cayley table model requires a different approach of realising the Lazard correspondence.

First we recall the Lazard correspondence. Let $\mathbf{Grp}_{p,n,c}$ denote the set of finite groups of order $p^n$ and 
class $c$, and let $\mathbf{Lie}_{p,n,c}$ denote the set of Lie rings of order $p^n$ and class $c$. 

\begin{theorem}[{Lazard Correspondence for finite groups \cite{lazard}, see, e.\,g., \cite[Ch.~9 
		\& 10]{khukhro} or \cite[Ch.~6]{naik}}] \label{thm:lazard}
	For any prime $p$ and any $1 \leq c < p$, there are functions $\logbf \colon 
	\mathbf{Grp}_{p,n,c} \leftrightarrow \mathbf{Lie}_{p,n,c} : \expbf$ such that (1) 
	$\logbf$ and $\expbf$ are inverses of one another, (2) two groups $G,H \in 
	\mathbf{Grp}_{p,n,c}$ are isomorphic if and only if $\logbf(G)$ and $\logbf(H)$ 
	are isomorphic, and (3) if $G$ has exponent $p$, then the 
	underlying abelian group of $\logbf(G)$ has exponent $p$. More strongly, $\logbf$ 
	is an isomorphism of categories $\mathbf{Grp}_{p,n,c} \cong \mathbf{Lie}_{p,n,c}$. 
\end{theorem}

We recall the following result from \cite{GQ2} about computing the Lazard correspondence in the matrix group model.
\begin{proposition}[{\cite{GQ2}, adapted from \cite[Exercise~10.6]{khukhro}}] \label{prop:lazard_matrices}
	Let $G \leq \GL(n,\F_p)$ be a $p$-group of exponent $p$. Then $\logbf(G)$ (from the Lazard 
	correspondence) can be realized as a finite Lie subalgebra of $n\times n$ 
	matrices over $\F_p$. Given a generating set for $G$ of $m$ matrices, a 
	generating set for $\logbf(G)$ can be constructed in $\poly(n, m, \log p)$ time.
\end{proposition}

The following proposition shows how to compute the Lazard correspondence in the Cayley table model. Recall that for an algebra $A$ and a linear basis $b_1, \dots, b_N$ of $A$, the structure constants $\tA=(a_{i, j, k})\in\T(N\times N\times N, \F)$ are those field elements satisfying $b_i\cdot b_j=\sum_k a_{i,j,k}b_k$. 
\begin{proposition}\label{prop:lazard_log}
Let $G$ be a finite $p$-group of class $c$ and exponent $p$ of order $N=p^n$, $c<p$, given by its Cayley table. Then an array of structure constants of $\logbf(G)$ can be computed in time $\poly(N)=p^{O(n)}$.
\end{proposition}
\begin{proof}
Construct the right regular representation of $G$ on itself to obtain a matrix group $\hat{G}\leq\GL(N, p)$ of order $N$. It is clear that $G\cong \hat{G}$. Apply Proposition~\ref{prop:lazard_matrices} to $\hat{G}$ to obtain a generating set of $\logbf(G)$, which is a Lie algebra in $\M(N, p)$. A linear basis $B$ of $\logbf(G)$ can be computed by a linear algebraic version of breadth-first search, and since $|\logbf(G)|=N=p^n$, $|B|=n$. We can then compute an array of structure constants by expanding the product of each pair of elements in $B$. It can be verified that all these steps can be carried out in time $\poly(N, \log p)=\poly(N)$, concluding the proof. 
\end{proof}

\begin{theorem}[Nilpotency class reduction for groups given by Cayley tables]\label{thm:nil_reduction}
There exists a polynomial-time computable function $f$ satisfying the following. Let $G$ and $H$ be two $p$-groups of class $c$ and exponent $p$, $c<p$, of order $N=p^n$. Then $f$ takes $G$ and $H$ and outputs $f(G)$ and $f(H)$, which are the Cayley tables of two $p$-groups of class $2$ and exponent $p$ of order $N^{O(1)}$, such that $G\cong H$ if and only if $f(G)\cong f(H)$.
\end{theorem}
\begin{proof}
First use Proposition~\ref{prop:lazard_log} to compute the structure constants of $\logbf(G)$ and $\logbf(H)$ in time $\poly(N)$. Note that by Theorem~\ref{thm:lazard}, $G\cong H \iff \logbf(G)\cong\logbf(H)$, and the sizes of $\logbf(G)$ and $\logbf(H)$ are upper bounded by $O(n)$. Use Theorem~\ref{thm:equivalence} (more specifically, the reduction from $j=5$ to $j=2$, skew-symmetric frontal slices) to compute skew-symmetric bilinear maps $\phi(G)$ and $\phi(H)$, such that $\logbf(G)\cong\logbf(H) \iff \phi(G)\cong \phi(H)$. Note that the size of $\phi(G)$ and $\phi(H)$ are upper bounded by $O(n)$. Finally turn $\phi(G)$ and $\phi(H)$ to $p$-groups of class $2$ and exponent $p$, $\tilde{G}$ and $\tilde{H}$, such that $\phi(G)\cong\phi(H)\iff\tilde{G}\cong\tilde{H}$, and $|\tilde{G}|=|\tilde{H}|=p^{O(n)}$. This concludes the proof.
\end{proof}

\subsection{Search-to-decision reductions}\label{subsec:search}

Given an oracle that decides whether two tensors over $\F_q$ are isomorphic, we would like to design an algorithm that computes an isomorphism if the input tensors are indeed isomorphic. This was done for \AltMatSpIsomlong in \cite{GQ2}, using a monomial-restriction gadget with a quadratic blow-up in the dimensions. Here, with the help of the Automorphism Gadget Lemma (\ref{lem:aut}), we present a reduction with only linear blow-up in the dimensions.

Let us briefly outline the procedure. Let $\tA, \tB\in \T(n\times m\times \ell, q)$, where $n\leq m\leq \ell$. Suppose the decision oracle tells us that $\tA$ and $\tB$ are isomorphic. Our goal is to apply a sequence of transformations $R_i\in\GL_n$, $i\in[n]$, such that the following happens. Let $\tA_k$ be the $3$-way array obtained by applying $R_kR_{k-1}\dots R_1$ on the left. Then $\tA_k$ and $\tB$ are isomorphic via $(L_k, M_k, N_k)\in\GL_n\times\GL_m\times\GL_\ell$, such that $L_k=\begin{bmatrix}
	U_k & 0 \\
	V_k & W_k
\end{bmatrix}$ where $U_k$ is a $k\times k$ monomial matrix. Therefore, after the $n$th step, $\tA_n$ and $\tB$ are isomorphic via $(L_n, M_n, N_n)$ where $L_n$ is an $n\times n$ monomial matrix. We can then enumerate all monomial matrices. For each monomial matrix $L'$, test whether $L'$ induces an isomorphism from $\tA_n$ to $\tB$ using the module isomorphism algorithm \cite{IQ19}, and compute one if there is. This gives us a search-to-decision reduction in time $n!\cdot q^{O(n)}\cdot \poly(n,m,\ell, \log q)$. To reduce that running time to $2^n\cdot q^{O(n)}\cdot \poly(n, m, \ell, \log q)=q^{O(n+m+\ell)}$, we need to resort to a dynamic programming algorithm as \cite[Section 4.3]{GQ2}.

Now we get to the details.

First, we need the following tensor. For $1\leq k\leq n$, let $\Gamma_k\in\T(n\times n\times n, \F)$ be the $3$-way array where $\Gamma_k(i,i,i)=1$ for $i\in[k]$, and $0$ for other entries. Then it is easy to verify that $(A, B, C)\in \aut(\Gamma_k)\leq\GL_n\times\GL_n\times\GL_n$ if and only if there exist $\pi\in\S_k$, $d_1, d_2, d_3\in \diag(k, \F)$,  $d_1 d_2 d_3 = I_n$, 
such that 
\begin{equation}\label{eq:partial_id}
	A=\begin{bmatrix}
		P_\pi d_1 & 0 \\
		* & * 
	\end{bmatrix}, 
B=\begin{bmatrix}
	P_\pi d_2 & 0 \\
	* & * 
\end{bmatrix}, 
C=\begin{bmatrix}P_\pi d_3 & 0 \\
	* & * 
\end{bmatrix}.
\end{equation}

Second, we need the following lemma. 
\begin{lemma}\label{lem:search-to-decision}
Let $k\in\{0, 1, \dots, k-1\}$. Suppose $\tA_k$ and $\tB$ in $\T(n\times m\times \ell, q)$ are isomorphic via $(L_k, M_k, N_k)\in\GL_n\times\GL_m\times\GL_\ell$, where $L_k=\begin{bmatrix}
	U_k & 0 \\
	V_k & W_k
\end{bmatrix}$, and $U_k$ is a $k\times k$ monomial matrix. Then we can compute $R_{k+1}\in \GL_n$ in time $q^{n}\cdot \poly(n, m, \ell, \log q)$, by querying the decision oracle with $3$-way arrays of size $O(n+m+\ell)$, such that $\tA_{k+1}:=R_{k+1}\tA_k$ and $\tB$ are isomorphic via $(L_{k+1}, M_{k+1}, N_{k+1})$, where $L_k=\begin{bmatrix}
U_{k+1} & 0 \\
V_{k+1} & W_{k+1}
\end{bmatrix}$, and $U_{k+1}$ is a $(k+1)\times (k+1)$ monomial matrix.
\end{lemma}
\begin{proof}
Let $r\in\F_q^n$ such that $r^t$ is the $(k+1)$th row of $L_k$. Let $R_{k+1}\in\GL_n$ such that $e_i^tR_{k+1}=e_i^t$ for $i\in[k]$, and $e_{k+1}^tR_{k+1}=r^t$. Then $L_{k+1}:=L_kR_{k+1}^{-1}$ is of the form $\begin{bmatrix}
U_k & 0 & 0 \\
0 & 1 & 0 \\
V_k' & w_k & W_k'
\end{bmatrix}$. Set $\tA_{k+1}:=R_{k+1}\tA_k$. Then $\tA_{k+1}$ and $\tB$ are isomorphic via $(L_{k+1}, M_k, N_k)$.

To compute the desired $R_{k+1}$, we can enumerate $r$ such that $e_1, \dots, e_k, r$ are linearly independent. Then we take some $v_{k+2}, \dots, v_n$ such that $e_1, \dots, e_k, r, v_{k+2}, \dots, v_n$ are linearly independent. Let $R'$ be the matrix $\begin{bmatrix}
	e_1 & \dots & e_k & r & v_{k+2} & \dots & v_n
\end{bmatrix}$. Let $\tA'=R'\tA_k$. We then apply the Automorphism Gadget Lemma (\ref{lem:aut}) with $\Gamma_k$ to $\tA'$ and $\tB$, and send the resulting $3$-way arrays to the decision oracle. If the decision oracle returns yes, then we have guessed the correct $r$, so we can set $R_{k+1}=R'$ and $\tA_{k+1}=\tA'$ as the output. The existence of such $r$ is ensured by the analysis in the paragraph above. 

Note that the above procedure runs in time $q^n\cdot \poly(n, m, \ell, \log q)$, as the main cost is to enumerate $r\in \F_q^n$. The $3$-way arrays to the decision oracle are of size upper bounded by $O(n+m+\ell)$, because of the  Automorphism Gadget Lemma (\ref{lem:aut}). This concludes the proof. 
\end{proof}

Lemma~\ref{lem:search-to-decision} allows us to realise the procedure outlined at the beginning of this subsection. To reduce the cost from $n!\cdot q^{O(n)}\cdot \poly(n, m, \ell, \log q)$ to $2^n\cdot q^{O(n)}\cdot \poly(n, m, \ell, \log q)$, we need the following result. 

\begin{proposition}\label{prop:monomial}
Suppose $\tA$ and $\tB$ in $\T(n\times m\times \ell, q)$ are isomorphic via $(L, M, N)\in\Mon_n\times\GL_m\times\GL_\ell$. Then there exists a $2^n\cdot q^{O(n)}\cdot \poly(n, m, \ell, \log q)$-time algorithm that computes an isomorphism from $\tA$ to $\tB$.
\end{proposition}
A similar result for \AltMatSpIsomlong was proved in \cite[Proposition 11]{GQ2}, by combining the dynamic programming scheme of Luks for isomorphism problems \cite{Luks99} with a generalised linear code equivalence problem. As a proof of Proposition~\ref{prop:monomial} can be achieved essentially the same as the proof of \cite[Proposition 11]{GQ2}, we omit it here. 

We can now conclude this subsection by a proof of the search-to-decision reduction in Theorem~\ref{thm:search-counting}.
\begin{proof}[{Proof of search-to-decision for Theorem~\ref{thm:search-counting}}]
Let $G$ and $H$ be two $p$-groups of class $2$ and exponent $p$ given by their Cayley tables. First test if their commutator subgroups and commutator quotients are isomorphic or not. If not reject. If yes, suppose $[G,G]\cong \Z_p^m$, and $G/[G,G]\cong\Z_p^n$. Note that $|G|=p^{n+m}$. Then compute their commutator brackets $\phi_G:G/[G,G]\times G/[G,G]\to[G,G]$, and $\phi_H$ similarly. Use the reduction from \AltMatSpIsomlong to \TIlong with linear blow-up in lengths, to get $\tA_G$ and $\tA_H$ whose lengths are $O(n+m)$, such that $\phi_G$ and $\phi_H$ are isomorphic (as bilinear maps) if and only if $\tA_G$ and $\tA_H$ are isomorphic. The reduction also ensures that an isomorphism from $\phi_G$ to $\phi_H$ can be computed from an isomorphism from $\tA_G$ to $\tA_H$. We can then use the decision oracle to test if $\tA_G$ and $\tA_H$ are isomorphic, and if so, compute one isomorphism. This would yield an isomorphism from $\phi_G$ to $\phi_H$, and thus an isomorphism from $G$ to $H$. 
\end{proof}

\subsection{Counting-to-decision reductions}\label{subsec:counting}

For the purpose of counting-to-decision reduction, we also need a partial diagonal restriction gadget as follows. For $1\leq k\leq n$, let $\Gamma_k\in \T(n\times k\times \ell, \F)$, $\ell=O(k)$, be the $3$-way array where the upper $k\times k\times \ell$ part is a $3$-way array from Section~\ref{subsec:diagonal} satisfying the diagonal restriction property. 
Then it is easy to verify that $(A, B, C)\in \aut(\Gamma_k)\leq\GL_n\times\GL_k\times\GL_\ell$ if and only if $A=\begin{bmatrix}
A' & 0 \\
* & * & 
\end{bmatrix}$, where $A'\in\GL_k$ and $A', B', C'$ are diagonal and form an automorphism of the upper $k\times k\times \ell$ part of $\Gamma_k$. 

We shall refer to $\Gamma_k$ together with Lemma~\ref{lem:aut} as the partial diagonal restriction gadget. We shall apply this gadget to \AMSI directly, by installing it in two directions as in the proof of Theorem~\ref{thm:AltMatSpIsom}. The rest of the argument follows that of \cite[Proposition 17]{GQ2}, and here we give a brief outline.

The basic strategy of counting-to-decision reductions goes back to that for \GIlong \cite{Mat79}, which uses a chain of subgroups so the automorphism group order can be achieved by successively computing the number of cosets. This strategy was adapted to \AltMatSpIsomlong in \cite{GQ2} as follows.

Let $\cA\leq\Lambda(n, q)$. We wish to compute the order of $\aut(\cA):=\{T\in\GL(n, q)\mid T^t\cA T=\cA\}$. For $i\in[n]$, let 
$A_i=\{T\in A \mid 
\forall 1\leq j\leq i, T(e_i)=\lambda_i e_i, \lambda_i\neq 0\in 
\F_q\}$. Observe that 
$A_n=A\cap \diag(n, q)$, and its order can be computed in 
time $q^{O(n)}$ by 
brute-force, namely enumerating all invertible diagonal 
matrices. Set $A_0=A$. We 
then have the tower of subgroups $A_0\geq A_1\geq \dots \geq 
A_n$. 

To compute the order of $A_0$, it is enough to compute 
$[A_k:A_{k+1}]$. Note that for 
$T, T'\in A_k$, $TA_{k+1}=T'A_{k+1}$ as left cosets in $A_k$ if 
and only if 
$T(e_{k+1})=\lambda 
T'(e_{k+1})$ for some $\lambda\neq 0\in\F_q$. So 
$[A_k:A_{k+1}]$ equals the size of the orbit of $e_{k+1}$ under $A_k$ in the 
projective space. 
Let $v\in 
\F_q^n$. To test whether $v$ is in the orbit of $e_{k+1}$ under 
$A_k$ in the 
projective space, we 
transform $\cA$ by $P^t\cdot P$, where $P\in\GL(n, q)$ sends 
$e_{k+1}$ to $v$ and 
$e_j$ to $e_j$ for $j\neq k+1$, to get $\cA'$. We then add the partial
diagonal 
restriction gadget, described at the beginning of this subsection, to the first $k+1$ lateral slices and the 
first $k+1$ 
horizontal slices of $\cA$ and $\cA'$, to obtain 
$\tilde \cA$ and $\tilde \cA'$ respectively. Then feed $\cA$ 
and $\cA'$ to the 
decision oracle. 
By the functionality of the diagonal restriction gadget, $v$ is 
in the orbit of 
$e_{k+1}$ in the projective space if and only if $\tilde \cA$ 
and $\tilde \cA'$ 
are isometric. Enumerating $v\in \F_q^n$ up to scalar multiples 
gives us the size 
of the orbit of 
$e_{k+1}$ under $A_k$ in the projective space. This finishes 
the description of the algorithm, concluding the proof of the counting-to-decision reduction for Theorem~\ref{thm:search-counting}.

\appendix

\bibliographystyle{alphaurl}
\bibliography{references}

\newcommand{\etalchar}[1]{$^{#1}$}
\begin{thebibliography}{LQW{\etalchar{+}}23}

\bibitem[AD17]{AllenderDas}
Eric Allender and Bireswar Das.
\newblock Zero knowledge and circuit minimization.
\newblock {\em Inf. Comput.}, 256:2--8, 2017.
\newblock \href {https://doi.org/10.1016/j.ic.2017.04.004}
  {\path{doi:10.1016/j.ic.2017.04.004}}.

\bibitem[AS06]{AS06}
Manindra Agrawal and Nitin Saxena.
\newblock Equivalence of {$\mathbb{F}$}-algebras and cubic forms.
\newblock In {\em {STACS} 2006, 23rd Annual Symposium on Theoretical Aspects of
  Computer Science, Proceedings}, pages 115--126, 2006.
\newblock \href {https://doi.org/10.1007/11672142_8}
  {\path{doi:10.1007/11672142_8}}.

\bibitem[AT05]{AT05}
Vikraman Arvind and Jacobo Tor{\'{a}}n.
\newblock Isomorphism testing: Perspective and open problems.
\newblock {\em Bulletin of the {EATCS}}, 86:66--84, 2005.

\bibitem[Bab16]{Bab16}
L{\'{a}}szl{\'{o}} Babai.
\newblock Graph isomorphism in quasipolynomial time [extended abstract].
\newblock In Daniel Wichs and Yishay Mansour, editors, {\em Proceedings of the
  48th Annual {ACM} {SIGACT} Symposium on Theory of Computing, {STOC} 2016,
  Cambridge, MA, USA, June 18-21, 2016}, pages 684--697. {ACM}, 2016.
\newblock \href {https://doi.org/10.1145/2897518.2897542}
  {\path{doi:10.1145/2897518.2897542}}.

\bibitem[Bae38]{Bae38}
Reinhold Baer.
\newblock Groups with abelian central quotient group.
\newblock {\em Transactions of the American Mathematical Society},
  44(3):357--386, 1938.

\bibitem[BCGQ11]{BCGQ11}
L{\'{a}}szl{\'{o}} Babai, Paolo Codenotti, Joshua~A. Grochow, and Youming Qiao.
\newblock Code equivalence and group isomorphism.
\newblock In {\em Proceedings of the Twenty-Second Annual {ACM-SIAM} Symposium
  on Discrete Algorithms, {SODA} 2011, San Francisco, California, USA, January
  23-25, 2011}, pages 1395--1408, 2011.

\bibitem[BG94]{BG94}
Mihir Bellare and Shafi Goldwasser.
\newblock The complexity of decision versus search.
\newblock {\em {SIAM} J. Comput.}, 23(1):97--119, 1994.
\newblock \href {https://doi.org/10.1137/S0097539792228289}
  {\path{doi:10.1137/S0097539792228289}}.

\bibitem[BI11]{BI}
Peter B{\"{u}}rgisser and Christian Ikenmeyer.
\newblock Geometric complexity theory and tensor rank.
\newblock In Lance Fortnow and Salil~P. Vadhan, editors, {\em Proceedings of
  the 43rd {ACM} Symposium on Theory of Computing, {STOC} 2011, San Jose, CA,
  USA, 6-8 June 2011}, pages 509--518. {ACM}, 2011.
\newblock \href {https://doi.org/10.1145/1993636.1993704}
  {\path{doi:10.1145/1993636.1993704}}.

\bibitem[BMW17]{BMW}
Peter~A. Brooksbank, Joshua Maglione, and James~B. Wilson.
\newblock A fast isomorphism test for groups whose {Lie} algebra has genus 2.
\newblock {\em J. Algebra}, 473:545--590, 2017.
\newblock \href {https://doi.org/A fast isomorphism test for groups whose {Lie}
  algebra has genus 2} {\path{doi:A fast isomorphism test for groups whose
  {Lie} algebra has genus 2}}.

\bibitem[Bor12]{Bor12}
Armand Borel.
\newblock {\em Linear algebraic groups}, volume 126.
\newblock Springer Science \& Business Media, 2012.

\bibitem[Bou11]{Bou11}
Charles Bouillaguet.
\newblock {\em Etudes d’hypotheses algorithmiques et attaques de primitives
  cryptographiques}.
\newblock PhD thesis, Universit{\'e} Paris-Diderot--{\'E}cole Normale
  Sup{\'e}rieure, 2011.

\bibitem[BPR{\etalchar{+}}00]{BPR+00}
Charles~H. Bennett, Sandu Popescu, Daniel Rohrlich, John~A. Smolin, and
  Ashish~V. Thapliyal.
\newblock Exact and asymptotic measures of multipartite pure-state
  entanglement.
\newblock {\em Physical Review A}, 63(1):012307, 2000.
\newblock \href {https://doi.org/10.1103/PhysRevA.63.012307}
  {\path{doi:10.1103/PhysRevA.63.012307}}.

\bibitem[BW12]{BW12}
Peter~A. Brooksbank and James~B. Wilson.
\newblock Computing isometry groups of {Hermitian} maps.
\newblock {\em Trans. Amer. Math. Soc.}, 364:1975--1996, 2012.
\newblock \href {https://doi.org/10.1090/S0002-9947-2011-05388-2}
  {\path{doi:10.1090/S0002-9947-2011-05388-2}}.

\bibitem[BW15]{BW15}
Peter~A Brooksbank and James~B Wilson.
\newblock The module isomorphism problem reconsidered.
\newblock {\em Journal of Algebra}, 421:541--559, 2015.

\bibitem[CGQ{\etalchar{+}}24]{TI3}
Zhili Chen, Joshua~A. Grochow, Youming Qiao, Gang Tang, and Chuanqi Zhang.
\newblock On the complexity of isomorphism problems for tensors, groups, and
  polynomials {III:} actions by classical groups.
\newblock In Venkatesan Guruswami, editor, {\em 15th Innovations in Theoretical
  Computer Science Conference, {ITCS} 2024, January 30 to February 2, 2024,
  Berkeley, CA, {USA}}, volume 287 of {\em LIPIcs}, pages 31:1--31:23. Schloss
  Dagstuhl - Leibniz-Zentrum f{\"{u}}r Informatik, 2024.
\newblock URL: \url{https://doi.org/10.4230/LIPIcs.ITCS.2024.31}, \href
  {https://doi.org/10.4230/LIPICS.ITCS.2024.31}
  {\path{doi:10.4230/LIPICS.ITCS.2024.31}}.

\bibitem[CH03]{CH03}
John Cannon and Derek~F. Holt.
\newblock Automorphism group computation and isomorphism testing in finite
  groups.
\newblock {\em J. Symb. Comput.}, 35(3):241--267, 2003.
\newblock \href {https://doi.org/10.1016/S0747-7171(02)00133-5}
  {\path{doi:10.1016/S0747-7171(02)00133-5}}.

\bibitem[DG00]{Graaf}
Willem~A De~Graaf.
\newblock {\em Lie algebras: theory and algorithms}.
\newblock Elsevier, 2000.

\bibitem[DW21]{DW21}
Heiko Dietrich and James~B. Wilson.
\newblock Group isomorphism is nearly-linear time for most orders.
\newblock In {\em 62nd {IEEE} Annual Symposium on Foundations of Computer
  Science, {FOCS} 2021, Denver, CO, USA, February 7-10, 2022}, pages 457--467.
  {IEEE}, 2021.
\newblock \href {https://doi.org/10.1109/FOCS52979.2021.00053}
  {\path{doi:10.1109/FOCS52979.2021.00053}}.

\bibitem[FGS19]{FGS19}
Vyacheslav Futorny, Joshua~A. Grochow, and Vladimir~V. Sergeichuk.
\newblock Wildness for tensors.
\newblock {\em Lin. Algebra Appl.}, 566:212--244, 2019.
\newblock \href {https://doi.org/10.1016/j.laa.2018.12.022}
  {\path{doi:10.1016/j.laa.2018.12.022}}.

\bibitem[FN70]{FN70}
V.~Felsch and J.~Neub\"user.
\newblock On a programme for the determination of the automorphism group of a
  finite group.
\newblock In Pergamon J.~Leech, editor, {\em Computational Problems in Abstract
  Algebra (Proceedings of a Conference on Computational Problems in Algebra,
  Oxford, 1967)}, pages 59--60, Oxford, 1970.

\bibitem[GQ17]{GQ17}
Joshua~A. Grochow and Youming Qiao.
\newblock Algorithms for group isomorphism via group extensions and cohomology.
\newblock {\em SIAM J. Comput.}, 46(4):1153--1216, 2017.
\newblock Preliminary version in IEEE Conference on Computational Complexity
  (CCC) 2014 (DOI:10.1109/CCC.2014.19). Also available as
  \arXiv{1309.1776}{[cs.DS]} and ECCC Technical Report TR13-123.
\newblock \href {https://doi.org/10.1137/15M1009767}
  {\path{doi:10.1137/15M1009767}}.

\bibitem[GQ21]{GQ2}
Joshua~A. Grochow and Youming Qiao.
\newblock On {$p$}-group isomorphism: search-to-decision, counting-to-decision,
  and nilpotency class reductions via tensors.
\newblock In {\em 36th {C}omputational {C}omplexity {C}onference}, volume 200
  of {\em LIPIcs. Leibniz Int. Proc. Inform.}, pages Art. No. 16, 38. Schloss
  Dagstuhl. Leibniz-Zent. Inform., Wadern, 2021.
\newblock Journal version to appear in ACM Transactions on Computation Theory.
\newblock \href {https://doi.org/10.4230/LIPIcs.CCC.2021.16}
  {\path{doi:10.4230/LIPIcs.CCC.2021.16}}.

\bibitem[GQ23]{GQ1}
Joshua~A. Grochow and Youming Qiao.
\newblock On the complexity of isomorphism problems for tensors, groups, and
  polynomials {I:} tensor isomorphism-completeness.
\newblock {\em {SIAM} J. Comput.}, 52(2):568--617, 2023.
\newblock Extended abstract appeared in ITCS '21.
\newblock URL: \url{https://doi.org/10.1137/21m1441110}, \href
  {https://doi.org/10.1137/21M1441110} {\path{doi:10.1137/21M1441110}}.

\bibitem[GQT22]{GQT}
Joshua~A. Grochow, Youming Qiao, and Gang Tang.
\newblock Average-case algorithms for testing isomorphism of polynomials,
  algebras, and multilinear forms.
\newblock {\em J. Groups Complex. Cryptol.}, 14(1):[Paper No. 9431], 21, 2022.
\newblock Extended abstract appeared in STACS '21.
\newblock \href {https://doi.org/10.46298/jgcc.2022.14.1.9431}
  {\path{doi:10.46298/jgcc.2022.14.1.9431}}.

\bibitem[GR16]{GR16}
Fran{\c{c}}ois~Le Gall and David~J. Rosenbaum.
\newblock On the group and color isomorphism problems.
\newblock {\em CoRR}, abs/1609.08253, 2016.
\newblock URL: \url{http://arxiv.org/abs/1609.08253}, \href
  {http://arxiv.org/abs/1609.08253} {\path{arXiv:1609.08253}}.

\bibitem[Gro12]{GrochowLie}
Joshua~A. Grochow.
\newblock Matrix {Lie} algebra isomorphism.
\newblock In {\em IEEE Conference on Computational Complexity (CCC12)}, pages
  203--213, 2012.
\newblock Also available as arXiv:1112.2012 [cs.CC] and ECCC Technical Report
  TR11-168.
\newblock \href {https://doi.org/10.1109/CCC.2012.34}
  {\path{doi:10.1109/CCC.2012.34}}.

\bibitem[Hel19]{helfgott}
Harald~Andr\'{e}s Helfgott.
\newblock Isomorphismes de graphes en temps quasi-polynomial [d'apr\`es {B}abai
  et {L}uks, {W}eisfeiler-{L}eman,\dots ].
\newblock Number 407, pages Exp. No. 1125, 135--182. 2019.
\newblock S\'{e}minaire Bourbaki. Vol. 2016/2017. Expos\'{e}s 1120--1135.
\newblock \href {https://doi.org/10.24033/ast} {\path{doi:10.24033/ast}}.

\bibitem[IQ19]{IQ19}
G\'abor Ivanyos and Youming Qiao.
\newblock Algorithms based on *-algebras, and their applications to isomorphism
  of polynomials with one secret, group isomorphism, and polynomial identity
  testing.
\newblock {\em SIAM Journal on Computing}, 48(3):926--963, 2019.
\newblock \href {https://doi.org/10.1137/18M1165682}
  {\path{doi:10.1137/18M1165682}}.

\bibitem[JQSY19]{JQSY19}
Zhengfeng Ji, Youming Qiao, Fang Song, and Aaram Yun.
\newblock General linear group action on tensors: {A} candidate for
  post-quantum cryptography.
\newblock In {\em Theory of Cryptography - 17th International Conference, {TCC}
  2019, Nuremberg, Germany, December 1-5, 2019, Proceedings, Part {I}}, pages
  251--281, 2019.
\newblock \href {https://doi.org/10.1007/978-3-030-36030-6\_11}
  {\path{doi:10.1007/978-3-030-36030-6\_11}}.

\bibitem[Kay11]{Kay11}
Neeraj Kayal.
\newblock Efficient algorithms for some special cases of the polynomial
  equivalence problem.
\newblock In Dana Randall, editor, {\em Proceedings of the Twenty-Second Annual
  {ACM-SIAM} Symposium on Discrete Algorithms, {SODA} 2011, San Francisco,
  California, USA, January 23-25, 2011}, pages 1409--1421. {SIAM}, 2011.
\newblock \href {https://doi.org/10.1137/1.9781611973082.108}
  {\path{doi:10.1137/1.9781611973082.108}}.

\bibitem[Kay12]{Kay12}
Neeraj Kayal.
\newblock Affine projections of polynomials: extended abstract.
\newblock In Howard~J. Karloff and Toniann Pitassi, editors, {\em Proceedings
  of the 44th Symposium on Theory of Computing Conference, {STOC} 2012, New
  York, NY, USA, May 19 - 22, 2012}, pages 643--662. {ACM}, 2012.
\newblock \href {https://doi.org/10.1145/2213977.2214036}
  {\path{doi:10.1145/2213977.2214036}}.

\bibitem[Ker06]{Ker06}
Adalbert Kerber.
\newblock {\em Representations of Permutation Groups I: Representations of
  Wreath Products and Applications to the Representation Theory of Symmetric
  and Alternating Groups}, volume 240.
\newblock Springer, 2006.

\bibitem[Khu98]{khukhro}
E.~I. Khukhro.
\newblock {\em {$p$}-automorphisms of finite {$p$}-groups}, volume 246 of {\em
  London Mathematical Society Lecture Note Series}.
\newblock Cambridge University Press, Cambridge, 1998.
\newblock \href {https://doi.org/10.1017/CBO9780511526008}
  {\path{doi:10.1017/CBO9780511526008}}.

\bibitem[KS06]{KS06}
Neeraj Kayal and Nitin Saxena.
\newblock Complexity of ring morphism problems.
\newblock {\em Computational Complexity}, 15(4):342--390, 2006.
\newblock \href {https://doi.org/10.1007/s00037-007-0219-8}
  {\path{doi:10.1007/s00037-007-0219-8}}.

\bibitem[KST93]{KST93}
Johannes K\"{o}bler, Uwe Sch\"{o}ning, and Jacobo Tor\'{a}n.
\newblock {\em The graph isomorphism problem: its structural complexity}.
\newblock Birkhauser Verlag, Basel, Switzerland, Switzerland, 1993.

\bibitem[Laz54]{lazard}
Michel Lazard.
\newblock Sur les groupes nilpotents et les anneaux de {Lie}.
\newblock {\em Ann. Sci. Ecole Norm. Sup. (3)}, 71:101--190, 1954.
\newblock \href {https://doi.org/0.24033/asens.1021}
  {\path{doi:0.24033/asens.1021}}.

\bibitem[LQ17]{LQ17}
Yinan Li and Youming Qiao.
\newblock Linear algebraic analogues of the graph isomorphism problem and the
  {Erd{\H{o}}s}--{R{\'{e}}nyi} model.
\newblock In Chris Umans, editor, {\em 58th {IEEE} Annual Symposium on
  Foundations of Computer Science, {FOCS} 2017, Berkeley, CA, USA, October
  15-17, 2017}, pages 463--474. {IEEE} Computer Society, 2017.
\newblock \href {https://doi.org/10.1109/FOCS.2017.49}
  {\path{doi:10.1109/FOCS.2017.49}}.

\bibitem[LQW{\etalchar{+}}23]{LQWWZ}
Yinan Li, Youming Qiao, Avi Wigderson, Yuval Wigderson, and Chuanqi Zhang.
\newblock Connections between graphs and matrix spaces.
\newblock {\em Israel Journal of Mathematics}, 256(2):513--580, 2023.

\bibitem[Luk82]{Luks82}
Eugene~M. Luks.
\newblock Isomorphism of graphs of bounded valence can be tested in polynomial
  time.
\newblock {\em J. Comput. Syst. Sci.}, 25(1):42--65, 1982.
\newblock \href {https://doi.org/10.1016/0022-0000(82)90009-5}
  {\path{doi:10.1016/0022-0000(82)90009-5}}.

\bibitem[Luk99]{Luks99}
Eugene~M. Luks.
\newblock Hypergraph isomorphism and structural equivalence of boolean
  functions.
\newblock In Jeffrey~Scott Vitter, Lawrence~L. Larmore, and Frank~Thomson
  Leighton, editors, {\em Proceedings of the Thirty-First Annual {ACM}
  Symposium on Theory of Computing, May 1-4, 1999, Atlanta, Georgia, {USA}},
  pages 652--658. {ACM}, 1999.
\newblock \href {https://doi.org/10.1145/301250.301427}
  {\path{doi:10.1145/301250.301427}}.

\bibitem[Mat79]{Mat79}
Rudolf Mathon.
\newblock A note on the graph isomorphism counting problem.
\newblock {\em Information Processing Letters}, 8(3):131--136, 1979.

\bibitem[Mil78]{Mil78}
Gary~L. Miller.
\newblock On the {$n^{\log n}$} isomorphism technique (a preliminary report).
\newblock In {\em STOC}, pages 51--58, New York, NY, USA, 1978. ACM.
\newblock \href {https://doi.org/http://doi.acm.org/10.1145/800133.804331}
  {\path{doi:http://doi.acm.org/10.1145/800133.804331}}.

\bibitem[MW84]{MW84}
Brendan~D. McKay and Nicholas~C. Wormald.
\newblock Automorphisms of random graphs with specified vertices.
\newblock {\em Comb.}, 4(4):325--338, 1984.
\newblock \href {https://doi.org/10.1007/BF02579144}
  {\path{doi:10.1007/BF02579144}}.

\bibitem[Nai13]{naik}
Vipul Naik.
\newblock {\em Lazard correspondence up to isoclinism}.
\newblock PhD thesis, The University of Chicago, 2013.
\newblock URL: \url{https://vipulnaik.com/thesis/}.

\bibitem[O'B94]{Obr94}
Eamonn~A O'Brien.
\newblock Isomorphism testing for $p$-groups.
\newblock {\em Journal of Symbolic Computation}, 17(2):133--147, 1994.

\bibitem[Pat96]{Pat96}
Jacques Patarin.
\newblock Hidden fields equations {(HFE)} and isomorphisms of polynomials
  {(IP):} two new families of asymmetric algorithms.
\newblock In {\em Advances in Cryptology - {EUROCRYPT} '96, International
  Conference on the Theory and Application of Cryptographic Techniques,
  Saragossa, Spain, May 12-16, 1996, Proceeding}, pages 33--48, 1996.
\newblock \href {https://doi.org/10.1007/3-540-68339-9_4}
  {\path{doi:10.1007/3-540-68339-9_4}}.

\bibitem[PR97]{PR}
Erez Petrank and Ron~M. Roth.
\newblock Is code equivalence easy to decide?
\newblock {\em {IEEE} Trans. Inf. Theory}, 43(5):1602--1604, 1997.
\newblock \href {https://doi.org/10.1109/18.623157}
  {\path{doi:10.1109/18.623157}}.

\bibitem[PSS19]{PSS18}
Max Pfeffer, Anna Seigal, and Bernd Sturmfels.
\newblock {Learning paths from signature tensors}.
\newblock {\em SIAM Journal on Matrix Analysis and Applications},
  40(2):394--416, 2019.
\newblock arXiv:1809.01588.
\newblock \href {https://doi.org/10.1137/18M1212331}
  {\path{doi:10.1137/18M1212331}}.

\bibitem[Ros13]{Ros13}
David~J. Rosenbaum.
\newblock Bidirectional collision detection and faster deterministic
  isomorphism testing.
\newblock arXiv preprint \arXiv{1304.3935}{[cs.DS]}, 2013.

\bibitem[RW15]{RW15}
David~J. Rosenbaum and Fabian Wagner.
\newblock Beating the generator-enumeration bound for p-group isomorphism.
\newblock {\em Theor. Comput. Sci.}, 593:16--25, 2015.
\newblock URL: \url{https://doi.org/10.1016/j.tcs.2015.05.036}, \href
  {https://doi.org/10.1016/J.TCS.2015.05.036}
  {\path{doi:10.1016/J.TCS.2015.05.036}}.

\bibitem[Sax06]{Saxena_thesis}
Nitin Saxena.
\newblock {\em Morphisms of rings and applications to complexity}.
\newblock PhD thesis, Indian Institute of Technology Kanpur, 5 2006.

\bibitem[Sun23]{Sun23}
Xiaorui Sun.
\newblock Faster isomorphism for $p$-groups of class 2 and exponent $p$.
\newblock In Barna Saha and Rocco~A. Servedio, editors, {\em Proceedings of the
  55th Annual {ACM} Symposium on Theory of Computing, {STOC} 2023, Orlando, FL,
  USA, June 20-23, 2023}, pages 433--440. {ACM}, 2023.
\newblock \href {https://doi.org/10.1145/3564246.3585250}
  {\path{doi:10.1145/3564246.3585250}}.

\bibitem[TDJ{\etalchar{+}}22]{TangDJPQS22}
Gang Tang, Dung~Hoang Duong, Antoine Joux, Thomas Plantard, Youming Qiao, and
  Willy Susilo.
\newblock Practical post-quantum signature schemes from isomorphism problems of
  trilinear forms.
\newblock In Orr Dunkelman and Stefan Dziembowski, editors, {\em Advances in
  Cryptology - {EUROCRYPT} 2022 - 41st Annual International Conference on the
  Theory and Applications of Cryptographic Techniques, Trondheim, Norway, May
  30 - June 3, 2022, Proceedings, Part {III}}, volume 13277 of {\em Lecture
  Notes in Computer Science}, pages 582--612. Springer, 2022.
\newblock \href {https://doi.org/10.1007/978-3-031-07082-2\_21}
  {\path{doi:10.1007/978-3-031-07082-2\_21}}.

\bibitem[Val76]{Val76}
Leslie~G. Valiant.
\newblock Relative complexity of checking and evaluating.
\newblock {\em Information processing letters}, 5(1):20--23, 1976.

\bibitem[Wil09]{Wil09a}
James~B. Wilson.
\newblock Decomposing {$p$}-groups via {Jordan} algebras.
\newblock {\em J. Algebra}, 322:2642--2679, 2009.
\newblock \href {https://doi.org/10.1016/j.jalgebra.2009.07.029}
  {\path{doi:10.1016/j.jalgebra.2009.07.029}}.

\bibitem[Wor99]{Wor99}
Nicholas~C. Wormald.
\newblock Models of random regular graphs.
\newblock {\em London Mathematical Society Lecture Note Series}, pages
  239--298, 1999.

\end{thebibliography}

\end{document}